%% file: counting.tex
\documentclass[11pt]{article}
\pdfoutput=1
\input{header.tex}
\begin{document}

\def\malig{\text{mal}}
\def\benign{\text{ben}}
\def\chiIn{\chi_{\text{in}}}
\def\chiOut{\chi_{\text{out}}}
\def\zetaIn{\zeta_{\text{in}}}
\def\zetaOut{\zeta_{\text{out}}}

\newcommand{\eventFont}[1]{{\text{\sf{#1}}}}
\def\out{{\eventFont{out}}}
\def\incorrect{{\eventFont{incorrect}}}
\def\good{{\eventFont{good}}}
\def\Good{{\text{Good}}}
\def\bad{{\eventFont{bad}}}
\def\best{{\eventFont{best}}}
\def\acceptX{\eventFont{accept}}
\def\acceptZ{\eventFont{accept}}
\def\reject{\neg \eventFont{accept}}
\def\rejectX{\neg \acceptX}
\def\rejectZ{\neg \acceptZ}
\def\goodZ{\good_Z}
\def\bestZ{\best_Z}
\def\badZ{\bad_Z}
\def\goodX{\good_X}
\def\badX{\bad_X}
\def\accept{\eventFont{accept}}
\def\kBest{k_{\text{best}}}
\def\kGood{k_{\text{good}}}
\def\goodEC{\good_{\text{EC}}}
\def\badEC{\bad_{\text{EC}}}
\def\gammaMin{\gamma_\text{min}}
\def\gammaMax{\gamma_\text{max}}
\def\gammaTH{\gamma_{\text{th}}}
\def\Pevent{\mathcal{P}}
\def\fail{\eventFont{fail}}

\def\poly{\mathcal{Q}}
\def\PrC{\mathcal{P}}
\def\kMin{k_\text{min}}
\def\event{\text{event}}
\def\kMax{k_\text{max}}

\def\threshOverlap{1.32 \times 10^{-3}}
\def\threshSteane{1.24 \times 10^{-3}}

\title{Fault-tolerant ancilla preparation and noise threshold lower bounds for the $23$-qubit Golay code}
\author{%
Adam Paetznick\thanks{School of Computer Science and Institute for Quantum Computing, University of Waterloo}
\and Ben W.~Reichardt\thanks{Electrical Engineering Department, University of Southern California}
}
\date{}

\maketitle

\begin{abstract}
In fault-tolerant quantum computing schemes, the overhead is often dominated by the cost of preparing codewords reliably.  This cost generally increases quadratically with the block size of the underlying quantum error-correcting code.  In consequence, large codes that are otherwise very efficient have found limited fault-tolerance applications.  Fault-tolerant preparation circuits therefore are an important target for optimization.  

We study the Golay code, a $23$-qubit quantum error-correcting code that protects the logical qubit to a distance of seven.  In simulations, even using a na{\"i}ve ancilla preparation procedure, the Golay code is competitive with other codes both in terms of overhead and the tolerable noise threshold.  We provide two simplified circuits for fault-tolerant preparation of Golay code-encoded ancillas.  The new circuits minimize error propagation, reducing the overhead by roughly a factor of four compared to standard encoding circuits.  By adapting the malignant set counting technique to depolarizing noise, we further prove a threshold above $\threshOverlap$ noise per gate.  
\end{abstract}

\section{Introduction} \label{sec:introduction}

A main obstacle to building a quantum computer is noise.  A fault-tolerance threshold theorem implies that reliable quantum computation is possible in principle~\cite{AharonovBenOr99, Kitaev96b}.  So long as the noise is weak enough, the probability that a computation executes correctly can be made arbitrarily close to one at the cost of increased circuit complexity, i.e., overhead. Fault-tolerant quantum circuit constructions typically aim to maximize the tolerable noise rate while maintaining modest overhead.

A quantum fault-tolerance scheme generally works by encoding data into a quantum error-correcting code and alternating steps of fault-tolerant computation and error correction (\figref{fig:ft-ckt-fragment}). The error-correction step, intended for recovery from accumulated noise, is normally much more complicated than the computation step.  Therefore error correction is the dominant factor in determining the scheme's resource overhead, and is usually the major bottleneck in determining the highest tolerable noise rate or ``noise threshold."  In particular, the details of how error correction is implemented are more important than the properties of the underlying quantum error-correcting~code.  

\def\ecBox{\push{\fbox{\raisebox{0em}[1em][0.5em]{EC}}}}

\begin{figure}
\centering
\includegraphics[scale=1]{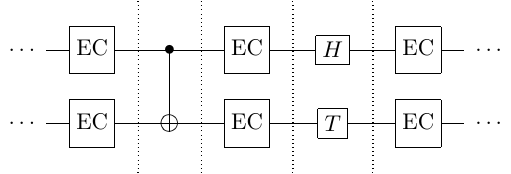}
\caption{\label{fig:ft-ckt-fragment}
The circuit fragment above shows an example of alternating rounds of fault-tolerant error-correction~(EC) and computation.  The wires represent encoded data blocks and the gate symbols (CNOT, H, T) represent encoded operations.}
\end{figure}

For example, with the nine-qubit Bacon-Shor code, a fault-tolerant logical controlled-not (CNOT) gate between two code blocks can be implemented using nine physical CNOT gates, whereas an optimized error-correction method uses $24$ physical CNOT gates~\cite{AliferisCross07subsystem}.  For larger quantum error-correcting codes, the asymmetry between computation and error correction is greater still.  With the $23$-qubit Golay code, a fault-tolerant logical CNOT gate requires only $23$ physical CNOT gates, whereas a standard error-correction method uses $2400$ physical CNOT gates.  

Larger quantum error-correcting codes, with higher distance and possibly higher rates, can still outperform smaller codes.  Separate numerical studies by Steane~\cite{Steane03} (see also~\cite{Steane04computer}) and Cross, DiVincenzo and Terhal~\cite{CrossDiVincenzoTerhal} have each compared fault-tolerance schemes based on a variety of codes.  They identify larger codes that, compared to the seven-qubit Steane code and the nine-qubit Bacon-Shor code, can tolerate higher noise rates with comparable resource requirements.  In particular, their estimates single out the Golay code as a top performer.  

We give an optimized fault-tolerant error-correction procedure for the Golay code that uses only $640$ CNOT gates, while also being highly parallelizable.  Our derivation is based on two main ideas.  First, we simplify Steane's Latin-rectangle-based scheme for preparing encoded $\ket 0$ states~\cite{Steane02}, by taking advantage of overlaps between the code's stabilizers. Second, we reduce the overall number of encoded $\ket 0$ states required for error correction by carefully tracking the exact propagation of errors. Both ideas are generally applicable to other large quantum error-correcting codes.

We then prove a lower bound on the threshold for depolarizing noise of $\threshOverlap$ noise per gate.  This result is an order of magnitude improvement over the best previous lower bound for the Golay code~\cite{AliferisCross07subsystem} based on an adversarial noise model. It is also about $5$ percent better than the lower bound due to~\cite{Aliferis2009} also for depolarizing noise, but based on a dramatically different fault-tolerance scheme. Our proof uses malignant set counting~\cite{AliferisGottesmanPreskill05}, extensively tailored for our specific error correction circuits and for depolarizing noise.  Instead of assuming adversarial noise at higher levels of code concatenation, the counting procedure keeps track of multiple types of malignant events to create a transformed stochastic noise model for each level, allowing for a more accurate analysis.

\subsection{Fault-tolerant error correction}

There are exceptions to the common paradigm, sketched in \figref{fig:ft-ckt-fragment}, of alternating computational and error correction steps.  In a scheme proposed by Knill, for example, error-correction and computation are performed simultaneously by teleporting into specially prepared ancilla states~\cite{Knill04schemes}.  Zalka~\cite{Zalka97} has suggested balancing the costs of computation and error correction by having multiple computation steps between error-correction rounds, but error propagation between code blocks makes such a scheme challenging to analyze precisely.  Surface-code-based quantum fault-tolerance schemes make a more radical change: they implement encoded gates using gradual code deformation, during which error correction occurs frequently.  However, while these schemes appear very promising~\cite{RaussendorfHarringtonGoyal05ftoneway, RaussendorfHarrington06onewaythreshold}, they have proved difficult to analyze precisely and rigorously~\cite{DennisKitaevLandahlPreskill01topological}.  

A variety of error-correction techniques have been studied, and three broad categories are so-called Shor-type~\cite{Shor96}, Steane-type~\cite{Steane97} and Knill-type~\cite{Knill04schemes} error correction.  This is only a rough categorization, and it leaves significant room for introducing new ideas and optimization within or beyond these categories; see, e.g.,~\cite{Reichardt04, DiVincenzoAliferis06slow, AliferisCross07subsystem}. The Shor-, Steane- and Knill-type error-correction schemes rely on the use of ancillary qubits to extract error information from the data blocks.  Before interacting with the data, the ancilla qubits need to be prepared in an entangled state. (Surface-code-based schemes are again an exception, and the nine-qubit Bacon-Shor code is an exception at the first level of code concatenation.)  
  
As a concrete example, consider Steane-type error-correction. Arbitrary errors can be written as linear combinations of tensor products of Pauli errors: the identity $I = \binomial{1&0}{0&1}$, a bit-flip error $X = \binomial{0&1}{1&0}$, a phase-flip error $Z = \binomial{1&0}{0&-1}$, and both bit- and phase-flip errors $Y = \binomial{0&-i}{i&0} = i X Z$.  Each tensor product can itself be decomposed as a product of a $Z$-error part---a tensor product of $I$ and $Z$ operators---and an $X$-error part---a tensor product of $I$ and $X$ operators.  Steane error-correction works by correcting $Z$ and $X$ errors separately.  First, $Z$ errors are copied from the data to an encoded $\ket{0}$ ancilla by transversal CNOT gates, i.e., CNOT gates from each qubit of the ancilla block to the corresponding qubit of the data.  $X$ errors are similarly copied onto an encoded ancilla state $\ket{+} = \frac{1}{\sqrt 2}(\ket 0 + \ket 1)$.  The ancillas are then measured in order to determine a correction.
  
Preparing ancilla states can be complicated, particularly because errors in the preparation circuit can spread through the ancilla block.  For example, a single physical fault may lead to errors on multiple ancilla qubits.  The code is limited by its distance and cannot necessarily protect against such correlated errors. As a result, the ancilla states themselves must be checked for errors.  

The complexity of verifying prepared ancillas against errors grows quickly as the code distance increases. For large codes, verification of encoded $\ket 0$ and $\ket +$ is accomplished by using additional identically prepared ``auxiliary'' ancillas.  In a manner similar to error correction, errors from the initial ancilla are copied onto the auxiliary ancillas and then the auxiliary ancillas are measured.  If measurements imply the presence of an error, then all of the ancillas are discarded and the entire process begins anew.  Otherwise, the ancilla is accepted and may be used for error correction. Of course, the auxilliary ancillas may also contain errors.  These errors can spread to the initial ancilla and invalidate the verification.  Thus the auxiliary ancillas must {also} be checked for errors by yet more ancillas.  The end result is a series of recursive verifications that involves many encoded ancillas, and dominates the overall overhead cost of error correction.

To maximize efficiency, preparation and verification circuits may be constructed using a pipeline architecture in which part of the computer is dedicated to preparing many ancillas in parallel. Even so, ancilla production constitutes the majority of the space requirement for a fault-tolerant quantum circuit. In~\cite{IsalovicNemanjaPatelKubiatowicz08}, for example, the ancilla pipeline is estimated to take up to $68$ percent of the entire circuit footprint.

For the Golay code, this recursive verification technique requires twelve encoded ancillas and at least $1177$ CNOT gates.  One such circuit is shown in \figref{fig:TwelveAncillaVerifyCkt}. Variants of this circuit have been used in previous studies of the Golay code, including in~\cite{Steane03} and~\cite{CrossDiVincenzoTerhal}.  The construction of this circuit implicitly assumes a kind of worst-case error behavior in which all possible codeword preparation circuits propagate errors in the same way.  However, DiVincenzo and Aliferis~\cite{DiVincenzoAliferis06slow} have observed that different preparation circuits exhibit different error propagation behavior, and this can be exploited.  By considering many different preparation circuits, we observe that some circuits give favorable combinations of correlated errors, and thus require fewer error verification steps, substantially reducing overhead.  In \secref{sec:ancilla-verification} we provide two circuits that require only four encoded $\ket{0}$ ancillas and as few as $297$ CNOT gates.  One of these circuits is specified by Figures~\ref{fig:FourAncillaVerifyCkt} and~\ref{fig:overlap-prep-ckt}, and \tabref{tbl:overlap-permutations}.  

\begin{figure}
\centering
\includegraphics[width=14.2cm]{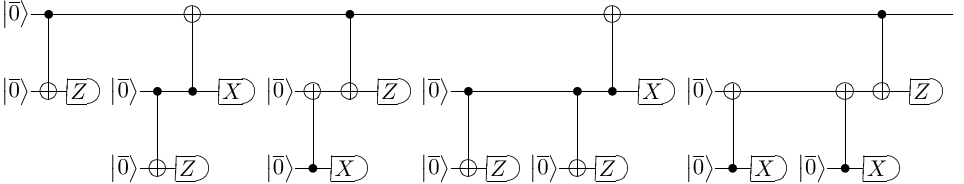}
\caption{This circuit produces a single Golay encoded $\ket 0$ state that is ready to be used in fault-tolerant error correction.  Each of the twelve encoded $\ket{0}$ ancillas, denoted $\lket{0}$, is identically prepared using the Steane Latin rectangle method (see \secref{sec:RandomSteane}).  The wires represent $23$-qubit code blocks and the indicated CNOT and measurement operations are transversal.
} \label{fig:TwelveAncillaVerifyCkt}
\end{figure}

The overhead required to prepare a fault-tolerant ancilla depends on the probability that any errors are detected.  \tabref{tbl:verify-compare} shows estimates of the probability that all of the verification stages accept along with the corresponding expected resource requirements for the different verification circuits. For depolarizing noise rates near $p = 10^{-3}$, our circuits reduce both the expected number of qubits and the expected number of CNOT gates by roughly a factor of four over the twelve-ancilla circuit.  A more detailed analysis of the acceptance probability and overhead is given in \secref{sec:Overhead}.

\subsection{Code concatenation and the noise threshold}
\label{sec:codeConcat-noiseThresh}

We consider fault-tolerant, noisy simulations constructed by compiling an ideal quantum circuit into a sequence of \emph{rectangles}, each of which contains an encoded operation and a trailing error correction~(TEC).  Following~\cite{AliferisGottesmanPreskill05}, we define a rectangle to be \emph{correct} if the action of the rectangle followed by an ideal decoder, i.e., a decoder containing no errors, is equivalent to the action of an ideal decoder followed by an ideal implementation of the corresponding gate.  If a rectangle is not correct then it is \emph{incorrect}.  In other words, a correct rectangle effectively acts as an encoded version of the intended gate.  If all rectangles are correct then a simple induction argument shows that the compiled, noisy circuit successfully simulates the original ideal circuit.  

\begin{figure}
\centering
\includegraphics[scale=1]{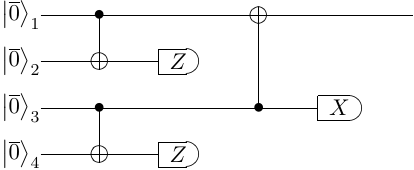}
\caption{Our simplified ancilla preparation and verification circuit uses only four encoded $\ket 0$ ancillas.  The ancillas are prepared using different encoding circuits, shown in \figref{fig:overlap-prep-ckt} and \tabref{tbl:overlap-permutations}.}  
\label{fig:FourAncillaVerifyCkt}
\end{figure}

\begin{figure}
\centering
\includegraphics[scale=1]{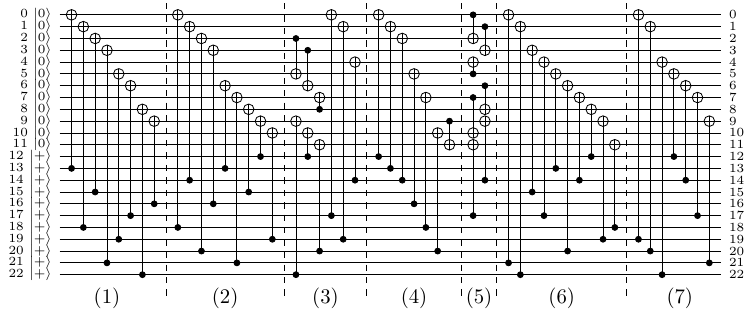}
\caption{An optimized circuit for preparing $\ket 0$ encoded in the Golay code uses $57$ CNOT gates applied in seven rounds.  Gates in the same round are applied in parallel.  The construction is detailed in \secref{sec:Overlap}.}
\label{fig:overlap-prep-ckt}
\end{figure}

\begin{table}
\centering
\small
\begin{tabular}{c|c}
\hline\hline
Ancilla & Qubit permutation \tabularnewline
\hline
$\lket{0}_2$ &  (0, 20, 13, 7, 12, 14, 1)(2, 11)(3, 19, 5, 4, 8, 22, 6, 15, 10, 16, 9, 18, 21, 17) 
\tabularnewline
$\lket{0}_3$ &  (0, 14, 6, 12, 16, 2, 11, 22, 17, 21, 9, 20, 5, 7, 3, 13, 18, 4, 15, 1, 10, 8, 19) 
\tabularnewline
$\lket{0}_4$ & (0, 12, 4, 17, 9, 6, 1)(2, 10, 18, 22, 21, 16, 13)(3, 11, 20, 15, 7, 19, 5)(8)(14) 
\\ \hline \hline
\end{tabular}
\caption{The first ancilla in \figref{fig:FourAncillaVerifyCkt} is prepared using the circuit of \figref{fig:overlap-prep-ckt}.  The other three ancillas are prepared in the same way, except with the qubits rearranged according to the above permutations.  
} \label{tbl:overlap-permutations}
\end{table}

\begin{table}
\centering 
\begin{tabular}{c|cccccccc}
\hline \hline 
Verification & $\Pr$[accept] & E[\# qubits] & min \# CNOTs & E[\# CNOTs] \tabularnewline
\hline
Steane-$12$ & $0.419 \pm 0.001$ & $5183 \pm 14.2$ & $1177$ & $1782 \pm 4.9$ \tabularnewline
Steane-$4$ & $0.648 \pm 0.002$ & $1413 \pm  3.7$ & $377$ & $497.6 \pm 1.3$ \tabularnewline
Overlap-$4$ & $0.633 \pm 0.002$ & $1399 \pm 3.8$ & $297$ & $399.4 \pm 1.1$ \tabularnewline
\hline \hline 
\end{tabular}
\caption{Estimates of the acceptance probability and overhead for the twelve-ancilla fault-tolerant ancilla preparation circuit and our two optimized circuits, at a depolarizing noise rate of $p = 10^{-3}$. The Steane-$4$ circuit is based on ancillas prepared according to~\tabref{tab:randomizedSchedules}.  Overlap-$4$ is based on ancillas prepared according to~\figref{fig:overlap-prep-ckt} and~\tabref{tbl:overlap-permutations}. The column labeled Pr[accept] gives the probability that all auxiliary ancilla measurements in the verification circuit detect no errors.  The next column, E[qubits], gives the expected number of physical qubits required to produce one verified encoded $\ket{0}$.  This is calculated recursively, by computing the expected number of qubits needed to pass each verification step.  The last two columns specify, respectively, the minimum number of CNOT gates and the expected number of CNOT gates required to produce a single verified ancilla.  
} \label{tbl:verify-compare}
\end{table}

For a fixed stochastic noise model and a fixed quantum error-correcting code, the probability that a rectangle is correct is a constant and therefore the probability that all rectangles are correct will generally be exponentially small in the number of gates in the circuit being simulated.  To achieve a constant success probability, a process known as code concatenation~\cite{Knill1996} is often used. In a concatenated fault tolerant simulation, each gate is first compiled into a rectangle, called a level-one rectangle ($1$-Rec), as described above.  Then, a level-two rectangle ($2$-Rec) is constructed by compiling each physical gate of the $1$-Rec into a rectangle.  This process is repeated as many times as desired.  The end result is a circuit composed of a hierarchy of rectangles.  

At each level $k$ of concatenation, the probability that the $k$-Rec is correct increases relative to level $k-1$ so long as the strength of the noise is below a certain value called the \emph{threshold}. The threshold is calculated by upper bounding the probability that each type of rectangle is incorrect.  In~\cite{AliferisGottesmanPreskill05} the upper bound is obtained by counting \emph{malignant} sets of locations inside an object called the \emph{extended rectangle}, or exRec, which consists of the rectangle together with its leading error correction~(LEC).  A set of locations is considered malignant if there exists some fixed combination of nontrivial Pauli errors acting on that set of locations that causes the rectangle to be incorrect.

Malignant set counting works for a broad class of noise models including so-called adversarial noise in which locations fail independently at random, but the error at each failing location is chosen by an adversary and may be correlated with errors at other failing locations.  For more restricted noise models such as depolarizing noise, however, malignant set counting is overly pessimistic.  Roughly, this is because the definition of a malignant set is independent of the underlying noise model and, therefore, malignant set counting ignores a large amount of information.  

In Sections~\ref{sec:Counting} and~\ref{sec:calc-threshold} we outline a modified malignant set counting technique that more accurately computes the threshold for depolarizing noise acting on fault-tolerant simulations constructed with our error-correction circuits.  Our counting method introduces two new ideas.  First, for computational efficiency, we break up the exRec into a hierarchy of components and count $X$ and $Z$ errors separately to keep the total number of error combinations small.  This technique allows us to analyze larger subsets of faulty locations than would otherwise be possible. Second, we construct at each level a transformed noise model, similar to depolarizing noise, by separately accounting for multiple types of malignant events.  ExRecs at each level of code concatenation behave in a self-similar manner under the transformed noise, thus admitting a straightforward threshold calculation.

\section{The Golay code}

The Golay code is a perfect CSS $[[23,1,7]]$ quantum error-correcting code (see e.g.,~\cite{Steane03}).  It has eleven $X$ and eleven $Z$ stabilizer generators.  The code is self-dual, and the $X$ and $Z$ stabilizer generators can both be given by the following eleven $23$-character strings: 
\begin{equation} \label{eq:golay-stabilizers}
\begin{tabular}{c@{\:}c@{\:}c@{\:}c@{\:}c@{\:}c@{\:}c@{\:}c@{\:}c@{\:}c@{\:}c@{\:}c@{\:}c@{\:}c@{\:}c@{\:}c@{\:}c@{\:}c@{\:}c@{\:}c@{\:}c@{\:}c@{\:}c@{\:}}
. &1 &. &. &1 &. &. &1 &1 &1 &1 &1 &. &. &. &. &. &. &. &. &. &. &1\\
1 &. &. &1 &. &. &1 &1 &1 &1 &1 &. &. &. &. &. &. &. &. &. &. &1 &.\\
. &1 &1 &. &1 &1 &1 &. &. &. &1 &1 &. &. &. &. &. &. &. &. &1 &. &.\\
1 &1 &. &1 &1 &1 &. &. &. &1 &1 &. &. &. &. &. &. &. &. &1 &. &. &.\\
1 &1 &1 &1 &. &. &. &1 &. &. &1 &1 &. &. &. &. &. &. &1 &. &. &. &.\\
1 &. &1 &. &1 &. &1 &1 &1 &. &. &1 &. &. &. &. &. &1 &. &. &. &. &.\\
. &. &. &1 &1 &1 &1 &. &1 &1 &. &1 &. &. &. &. &1 &. &. &. &. &. &.\\
. &. &1 &1 &1 &1 &. &1 &1 &. &1 &. &. &. &. &1 &. &. &. &. &. &. &.\\
. &1 &1 &1 &1 &. &1 &1 &. &1 &. &. &. &. &1 &. &. &. &. &. &. &. &.\\
1 &1 &1 &1 &. &1 &1 &. &1 &. &. &. &. &1 &. &. &. &. &. &. &. &. &.\\
1 &. &1 &. &. &1 &. &. &1 &1 &1 &1 &1 &. &. &. &. &. &. &. &. &. &.
\end{tabular}
\end{equation}
Here, the $1$s in a row either all indicate $Z$ operators or all indicate $X$ operators, and dots indicate identity operators.  
For example, the first line indicates that $I \otimes Z \otimes I \otimes I \otimes Z \otimes \cdots \otimes Z$ and $I \otimes X \otimes I \otimes I \otimes X \otimes \cdots \otimes X$ are stabilizers.  Note that each stabilizer generator has weight eight.  We index the qubits left to right, from $0$ to $22$.

The stabilizer generators partition the group of Pauli errors on $23$-qubits into $2^{24}$ cosets.  In particular, there are $2^{23-11}=2^{12}$ inequivalent $X$ errors (tensor products of $X$ and $I$) and $2^{12}$ inequivalent $Z$ errors (tensor products of $Z$ and $I$).

The Golay code is preserved by qubit permutations in a symmetry group known as the Mathieu group $M_{23}$.  
This is a four-transitive group of order $23 \cdot 22 \cdot 21 \cdot 20 \cdot 48$.  It is generated by a cyclic shift, and by the permutation 
$$
(2, 16, 9, 6, 8)(3, 12, 13, 18, 4)(7, 17, 10, 11, 22)(14, 19, 21, 20, 15)
\enspace,
$$
in cycle notation.  See~\cite[pp.~1411]{Huffman1998} and~\cite{Gang}.

\section{Ancilla preparation and verification} \label{sec:ancilla-verification}

Our error-correction circuits require multiple encoded $\ket{0}$ and $\ket{+}$ states.  The Golay code is self-dual, so encoded $\ket +$ is prepared by taking the dual of the $\ket 0$ circuit in the natural way (i.e., swapping $\ket 0$ and~$\ket +$ and reversing the direction of each CNOT gate).  There are many possible ways to encode~$\ket 0$.  Our most efficient preparation circuit is shown in \figref{fig:overlap-prep-ckt}, and prepares encoded $\ket 0$ in the Golay code with a total of $57$ CNOT gates.  This circuit provides an efficient means of preparing encoded ancillas, but it is not fault-tolerant on its own.  

We define strict fault-tolerance as follows: 

\begin{definition}
An ancilla encoded into a code with distance $d$ is \emph{strictly fault-tolerant} if for all $k \leq \lfloor d/2 \rfloor$, any error of probability order $k$ propagates to an error of weight at most $k$.  
\end{definition}
  
The circuit in \figref{fig:overlap-prep-ckt} is clearly not strictly fault-tolerant because, for example, with first order probability a faulty CNOT gate may produce an error of weight two when an $X$-error occurs on both its control qubit and its target qubit.  We call this error, and any other for which the weight of the resulting error is larger than the probability order with which it occurs, ``correlated''.  

To achieve strict fault tolerance, additional encoded ancillas are prepared.  Errors from the original ancilla are copied onto the additional ancillas which are then measured.  If no errors are detected, the ancilla is accepted and may be used for error correction on the data.  Otherwise, the ancilla is rejected, all of the prepared ancillas are discarded, and the procedure is restarted.  In \figref{fig:FourAncillaVerifyCkt}, two pairs of ancillas are prepared.  One of the ancillas from each pair is checked for $X$ errors.  If neither check detects an error, then one of the two remaining ancillas is used to check the other for $Z$ errors.  

Care must be taken in preparing the additional ancillas, however.  For example, say that two encoded ancillas are identically prepared.  Assume that a single failure occurs in the first ancilla and propagates through the preparation circuit to produce a weight three error.  Then the same single failure in the other ancilla will produce the \emph{same} weight three error.  When the error from the first ancilla is copied to the second, the two errors will cancel each other and no error will be detected.  This is a second order event that results in a weight-three error.

Therefore, we seek to prepare encoded ancillas that produce different correlated error sets.  In this section we provide two related methods for constructing fault-tolerant ancillas encoded in the Golay code using the circuit in \figref{fig:FourAncillaVerifyCkt}.  In \secref{sec:RandomSteane} we analyze the correlated errors produced by preparation circuits constructed with the standard Latin rectangle method and provide a randomized method for finding ancillas with different correlated error sets.  In \secref{sec:Overlap} we describe a new preparation circuit specific to the Golay code and again provide a randomized method for finding ancillas with different correlated error sets.

\subsection{Randomized method for preparing encoded \texorpdfstring{$\ket 0$}{0}} \label{sec:RandomSteane}

The standard method for preparing encoded states for CSS stabilizer codes, including the Golay code, is to construct and solve a partial Latin rectangle based on the stabilizer generators \cite{Steane02}.  To prepare $\ket{0}$ in the Golay code, consider the eleven stabilizer generators that are tensor products of Pauli $X$ operators.  These stabilizer generators form an $11\times 23$ binary matrix in which the $X$ operators in the tensor product are represented as $1$s, as in Eq.~\eqnref{eq:golay-stabilizers}.  Gaussian elimination is performed until the matrix is of the form 
\begin{equation}
{\scriptstyle{11}} \left\{ \left(
\begin{array}{c|c}
\raisebox{0ex}[1.5ex]{\mbox{\ensuremath{\overbrace{I}^{11}}}} & 
\raisebox{0ex}[1.5ex]{\mbox{\ensuremath{\overbrace{A}^{12}}}}
\end{array}
\right)\right.
\end{equation}
The first eleven qubits, called ``control" qubits, are prepared as $\ket{+}$, and the remaining ``target" qubits are prepared as $\ket{0}$.  The matrix $A$ represents a partial Latin rectangle, the solution to which is used to schedule rounds of CNOT gates from control to target qubits.  

\begin{table}
\centering 
\begin{tabular}{c|cccccccc}
\hline \hline 
Weight: & 
0 & 1 & 2 & 3 & 4 & 5 & 6 & 7 \tabularnewline
\hline
Number of $X$ errors: &
1 & 23 & 253 & 1771 & 1771 & 253 & 23 & 1 \tabularnewline
Number of $Z$ errors: & 
1 & 23 & 253 & 1771 & 0 & 0 & 0 & 0 \tabularnewline
\hline \hline 
\end{tabular}
\caption{The number of errors on Golay encoded $\ket{0}$ by Hamming weight.  All $Z$ errors are correctable so there are no $Z$ errors of weight greater than three.}
\label{tbl:ketZero-errorWeights}
\end{table}

An $X$ error in the preparation circuit can propagate to other qubits only if it occurs on a control qubit, and then only through the $X$ stabilizer being created from that control qubit.  Thus single faults can create up to $22$ weight-two errors (for each of the eleven $X$ stabilizers, either $IIIIIIXX$ or $IIXXXXXX \sim XXIIIIII$), $22$ weight-three errors and eleven weight-four errors ($IIIIXXXX$ for each stabilizer).  

A single $X$ fault, i.e., a fault resulting in an $X$ error, cannot break the verification circuit in \figref{fig:FourAncillaVerifyCkt}.  If it creates a correlated error on the first ancilla, that error will be detected on the second ancilla, and both will be discarded.  Four or more $X$ faults also cannot break the verification circuit because we only seek fault tolerance up to order three.  

Two $X$ faults can break the verification circuit only if there is one failure in each ancilla preparation that propagates to an error of weight at least three---necessarily the same error so that it is undetected.  To obtain a crude estimate for how likely this is to occur, pretend that the correlated errors created by a random preparation circuit are uniformly distributed among all errors of the same weights.  The number of errors on encoded $\ket{0}$ for each weight are given in \tabref{tbl:ketZero-errorWeights}.  Then the probability that two preparation circuits share no such correlated errors is estimated as 
\begin{equation*}
\frac{\binomial{1771-22}{22}}{\binomial{1771}{22}} \cdot \frac{\binomial{1771-11}{11}}{\binomial{1771}{11}} \approx 0.71
 \enspace .
\end{equation*}

Three $X$ errors can break the circuit if they lead to an undetected error of weight four or greater on the first ancilla.  Consider the case that there are two failures while preparing the first ancilla and one failure while preparing the second ancilla.  The number of different weight-four errors created with second-order probability (i.e., excluding those created with first-order probability) depends on the circuit. For ten random circuits, the smallest count we obtained was $688$ and the largest $735$, with an average of $711$.  Using this average value, we estimate that the probability of a random circuit succeeding against three $X$ errors is roughly $[\binomial{1771-711}{11} / \binomial{1771}{11}]^2 \approx 1.2 \cdot 10^{-5}$.  (Here the square is because we want the circuit to work against both the case of two failures in the first ancilla, one failure in the second, and vice versa.) Overall, we expect to have to try about $1.2 \cdot 10^5$ random pairs of preparation circuits before we find one that gives fully fault-tolerant $X$-error verification.

The result of $X$-error verification is a single ancilla free of correlated $X$ errors up to weight-three, but possibly containing correlated $Z$ errors.  The $Z$-error propagation can be analyzed in a manner similar to that used for $X$ errors.  A single failure in an $X$-error verified ancilla can produce roughly $60$ $Z$ errors of weight three. Again assuming a uniform distribution, the probability of finding two $X$-error verified ancillas that share no correlated $Z$ errors of weight three is $\binomial{1771-60}{60} / \binomial{1771}{60} \approx 0.12$.  In total, we expect to try about five $X$-error fault-tolerant pairs in order to find two pairs that are fully fault-tolerant for both $X$-error and $Z$-error verification, as $\binomial{5}{2} = 10$.  

To find fault-tolerant verification circuits in this way, one needs to be able to generate sufficiently random preparation circuits.  As the Latin rectangle procedure for finding encoding circuits is fully algorithmic, it can be randomized by starting with a random presentation of the Golay code.  Alternatively, one can begin with a fixed encoding circuit and randomly permute the seven rounds of CNOT gates (all of the CNOTs commute), or permute the qubits according to a random element of the symmetry group~$M_{23}$.  By trying roughly $10^5$ random pairs, we found $14$ pairs of ancillas that were fully fault-tolerant against $X$ errors.  Of the $\binomial{14}{2}$ combinations, six were also fully fault-tolerant against $Z$ errors.  \tabref{tab:randomizedSchedules} presents one such set.  

\begin{table}
\newcommand{\myfontsize}{\fontsize{7}{9}\selectfont}
\myfontsize
\def\colspace{{$\;\,$}}
\centering
\addtolength{\subfigcapskip}{0.1cm}
\begin{tabular}{c c}
\subtable[Ancilla 1]{
\begin{tabular}{c|r@{\colspace}r@{\colspace}r@{\colspace}r@{\colspace}r@{\colspace}r@{\colspace}r}
  & 1& 2& 3& 4& 5& 6& 7\\ \hline
 2& 0&22& 7&11& 4& 8&19\\
 3& 9&19& 4& 8& 7& 1& 6\\
10& 5& 1& 0& 6&14& 7& 9\\
12& 1& 0&14& 5&22&11& 4\\
13& 6& 8&22& 9& 0& 4& 5\\
15& 4& 5& 9&14&19&22& 7\\
16&14& 7& 5& 4&11& 6& 8\\
17& 8&11& 6&19& 5& 0& 1\\
18& 7& 9& 1&22& 8& 5&11\\
20&19& 6&11& 7& 1&14&22\\
21&11& 4&19& 0& 6& 9&14
\end{tabular}
}
&
\subtable[Ancilla 2]{
\begin{tabular}{c|r@{\colspace}r@{\colspace}r@{\colspace}r@{\colspace}r@{\colspace}r@{\colspace}r}
  & 1& 2& 3& 4& 5& 6& 7\\ \hline
 0& 5&16&17&22& 1&15& 9\\
 3&15& 2& 6& 5&17&16&11\\
 7& 1&22& 4&17& 2& 5& 6\\
 8& 6&13&16& 1&15& 4&17\\
10&22&11& 5&13&16& 6& 1\\
12& 9&17&13& 2& 6&22&16\\
14& 4& 6&11&15&13& 2&22\\
18&16& 1&15&11& 9&13& 2\\
19&17& 4& 1& 9&22&11&13\\
20&11&15& 9& 6& 4& 1& 5\\
21& 2& 5&22&16&11& 9& 4
\end{tabular}
}
\\
\subtable[Ancilla 3]{
\begin{tabular}{c|r@{\colspace}r@{\colspace}r@{\colspace}r@{\colspace}r@{\colspace}r@{\colspace}r}
  & 1& 2& 3& 4& 5& 6& 7\\ \hline
 1&21&16& 7&13&10&15& 0\\
 2&16& 7&12& 0&18&19&13\\
 3&13& 0&15&12&19&10&20\\
 4&12&21&18&20& 7&13&10\\
 5& 6&13&21&10& 0&18&19\\
 8&18&19&13&21&15&20&16\\
 9&19& 6&10&15&20& 7&21\\
11&20&12& 6& 7&13&16&15\\
14& 7&18&20&16&21& 0& 6\\
17& 0&15&19& 6&16&21&12\\
22&10&20&16&19& 6&12&18
\end{tabular}
}
&
\subtable[Ancilla 4]{
\begin{tabular}{c|r@{\colspace}r@{\colspace}r@{\colspace}r@{\colspace}r@{\colspace}r@{\colspace}r}
  & 1& 2& 3& 4& 5& 6& 7\\ \hline
 0& 1&16& 3&12&17&13&11\\
 2&22&18&14& 3&20&17& 6\\
 4& 3&20& 6& 1&12&22&13\\
 5& 6&14&16&20& 1&12&17\\
 7&20&22&17&13&16&18& 1\\
 8&16& 6&18&11& 3& 1&20\\
 9&18&12&13&16&14&20& 3\\
10&14&17&20&22&13&11&12\\
15&12&11& 1&17&18& 6&22\\
19&17&13&11&18& 6&16&14\\
21&11& 3&12& 6&22&14&16
\end{tabular}
}
\end{tabular}
\addtolength{\subfigcapskip}{-0.1cm}
\caption{
Four seven-round ancilla-preparation schedules.  In each table, the entry in row $i$, column $j$ specifies the target qubit of a CNOT gate with control qubit $i$ applied in round~$j$.  Using these schedules in the verification circuit of \figref{fig:FourAncillaVerifyCkt}, the output encoded $\ket 0$ state is fully fault-tolerant against both $X$ and $Z$~errors.  
} \label{tab:randomizedSchedules}
\end{table}

\subsection{Overlap method for preparing encoded \texorpdfstring{$\ket 0$}{0}} \label{sec:Overlap}

The above procedure for finding a fault-tolerant verification circuit uses ancilla preparation circuits constructed from the Latin rectangle method. We now show an alternative construction based on a modification of the Latin rectangle method.  By carefully analyzing the stabilizer generators of the Golay code we can reduce the number of CNOT gates required to prepare encoded $\ket 0$.  

To explain the optimization, first consider the Steane $[[7,1,3]]$ code.  A Latin rectangle-based encoding schedule, shown in \figref{fig:steane-prep}, needs nine CNOT gates.  An equivalent circuit requiring only eight CNOT gates is shown in \figref{fig:steane-overlap}. This circuit removes two of the CNOTs for which qubit six is a target and replaces them with a single CNOT from qubit five to qubit six in round three.  This works because in~\ref{fig:steane-prep} qubits five and six are both the targets of CNOTs from qubits one and three; the corresponding stabilizer generators overlap on qubits five and six.  

\begin{figure}
\centering
\leavevmode
\subfigure[]{
\includegraphics[scale=.75]{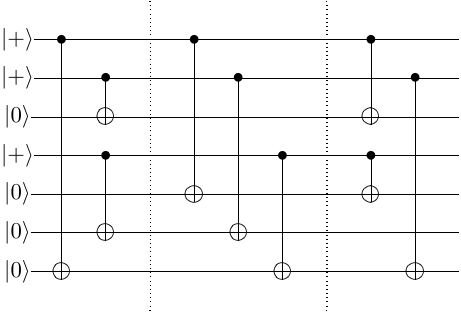}
\label{fig:steane-prep}
}
$\qquad\quad$
\subfigure[] {
\includegraphics[scale=.75]{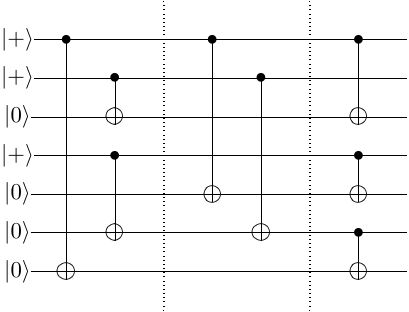}
\label{fig:steane-overlap}
}
\caption{Two alternative circuits for preparing encoded $\ket 0$ in the Steane code, a self-dual code with $X$ stabilizers $IIIXXXX$, $IXXIIXX$ and $XIXIXIX$.  The circuit in (a) follows Steane's Latin rectangle encoding method.  The circuit in (b) prepares the same state using one fewer CNOT gate.  The new CNOT gate has the same effect as the two removed gates.  
}
\end{figure}

The same technique of exploiting overlap between stabilizers extends to larger CSS codes.  For the Golay code, \figref{fig:overlap-prep-ckt} gives an encoding circuit with only $57$ CNOT gates, $20$ fewer than given by the Latin rectangle method.  

By reducing the number of CNOT gates, this circuit also reduces the number of correlated errors.  For example, a single failure in the Latin rectangle encoded circuits can cause up to $22$ weight-two errors, but a single failure in the new circuit can only cause up to $16$ weight-two errors.  The correlated error counts for first and second order are shown in \tabref{tbl:overlap-errorWeights}.  The smaller number of correlated errors means that it should be easier to find fault-tolerant circuits by randomization.  However, unlike Steane schedules the overlap schedule depends on a fixed code presentation and on a fixed round ordering, since the CNOT gates do not commute.  

\begin{table}
\centering 
\begin{tabular}{c|cccccccc|}
\hline \hline 
$X$-error weight: & 
2 & 3 & 4 & 5 & 6 & 7 \tabularnewline
\hline
Order 1: &
16 & 14 & 4 & 0 & 0 & 0 \tabularnewline
Order 2: & 
- & 493 & 400 & 35 & 2 & 0 \tabularnewline
\hline \hline 
\end{tabular}
\caption{Correlated $X$-error counts for the encoded $\ket{0}$ circuit in \figref{fig:overlap-prep-ckt}.} \label{tbl:overlap-errorWeights}
\end{table}

To obtain randomized overlap method encoding circuits, we use the permutation symmetry of the Golay code and permute the qubits of \figref{fig:overlap-prep-ckt} according to a pseudo-random element of the symmetry group~$M_{23}$.  By analyzing the correlated error sets of randomly permuted circuits, we have found many sets of fault-tolerant four-ancilla preparation circuits.  In fact, we have even found sets for which the order required for a weight-$k$ error to pass verification is at least $k+1$ (rather than $k$) for all $k \leq 2$. This reduces, for example, the probability of accumulating an uncorrectable error on the data block by first a weight-two error in $Z$-error correction and then another weight-two error in $X$-error correction.  One such set of four permutations is given in \tabref{tbl:overlap-permutations}.  

We briefly note that the overlap method, and the circuit in~\figref{fig:overlap-prep-ckt} in particular, may not be optimal.  Indeed the are equivalent circuits with fewer CNOT gates. However, \figref{fig:overlap-prep-ckt} is the smallest circuit we found that also preserves \emph{depth}. In the asymptotic setting, for arbitrarily large circuits of CNOT gates, our method bares resemblance to the algorithm presented in~\cite{Patel2003}.  Both methods exploit similarities across columns (or rows) of a matrix to eliminate CNOT gates.  Our method differs in that we use only the redundancy matrix rather than the full $n\times n$ linear transformation, and we exploit similarities between columns without first using Gaussian elimination to make the columns identical.  This way, making the optimization by hand, we are able to preserve circuit depth.

\section{Threshold analysis} \label{sec:threshold-analysis}

The remainder of this article focuses on analyzing the noise threshold and resource overhead for fault-tolerant quantum computation using the Golay code ancilla preparation and verification circuits from \secref{sec:ancilla-verification}.  Our threshold analysis relies on a malignant set counting technique that is tailored for a depolarizing noise model.  The counting technique is outlined in \secref{sec:Counting} and proof of the threshold lower bound is presented in \secref{sec:calc-threshold}.  (Some details are given in the appendices.)  Threshold calculation results for our circuits are discussed in \secref{sec:results-threshold}. Resource overhead is considered in~\secref{sec:Overhead}.  

\subsection{Noise model} \label{sec:NoiseModel}

We begin by defining the depolarizing noise model, a standard model used before in, e.g.,~\cite{Knill04NoisyDevices}.  We study noisy circuits constructed from the following physical operations: $\ket{0}$ and $\ket{+}$ initialization, a CNOT gate, and single-qubit measurement in the $Z$ and $X$ eigenbases.  Every qubit in the computer can be involved in at most one operation per discrete time step.  CNOT gates are allowed between arbitrary qubits, without geometry constraints.  Resting qubits are also subject to noise.  

\begin{definition}[Independent depolarizing noise with parameter $\gamma$]
\label{def:depolarizing-noise}
Noisy operations are modeled~by: 
\begin{enumerate}
\item
A noisy CNOT gate is a perfect CNOT gate followed by, with probability $\frac{16}{15} \gamma$, the simultaneous depolarization of the two involved qubits.  Equivalently, after applying the ideal CNOT gate, with probability $15\gamma$ a non-trivial two-qubit Pauli error drawn uniformly and independently from $\{I,X,Y,Z \}^{\otimes 2} \setminus \{I \otimes I\}$ is applied.    
\item
Noisy preparation of a $\ket 0$ state is modeled as ideal preparation of $\ket 0$, followed by application of an $X$ error with probability $4 \gamma$.  Similarly, noisy preparation of $\ket +$ is modeled as ideal preparation of $\ket +$ with probability $1-4\gamma$ and of $\ket - = Z \ket +$ with probability $4 \gamma$.
\item
Noisy $Z$-basis ($\ket 0, \ket 1$) measurement is modeled by applying an $X$ error with probability $4 \gamma$, followed by ideal $Z$-basis measurement.  Similarly, noisy $X$-basis ($\ket +, \ket -$) measurement is modeled as ideal measurement except preceded by a $Z$ error with probability $4 \gamma$.  
\item
A noisy rest operation is modeled as applying either the identity gate, with probability $1-12\gamma$, or with probability $4 \gamma$ each, one of the Pauli errors $X$, $Y$ or $Z$.  
\end{enumerate}
All locations fail independently of each other.  
Let $p = 15 \gamma$, the probability for a CNOT gate to fail.  
\end{definition}

To justify this noise model, note that the noise on a resting qubit is the one-qubit marginal of the CNOT gate noise.  The noise rate for preparation and measurement is lower, only $4 \gamma$, because any higher noise rate could be reduced to $4 \gamma + O(\gamma^2)$ by repeating the preparation or measurement operation using two qubits coupled by a CNOT.  

During error counting, $X$ and $Z$ errors are usually considered separately and the error probability is computed by omitting the $Z$ or $X$ part of each error, respectively.  For example, when considering only $X$ errors $XY$ is equivalent to $XX$, $XZ$ is equivalent to $XI$ and so on.  Thus, the marginal distribution of $X$ errors for a CNOT gate applies with probability $12 \gamma$ a uniformly random error from $\{ IX, XI, XX \}$.  Similarly, for $Z$ errors, the  marginal error distribution applies a random error from $\{ IZ, ZI, ZZ \}$.  For preparing $\ket 0$ or measuring in the $Z$ basis, no $Z$ errors are possible, and similarly no $X$ errors are possible for preparing $\ket +$ or measuring in the $X$ basis.  The marginal $X$- and $Z$-error distributions for a rest are to apply $X$ and $Z$, respectively, errors with probability $8 \gamma$.

\subsection{Counting malignant sets} \label{sec:Counting}

As we have shown, our two optimized ancilla preparation and verification circuits significantly reduce the overhead required for fault-tolerant ancilla preparation.  We would also like to know how these circuits impact the tolerable noise threshold.   With fewer verification stages our ancillas are slightly more likely to contain errors and so one might expect a lower noise threshold when compared to previous verification circuits.  On the other hand, the smaller size of our circuits makes it easier to give a tighter analysis.  
  
The threshold calculation is most limited by the exRec with the largest number of locations.  The Golay code admits transversal implementations of encoded Clifford group unitaries.  Universality can be achieved by injection and distillation~\cite{BravyiKitaev04magic,Knill04NoisyDevices}, which involves only Clifford group unitaries, and the single-qubit preparations and measurements that are already assumed by our model.  Therefore the largest exRec in our case is for the encoded CNOT gate, an exRec that consists of four Steane-type error corrections plus $23$ CNOT gates (see \figref{fig:cnot-exrec-components}).  \tabref{tab:MarginalCounts} gives a breakdown of the number of locations for our preparation circuits, and the total number of locations in the CNOT exRec.  

Monte Carlo simulations of circuits using the Golay code~\cite{Steane03,CrossDiVincenzoTerhal} indicate that the depolarizing noise threshold should be on the order of $p = 10^{-3}$.  Unfortunately, it is not straightforward to prove such a high threshold using malignant set counting.  For example, say that we check for malignancy all location subsets of size up to $\kGood$, and we assume that all larger subsets are malignant.  Then the estimate we obtain for the probability of an incorrect rectangle is at least $\sum_{k=\kGood+1}^n \binomial{n}{k} p^k (1-p)^{n-k}$.  For $n = 5439$ locations and $p = 10^{-3}$, this term drops below $10^{-3}$ only for $\kGood \geq 14$.  However, there are more than $10^{41}$ subsets of size at most $14$, so checking them one at a time is computationally intractable.  
  
Instead of checking each set for malignancy, one can sample random sets of locations in order to estimate the fraction that are malignant.  This technique, called malignant set sampling, can provide threshold estimates with statistical confidence intervals.  However, both malignant set counting and sampling techniques study the threshold for worst-case adversarial noise, and may be overly conservative for a more physically realistic, non-adversarial noise model such as depolarizing noise.  For example, malignant set sampling results from~\cite{AliferisCross07subsystem} estimate a threshold of only $p \approx 10^{-4}$ for the Golay code.  

\begin{figure}
\centering
\begin{tabular}{c@{$\qquad$}c}
\raisebox{-.5in}{\setcounter{subfigure}{0}\subfigure[\label{fig:exRec}]{
\includegraphics[width=5cm]{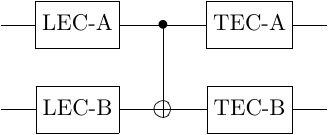}
}} 
& 
\raisebox{-.5in}{\setcounter{subfigure}{2}\subfigure[\label{fig:ancilla-verification}]{
\includegraphics[width=6.5cm]{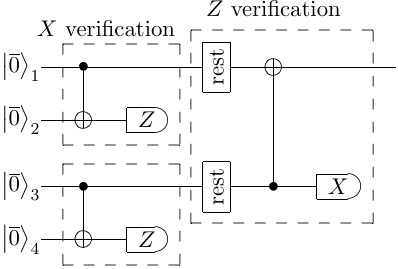}
}}
\\[2cm] \multicolumn{2}{c}{\setcounter{subfigure}{1}\subfigure[\label{fig:ECclean}]{
\includegraphics[width=13.5cm]{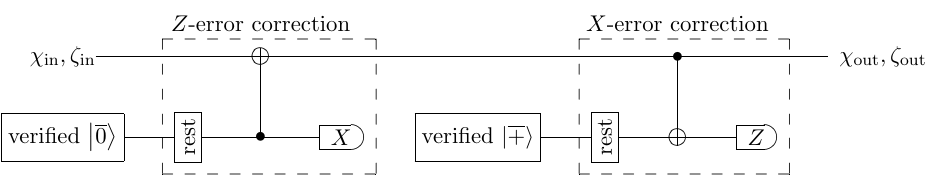}
}}
\end{tabular}
\caption{
Organization of a CNOT extended rectangle, or ``exRec."  
(a) The CNOT exRec includes four error corrections, two leading (LEC) and two trailing (TEC), and a transversal CNOT gate.  (b) Each error-correction component consists of separate $Z$ and $X$ error corrections.  $Z$-error correction requires a $\lket{0}$ state that has been verified against errors, and $X$-error correction requires a verified $\lket{+}$ ancilla state.  (c) A verified $\lket 0$ state is prepared by checking two pairs of prepared $\lket 0$ states against each other for $X$ errors, then, conditioned on no $X$ errors being detected, checking the results against each other for $Z$ errors.  Verified $\lket +$ is prepared by taking the dual of the $\lket 0$ circuit.  These components are discussed, in reverse order, in Sections~\ref{sec:X-verification} to~\ref{sec:counting-exRec}.  
} \label{fig:cnot-exrec-components}
\end{figure}

\begin{table}
\centering
\begin{tabular}{c|cccc|c|c}
\hline \hline
$\lket 0$ preparation & \multicolumn{4}{c|}{Location type} &  &CNOT exRec \\ 
circuit & CNOT & Prep. & Meas. & Rest & Total & total \\
\hline
Steane & 77 & 23 & 0 & 6 & 106 & 5439 \\
Overlap & 57 & 23 & 0 & 38 & 118 & 5823 \\
\hline \hline
\end{tabular}
\caption{Location counts for preparing encoded $\ket 0$ in the Golay code. Encoded $\ket{0}$ ancillas are prepared with either the pseudorandomly constructed Steane preparation circuits (\tabref{tab:randomizedSchedules}), or the overlap preparation circuits (\figref{fig:overlap-prep-ckt} and \tabref{tbl:overlap-permutations}).  The last column shows the total number of locations inside the CNOT exRec shown in~\figref{fig:cnot-exrec-components}, including the transversal CNOT operation and four error corrections.}
\label{tab:MarginalCounts}
\end{table}

We therefore present an alternative to malignant set counting that is tailored to circuits based on the Golay code and depolarizing noise.  Roughly, we divide the exRec into a hierarchy of components and sub-components.  We then compute an upper bound on the probability of each error a component may produce, essentially by checking location sets up to a certain small size.  At the exRec level, we synthesize the component error bounds into upper bounds on the probability that the rectangle is incorrect.  The resulting error probabilities are treated as an effective transformed noise model for the encoded circuit.  With some care, the transformed noise model can be fed recursively back into the procedure to determine an effective noise model for the next level of encoding, and so on.  

Effectively, dividing the exRec into components allows us to account efficiently for even very large location subsets.  Most large sets will be roughly evenly divided between the components, with only a small number of locations in each component.  

The remainder of this section outlines our modified malignant set counting technique.  Details of the threshold analysis are given in \secref{sec:calc-threshold}.

\subsubsection{Characterizing exRec components} \label{sec:basic-operation}

We will divide the exRec into its encoded operation and its error corrections.  The error corrections will each divide into $X$-error correction and $Z$-error correction, and further recursive divisions will continue until reaching the physical location~level.  

\begin{figure}
\centering
\includegraphics[scale=1]{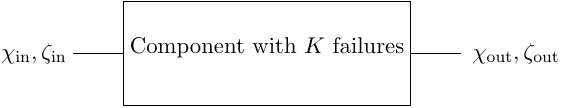}
\caption{A circuit component with input error $(\chiIn, \zetaIn)$ and output error $(\chiOut, \zetaOut)$} \label{fig:component}
\end{figure}

\def\xIn{x_\text{in}}
\def\zIn{z_\text{in}}
\def\xOut{x_\text{out}}
\def\zOut{z_\text{out}}

Each component in the hierarchy has input error $(\chiIn,\zetaIn)$, some number of internal failures~$K$, and output error $(\chiOut,\zetaOut)$ which depends on the internal failures and on the input error (see \figref{fig:component}).  For every error equivalence class on the inputs and outputs and for every~$k \in \N$, we would like to compute
\begin{equation}
  \label{eq:pr-component}
  \Pr\big[(\chiOut, \zetaOut)=(\xOut,\zOut), K=k \,\vert\, (\chiIn, \zetaIn)=(\xIn,\zIn)\big] \enspace ,
\end{equation}
the probability that there are exactly $k$ failures and the output error is $(\xOut, \zOut)$ conditioned on the input error $(\xIn, \zIn)$.  Here, the notation $(x, z)$ indicates an error equal to the product $xz$ where~$x$ is a tensor product of $X$ and $I$ operators and $z$ is a tensor product of $Z$ and $I$ operators.  

For components that are physical gate locations the probability in~\eqnref{eq:pr-component} is defined by the depolarizing noise model (\defref{def:depolarizing-noise}).  Larger components are analyzed by first analyzing each enclosed sub-component.  At the exRec level the LEC, transversal CNOT and TEC components provide all of the information necessary to determine the probability that the exRec is incorrect.  Indeed, we shall see in \secref{sec:counting-exRec} that they contain enough information to compute the probability for each \emph{way} that the exRec can be incorrect.  

There are, however, two logistical problems.  First, on each $23$-qubit code block, there are $2^{12}$ inequivalent $X$ errors and $2^{12}$ inequivalent $Z$ errors, and thus $2^{24}$ inequivalent Pauli errors total.  (For example, there are $2^{23}$ different tensor products of $X$ and $I$ operators, and the group of $X$ stabilizers has size $2^{11}$, leaving $2^{12}$ $X$-error equivalence classes.)  For a component involving two code blocks, this means we should compute for each~$k$ up to $(2^{24})^4$ quantities, one for each combination of input and output errors.  Second, since there are $\binomial{n}{k}$ size-$k$ subsets of $n$ locations and since each CNOT gate has $15$ different ways to fail, a computation that accounts for all possibilities scales roughly as $\binomial{n}{k} 15^k$.  Such a computation is feasible only for small~$k$ and small~$n$.  

The first problem can be solved by observing that in Steane error correction $X$ errors and $Z$ errors are corrected separately.  Furthermore, there are no Hadamard gates or other ways of transforming an $X$ error into a $Z$ error, or vice versa, so $X$ and $Z$ errors mostly propagate independently.  $X$ and $Z$ errors cannot be treated independently entirely, because $X$ and $Z$ failures are highly correlated in the depolarizing noise model, and the postselection steps in ancilla verification could amplify any initial dependencies.   Still, for most components, the $X$-error part of the output of a component depends only on the $X$-error part of the input and the $X$ failures that occur inside the component.  A similar observation holds for $Z$ errors.  Thus, expression~\eqnref{eq:pr-component} may be split into separate $X$ and $Z$ parts: 
\begin{subequations}\begin{align}
  \label{eq:prX-component}
  \Pr[\chiOut=\xOut, K_X &=k \vert \chiIn=\xIn] \\
  \label{eq:prZ-component}
  \Pr[\zetaOut=\zOut, K_Z &=k \vert \zetaIn=\zIn] \enspace .
\end{align}\end{subequations}
Here, the random variable $K_X$ is the number of failures inside the component that contain an $X$ when decomposed into a tensor product of Pauli operators.  The value $K_Z$ is similarly defined for~$Z$.  When considering $X$ and~$Z$ errors separately, the input and output of a two-block component contain at most $2^{24}$ inequivalent errors and the worst case combination is a large but manageable $2^{48}$ cases.

The second problem is eliminated by noting that, for a fixed $k$, the probability of an order-$k$ fault decreases rapidly as the size of the component decreases.  For example, for $p = 1\times 10^{-3}$, the probability of an order-ten fault in the exRec is about $0.027$.  However, the probability that all ten failures are located in a single error correction is only about $1.4 \times 10^{-6}$. Thus there is little gain in counting errors of order-ten or higher in the error correction component.  

In general, the probability that a component contains a fault of order greater than $\kGood$ can be bounded according to
\begin{equation}
  \Pr[K > \kGood] \leq \sum_{k=\kGood+1}^n \binomial{n}{k} (1-p)^{n-k} p^k \enspace .
  \label{eq:prBad-bound}
\end{equation}
(A tighter bound can be achieved by considering separate $k$ for each location type.  See \appref{sec:bounding-bad}.)  We will choose a value of $\kGood$ for each component and then pessimistically assume that all faults of order greater than $\kGood$ within the component cause the rectangle to be incorrect. For large enough values of~$\kGood$ the overall impact on the threshold is negligible.  There is a tradeoff here between running time and accuracy.  A larger value of $\kGood$ yields a more accurate bound on the probability that the rectangle is incorrect.  A smaller value of $\kGood$ is easier to compute.  We must choose for each component a suitable $\kGood$ that balances the two.  

In the end we are left with two sets of faults for each component, those of order at most~$\kGood$ and those of order greater than~$\kGood$.  Each fault in the first set is counted to obtain accurate estimates of \eqnref{eq:prX-component} and \eqnref{eq:prZ-component}.  When a fault from this set occurs we call it a \emph{good} event.  Faults in the second set are not counted and are instead bounded using \eqnref{eq:prBad-bound} and pessimistically added to the final incorrectness probability bounds for the exRec.  When a fault from this set occurs we call it a \emph{bad} event.  The probability that the rectangle is incorrect is then upper-bounded by 
\begin{equation*}
  \Pr[\incorrect] \leq \Pr[\incorrect, \good] + \Pr[\bad]
  \enspace .
\end{equation*}

In general, there are four quantities we need to upper bound for each component:
$\Pr[\chiOut=\xOut, K_X=k, \goodX \vert \chiIn]$, $\Pr[\zetaOut=\zOut, K_Z=k, \goodZ \vert \zetaIn]$, $\Pr[\badX]$, and $\Pr[\badZ]$.
The event $\goodX \equiv \neg \badX$ occurs when there is a set of $X$-error failures in the component that we choose to count.  It will usually depend only on~$\kGood$ in which case $\goodX \Leftrightarrow (K_X \leq \kGood)$. In some cases $\goodX$ may depend on a vector~$\vec{k}$ representing the number of $X$-error failures across multiple sub-components.  The event $\goodZ \equiv \neg \badZ$ is similarly defined for $Z$.  

In the remainder of this section we outline the procedure for computing the above quantities for each component of the CNOT exRec.  A more precise analysis is presented in \appref{sec:counting-detail}.

\subsubsection{\texorpdfstring{$X$}{X}-error verification} \label{sec:X-verification}

$X$-error verification requires two encoded $\lket 0$ states. The first is verified against the second for $X$ errors by applying transversal CNOT gates between the two code blocks and then measuring each qubit of the second block in the $Z$ eigenbasis ($\ket 0, \ket 1$ basis).  Conditioned on no $X$ errors being detected, the first code block is accepted.  See \figref{fig:ancilla-verification}.  

Letting $\acceptX$ denote the event that no $X$ errors are detected, we use Bayes's rule 
\begin{equation}
  \Pr[\text{event} \vert \acceptX]=\frac{\Pr[\text{event}, \acceptX]}{\Pr[\acceptX]}
\end{equation} 
to compute the conditional probabilities of different error events.  For an event $\chi$ involving only~$X$ errors, this calculation is straightforward.  

However, if the event is a $Z$ error $\zeta$, then the numerator $\Pr[\zeta=z, \acceptX]$ is difficult to compute as it mixes $X$ and $Z$ errors.  The obvious bound, $\Pr[\zeta=z, \acceptX] \leq \Pr[\zeta=z]$, is quite pessimistic because in the depolarizing noise model we expect $X$ errors to occur with $Z$ errors roughly half of the time, and so $X$-error verification should remove many $Z$ errors.  It is important to obtain an accurate count of $Z$ errors since they strongly influence the acceptance rate of the upcoming $Z$-error verification.  Therefore, we also count $X$ and $Z$ errors \emph{together} for very low-order faults and apply a correction to the $Z$-only counts.  Details of the correction are worked out in \appref{sec:X-verification-details}.

\subsubsection{\texorpdfstring{$Z$}{Z}-error verification} \label{sec:Z-verification}

$Z$-error verification is similar to $X$-error verification.  However, as shown in \figref{fig:ancilla-verification}, we add a pause, i.e., transversal rest operations, to allow the preceding $X$-error postselection to complete.  

Similar to $X$-error verification, all events are now conditioned on the event $\acceptZ$ of no~$Z$ errors being detected.  When considering an $X$-error event $\chi$, we generally use the pessimistic inequality $\Pr[\chi=x, \acceptZ] \leq \Pr[\chi = x]$.  This inequality is less of a problem than the similar inequality $\Pr[\zeta = z, \acceptX] \leq \Pr[\zeta = z]$ we encountered during $X$-error verification, for two reasons.  First, many $X$ errors have already been eliminated, so the probabilities start out much lower.  Second, overestimating the probability of an $X$ error now is less serious; since there are no remaining postselection steps, the distribution of errors will not need to be renormalized again.  Even so, we count $X$ and $Z$ errors together for very low-order faults, since it is relatively easy to do so.

\subsubsection{Error correction} \label{sec:counting-EC}

An error-correction component consists of $Z$-error correction and $X$-error correction, as shown in \figref{fig:ECclean}.  Each sub-component begins with a pause on the input verified ancilla state, to allow for the previous postselection to complete.  After extracting the error syndrome, the lowest-weight correction is computed, and this correction is applied by a change in the qubits' Pauli frames~\cite{Knill04NoisyDevices}.  

There are two types of error correction components, leading error correction (LEC) and trailing error correction (TEC).  For the LEC, we may assume that the input errors $\chiIn$ and $\zetaIn$ are both zero.  This because the probability that the rectangle is incorrect depends only on the syndrome of the output of the LEC and that syndrome depends only on the errors inside of the LEC~\cite{CrossDiVincenzoTerhal}.  That is, we do not care about the logical state at the output of the LEC, we care only that it is correctly manipulated by the rectangle.  For trailing error correction, we care only about the result of applying a logical decoder to the output.  In other words, we only need to know whether the output errors $\chiOut$ and $\zetaOut$ represent correctable errors or not.  The four relevant quantities are  
\begin{align*}
  \Pr[\chiOut=\xOut, K_X=k, \good \vert \chiIn=0]& &\Pr[D(\chiOut)=d, K_X=k, \good \vert \chiIn=\xIn] \\
  \Pr[\zetaOut=\zOut, K_Z=k, \good \vert \zetaIn=0]& &\Pr[D(\zetaOut)=d, K_Z=k, \good \vert \zetaIn=\zIn] 
\end{align*}
where $d \in \{0, 1\}$ and $D(e)$ identifies whether $e$ is a correctable error ($0$) or an uncorrectable error ($1$).  That is, $D(e)=1$ if and only if $e$ decodes to a nontrivial Pauli error.

\subsubsection{exRec} \label{sec:counting-exRec}

The CNOT exRec, shown in \figref{fig:exRec}, is divided into five components: two leading error corrections, a transversal CNOT, and two trailing error corrections.  At this level, we are interested in \emph{malignant} events---the events for which the rectangle is incorrect. More specifically, when a malignant event occurs we would like to know \emph{how} the rectangle is incorrect.  

Let $\ket{\psi_1}$ be the two-qubit state obtained by applying ideal decoders on the two blocks of the CNOT immediately following the LECs.  Similarly let $\ket{\psi_2}$ be the state obtained by applying ideal decoders immediately following the TECs. Then define $\malig_{IX}$ as the event that $(I\otimes X) U_{cnot} \ket{\psi_1} = \ket{\psi_2}$, where $U_{cnot}$ is the two-qubit unitary corresponding to the ideal CNOT gate.  Similarly define the events $\malig_{XI}$, $\malig_{XX}$, $\malig_{IZ}$, $\malig_{ZI}$, $\malig_{ZZ}$.  The event $\malig_E$ can be informally interpreted as the event in which the rectangle introduces a ``logical'' error $E$.  

The relevant quantities are $\Pr[M_X, K_X=k, \good]$ and $\Pr[M_Z, K_Z=k, \good]$ for $M_X \in \{ \malig_{IX}, \malig_{XI}, \malig_{XX} \}$ and $M_Z \in \{ \malig_{IZ}, \malig_{ZI}, \malig_{ZZ} \}$.  Since we count $X$ and $Z$ errors separately, it is not possible to compute logical $Y$ error quantities.  Our analysis will therefore double-count $Y$ errors.  Intuitively this is not a great loss, because the correlations between $X$ and $Z$ are much smaller at this level.  In the next section we show how to use these quantities to compute a lower bound on the threshold for depolarizing noise.

\subsection{Calculating the error threshold} \label{sec:calc-threshold}

As discussed in \secref{sec:codeConcat-noiseThresh}, the standard way of calculating the asymptotic error threshold involves finding subsets of faulty exRec locations (called ``malignant") for which some combination of Pauli errors at those locations causes the enclosed rectangle to be incorrect.  Our counting method is different.  We count subsets of faulty locations, but the counted information is synthesized into error probability upper bounds based on a particular noise model and error correction scheme.

In this section we outline an alternative method for rigorously lower bounding the noise threshold that is tailored specifically to the information obtained by our counting procedure.  The basic idea is to treat each level-one rectangle in the level-two simulation as a single ``location'' with a transformed noise model based on the malignant event upper bounds obtained in \secref{sec:Counting}.  In particular, we show how to treat each level-one exRec independently while maintaining valid upper bounds on the error probabilities.

\subsubsection{Calculating the pseudo-threshold}

One quantity that is particularly easy to calculate from our counts is the so-called pseudo-threshold~\cite{SvoreCrossChuangAho} for the CNOT location. The pseudo-threshold for location $l$ is defined as the solution to the equation $p = p_l^{(1)}$, where $p_l^{(1)}$ is the probability that the $1$-Rec for location $l$ is incorrect. We may compute a lower bound on the pseudo-threshold for CNOT by upper bounding 
\begin{equation}
  p_\text{cnot}^{(1)} \leq \Pr[\bad \vert \accept] + \sum_{k} \big( \Pr[\malig_X, K_X=k, \good] + \Pr[\malig_Z, K_Z=k, \good] \big) \enspace ,
\end{equation}
where $\malig_X \equiv (\malig_{IX} \vee \malig_{XI} \vee \malig_{XX})$, $\malig_Z \equiv (\malig_{IZ} \vee \malig_{ZI} \vee \malig_{ZZ})$ and $\accept$ is the event that all $X$-error and $Z$-error verifications in the CNOT exRec succeed.  
  
The pseudo-threshold is is of practical interest for cases in which a finite failure probability is acceptable and only a few levels of concatentation are desired.  For example, when the physical failure rate is sufficiently below the pseudo-threshold, the Golay code could be used to bootstrap into other codes with lower overhead.

The pseudo-threshold is useful to us for two reasons.  First, pseudo-threshold estimates have been calculated for a variety of fault-tolerant quantum circuits including circuits based on the Golay code~\cite{CrossDiVincenzoTerhal}, and therefore serve as a reference for our counting results.  Second, it was conjectured by~\cite{SvoreCrossChuangAho} that the pseudo-threshold is an upper bound on the asymptotic threshold.  It thus provides a reasonable target for our calculation of the asymptotic threshold lower bound, which requires a noise strength maximum to be specified (see, in particular, \appref{sec:bounding-polys}).

Pseudo-threshold results are listed in~\tabref{tab:threshold-results} and discussed in detail in~\secref{sec:results-threshold}.

\subsubsection{Asymptotic threshold analysis}
\label{sec:asymptotic-threshold}

The asymptotic noise threshold is defined as the largest value $\gammaTH$ such that, for all $\gamma < \gammaTH$, the probability that the fault-tolerant simulation succeeds can be made arbitrarily close to one by using sufficiently many levels of code concatenation.  To prove a lower bound on the threshold we must show, in particular, that the probability of an incorrect CNOT $k$-Rec decreases monotonically with~$k$ for all $\gamma < \gammaTH$.  Our counting technique gives an upper bound on the probability that a CNOT $1$-Rec is incorrect.  We now show how to upper bound incorrectness for level-two and higher and therefore lower bound $\gammaTH$.

Consider an isolated level-one CNOT exRec. Let $\Pr[\malig_E]$ be the probability that the malignant event $\malig_E$ occurs.  For this event, the enclosed $1$-Rec behaves as an encoded CNOT gate followed by a two-block error that, when ideally decoded, leaves a two-qubit error~$E$ on the decoded state. Then our counting technique provides upper bounds on $\Pr[\malig_E]$ for $E \in \{IX, XI, XX, IZ, ZI, ZZ\}$.  These upper bounds can be viewed as an error model for the CNOT $1$-Rec in which the correlations between $X$ and $Z$ errors are unknown.

We would now like to analyze the level-two CNOT exRec.  Ideally, we could treat each $1$-Rec in the level-two simulation as a single ``location'' and use the error model obtained from level-one to describe the probability of failure.  Then level-two analysis could proceed by feeding this ``transformed'' error model back into the counting procedure in order to compute $\Pr[\malig_E]$ for the CNOT $2$-Rec.

\def\ecBox{\push{\fbox{\raisebox{0em}[1em][0.5em]{EC}}}}

\begin{figure}
\centering
\includegraphics[scale=1]{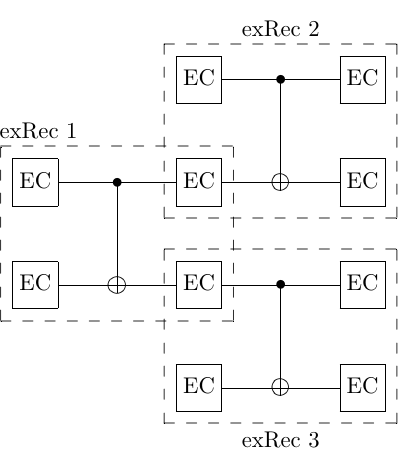}
\caption{Overlapping exRecs: exRec 1 shares one error correction with exRec 2 and one error correction with exRec 3.}
\label{fig:overlapping-exrecs}
\end{figure}

However, the transformed error model is based on analysis of an isolated level-one CNOT exRec.  A typical level-one simulation will contain many exRecs, and adjacent exRecs may share error corrections at which point they can no longer be considered independently.  For example, the CNOT exRec in \figref{fig:overlapping-exrecs} shares an error correction with both of the CNOT exRecs that follow it.  In~\cite{AliferisGottesmanPreskill05} (see also~\cite{Aliferis07thesis}) this problem is solved by the following procedure known as \emph{level reduction}: 

\begin{enumerate}
  \item Examine exRec $2$.  If the enclosed rectangle is incorrect then replace the entire \emph{exRec} with a faulty version of the associated (level-zero) gate.  Otherwise, replace the \emph{rectangle} with an ideal version of the associated gate.  
  \item Examine exRec $3$.  Follow the same procedure as for exRec $2$.  
  \item Examine exRec $1$.  Depending on the outcomes of exRec $2$ and exRec $3$, one or both of the TECs may have been removed.  The enclosed rectangle now consists of the encoded CNOT and any remaining TECs.  If the remains of rectangle $1$ are incorrect, exRec $1$ is replaced with a faulty level-zero gate.  Otherwise, the rectangle is replaced with an ideal level-zero gate.  
\end{enumerate}

Level reduction allows the level-two analysis to proceed by treating each $1$-Rec as a single \emph{independent} location.  The probability that a ``location'' fails in the level-two simulation is upper bounded by the probability that the corresponding $1$-Rec is incorrect. The reason that level reduction works when counting sets of malignant locations is because exRecs with incorrect rectangles are replaced with faulty gates in the same way regardless of the malignant event that actually occurs. The quantity used to bound incorrectness probability is strictly non-increasing as locations (i.e., TECs) are removed.  To see this, consider sets of exRec locations of size $k$ and denote the set of all such sets by~$S$. Let $M \subseteq S$ be those sets for which some combination of nontrivial errors at the $k$ locations causes the rectangle to be incorrect (i.e., the malignant sets).  The probability that the rectangle is incorrect due to failures at exactly $k$ locations is then no more than $\abs{M} p^k$.  If an error correction is removed from the exRec, some of the sets in $M$ now contain fewer than $k$ exRec locations. The remaining sets with $k$ exRec locations are those that do not contain a location in the removed error correction.  The number of such sets is at most $\abs{M}$ and so the original bound on the incorrectness probability still holds.  

\def\X{\eventFont{X}}
\def\I{\eventFont{I}}

The disadvantage to this approach for non-adversarial noise models is that it fails to consider all of the available information. In particular, for a fixed set of malignant locations it assumes the worst-case error for each location. The probability that a given set of $k$ locations is actually malignant can be significantly less than $p^k$.  To obtain a more accurate analysis of the second level, we would like to replace each incorrect $1$-Rec according to the malignant event that has actually occurred. 

Our transformed noise model of an isolated CNOT exRec provides upper bounds on the probability of each type of malignant event, but we must show that that the bounds still hold when exRecs overlap.  Unfortunately, the bounds almost certainly will \emph{not} hold.  Consider, for example, the control block of the CNOT exRec, shown in \figref{fig:cnot-block}. Assume that the error immediately preceding the transversal CNOT is correctable (the error itself is not important). Let $\X$ be the event that an uncorrectable $X$ error exists on the output of the TEC and $\I$ be the event that the error on the output is correctable.  In other words $\X \equiv (\malig_{XI} \lor \malig_{XX})$ and $\I \equiv \lnot \X$.  Then define $\X' \equiv \lnot \I'$ as the event that an uncorrectable $X$ error exists on the block following the transversal CNOT but before error correction.  $\Pr[\malig_{XI}]$ will be non-increasing when removing the trailing error correction only if $\Pr[\X'] \leq \Pr[\X]$.  On the other hand, $\Pr[\malig_{IX}]$ will be non-increasing only if $\Pr[\I'] \leq \Pr[\I]$.  Since $\Pr[\X]+\Pr[\I]=\Pr[\X']+\Pr[\I']=1$, both conditions are satisfied only if $\Pr[\X]=\Pr[\X']$ and $\Pr[\I]=\Pr[\I']$, which of course is highly unlikely.  

\begin{figure}
\centering
\includegraphics[scale=1]{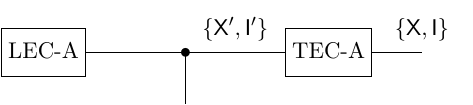}
\caption{Upper block of the CNOT exRec.  The error at the output of the TEC is either correctable ($\I$), or not ($\X$).  Similarly the error immediately preceding the TEC is either correctable ($\I'$) or not ($\X'$).}
\label{fig:cnot-block}
\end{figure}

In order to ensure a proper upper bound on each of the malignant event probabilities, we must calculate upper bounds for the complete exRec and for incomplete exRecs in which one or more trailing error corrections have been removed.  Calculations for the complete exRec were discussed in \secref{sec:Counting}.  Calculations for the incomplete exRecs are the same except that some of the TEC components are not considered.  Bounding the malignant event probability is a matter of finding a polynomial that bounds all four cases (see \appref{sec:bounding-polys}).  

Once proper bounds on the level-one malignant event probabilities are determined, we would like to plug the transformed error model into our counting procedure in order to determine the level-two error probabilities.  There are a few things to consider before doing so.  First, part of the counting strategy relies on using the correlations between $X$ and $Z$ errors in order to make corrections for over-counting that occurs during postselection.  The transformed error model, however, contains no such correlation information, so these corrections must be omitted.  Second, the CNOT malignant event upper bounds do not contain information about rest, preparation or measurement locations.  Level-one error models for these locations can be computed using the same counting strategy as the CNOT, but with an appropriately modified exRec.\footnote{Alternatively, they can be incorporated into the CNOT exRecs~\cite{AliferisCross07subsystem}.}  

Finally, in the depolarizing noise model, the error probabilities of each location are constant multiples of the noise strength~$\gamma$.  Our upper bounds on the malignant event probabilities, however, need not have any scalar relationship. For computer analysis, error probabilities must be re-normalized in terms of $\gamma$ and error weights recalculated as follows. Let $\mathcal{P}^{(1)}_E$ be our upper bound on the level-one malignant event $\malig_E$.  Then construct a polynomial $\Gamma^{(1)}$ and choose constants $\alpha_E$ such that 
\begin{equation}
  \mathcal{P}^{(1)}_E(\gamma) \leq \alpha_E \Gamma^{(1)}(\gamma)
\end{equation}
for all $E$.  The polynomial $\Gamma^{(1)}$ can be viewed as an effective noise strength ``reference'' for level-one.  $\Gamma^{(1)}(\gamma)$ is a function of $\gamma$, but we will usually denote it as $\Gamma^{(1)}$ for convenience of notation. Together with weights $\alpha_E$, $\Gamma^{(1)}$ defines a noise model similar to the depolarizing noise model defined in \secref{sec:NoiseModel}. See \appref{sec:transformed-noise-model} for details of the construction.  

Now the new error model is input into the counting procedure and upper bounds on the level-two error rates are computed.  Let $\mathcal{P}^{(2)}_E(\Gamma)$ be the upper bound computed for $\malig_E$ at level-two.  Then we have the following conditions on the level-one and level-two malignant event probabilities: 
\begin{equation}\begin{split}
  \Pr[\malig_E^{(1)}] &\leq \mathcal{P}^{(1)}_E(\gamma) \leq \alpha_E \Gamma^{(1)} \\
  \Pr[\malig_E^{(2)}] &\leq \mathcal{P}^{(2)}_E(\Gamma^{(1)})
  \enspace .
\end{split}\end{equation}
We also claim that $\mathcal{P}^{(2)}_E$ obeys the following property: 

\begin{claim}
  For $0 \leq \epsilon \leq 1$, $\mathcal{P}^{(2)}_E (\epsilon \Gamma^{(1)}(\gamma)) \leq \epsilon^4 \mathcal{P}^{(2)}_E (\Gamma^{(1)}(\gamma))$.
  \label{clm:P2-multiplicative}
\end{claim}
Proof of this claim is based on the form of the polynomials constructed by our counting technique and the fact that our circuits are strictly fault-tolerant.  Details of the proof are delegated to \appref{sec:P2-multiplicative-proof}.

We are now in a position to establish conditions for a noise threshold, i.e., the conditions under which the probability of a successful simulation can be made arbitrarily close to one. 

\begin{theorem}
 Let $M$ be the set of all level-one CNOT, preparation, measurement and rest malignant events consisting of: $\malig_{IX}$, $\malig_{XI}$, $ \malig_{XX}$, $\malig_{IZ}$, $\malig_{ZI}$, $\malig_{ZZ}$, $\malig_{X}^{\text{prep}}$, $\malig_{Z}^{\text{prep}}$,$\malig_{X}^{\text{meas}}$, $\malig_{Z}^{\text{meas}}$, $\malig_{X}^{\text{rest}}$ and $\malig_{Z}^{\text{rest}}$. Also let $\mathcal{P}^{(1)}_E$, $\mathcal{P}^{(2)}_E$ and $\Gamma^{(1)}$ be polynomials and $\alpha_E$ constants as discussed above.
 Then the tolerable noise threshold for depolarizing noise is lower bounded by the largest value $\gammaTH$ such that
 \begin{equation}
   \mathcal{P}^{(2)}_E(\Gamma^{(1)}(\gammaTH)) \leq \alpha_E \Gamma^{(1)}(\gammaTH)
 \end{equation}
 for all $\malig_E \in M$.  
 \label{thm:threshold}   
\end{theorem}

\proof{
Assume that $\mathcal{P}^{(2)}_E(\Gamma^{(1)}) < \alpha_E \Gamma^{(1)}$, for all $\malig_E$ and $\gamma \in (0,\gammaTH)$.  Then, for a fixed $\gamma \in [0,\gammaTH)$, there exists some positive $\epsilon < 1$ such that, for all malignant events $\malig_E$, $\mathcal{P}^{(2)}_E(\Gamma^{(1)}) \leq \epsilon \alpha_E \Gamma^{(1)}$.  

By choosing $\Gamma^{(2)} := \epsilon \Gamma^{(1)}$ we obtain an effective noise model for level two in which the weights~$\alpha_E$ are unchanged.  Since our counting method depends only on the error weights, the polynomials that upper bound the level-three malignant events will be the same as the polynomials that upper bound the level-two malignant events. That is, $\mathcal{P}^{(3)}_E(\Gamma) = \mathcal{P}^{(2)}_E(\Gamma)$.
Thus,
\begin{align}
  \Pr[\malig^{(3)}_E] \leq \mathcal{P}^{(3)}_E(\Gamma^{(2)}) 
  = \mathcal{P}^{(2)}_E(\epsilon \Gamma^{(1)})
  < \epsilon^5 \alpha_E \Gamma^{(1)}
  \enspace ,
\end{align}
where the last inequality follows from \claimref{clm:P2-multiplicative}.
 Repeating this process $k$ times yields
\begin{equation}
\label{eq:PrE-recurrence}
  \Pr[\malig_E^{(k+1)}] 
  \leq \mathcal{P}^{(k+1)}_E(\Gamma^{(k)})
  < \epsilon^{4k-3} \alpha_E \Gamma^{(1)}
  \enspace ,
\end{equation}
which approaches zero in the limit of large $k$.  
}

Testing of the assumption $\mathcal{P}^{(2)}_E(\Gamma^{(1)}) < \alpha_E \Gamma^{(1)}$ over a fixed interval $(0,\gammaTH)$ is straightforward because, as discussed in \appref{sec:monotonicity}, all of our malignant event polynomials (including $\Gamma^{(1)}$) are monotone non-decreasing up to sufficiently large values of $\gamma$.

\subsection{Results: Threshold lower bounds} \label{sec:results-threshold}

Threshold results were obtained by implementing our counting technique as a collection of modules written in Python and C; the source code is available at~\cite{Paetznick}.  Rigorous threshold lower bounds for both of our four-ancilla preparation and verification circuits are given in~\tabref{tab:threshold-results}.  The main program takes as input the four-ancilla preparation circuits, the noise model, and the good and bad event settings.  It outputs, for each type of exRec and each malignant event, a polynomial representing an upper bound on the event probability.  See \figref{fig:malig-events}.  These polynomials are either evaluated directly to calculate the pseudo-threshold, or processed into a transformed error model and fed back into the program.

The Python modules are broken up according to the components described in \secref{sec:Counting}.  The main task for each component is, for each error equivalence class, to compute a weighted count of location sets that produce that error.  Counts for each component are obtained by first computing counts for all of its sub-components and then convolving the results.  Details are discussed in \appref{sec:computer-analysis-details}.

The most time-consuming part of the computation involves the CNOT exRec component.  Computing weighted counts for this component required a custom convolution with nearly four trillion combinations.  This part of the program was written in C to save time.  Even so, running the entire program to completion for a fixed ancilla preparation and verification schedule on $31$ cores in parallel took about four days.  

\begin{figure}
\centering
\subfigure[$X$-error malignant events]{
\includegraphics[width=6.79cm]{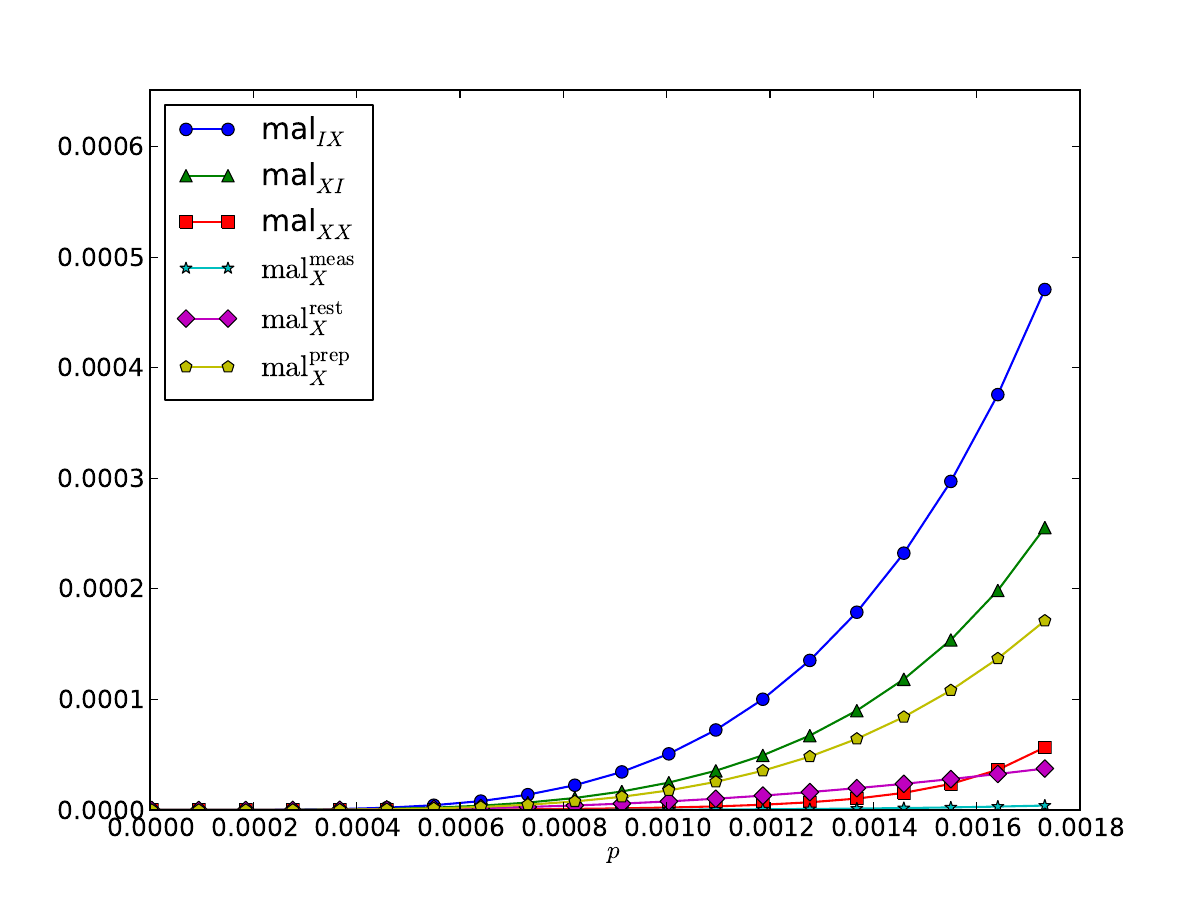}
\label{fig:malig-events-x}
}
\subfigure[$Z$-error malignant events]{
\includegraphics[width=6.79cm]{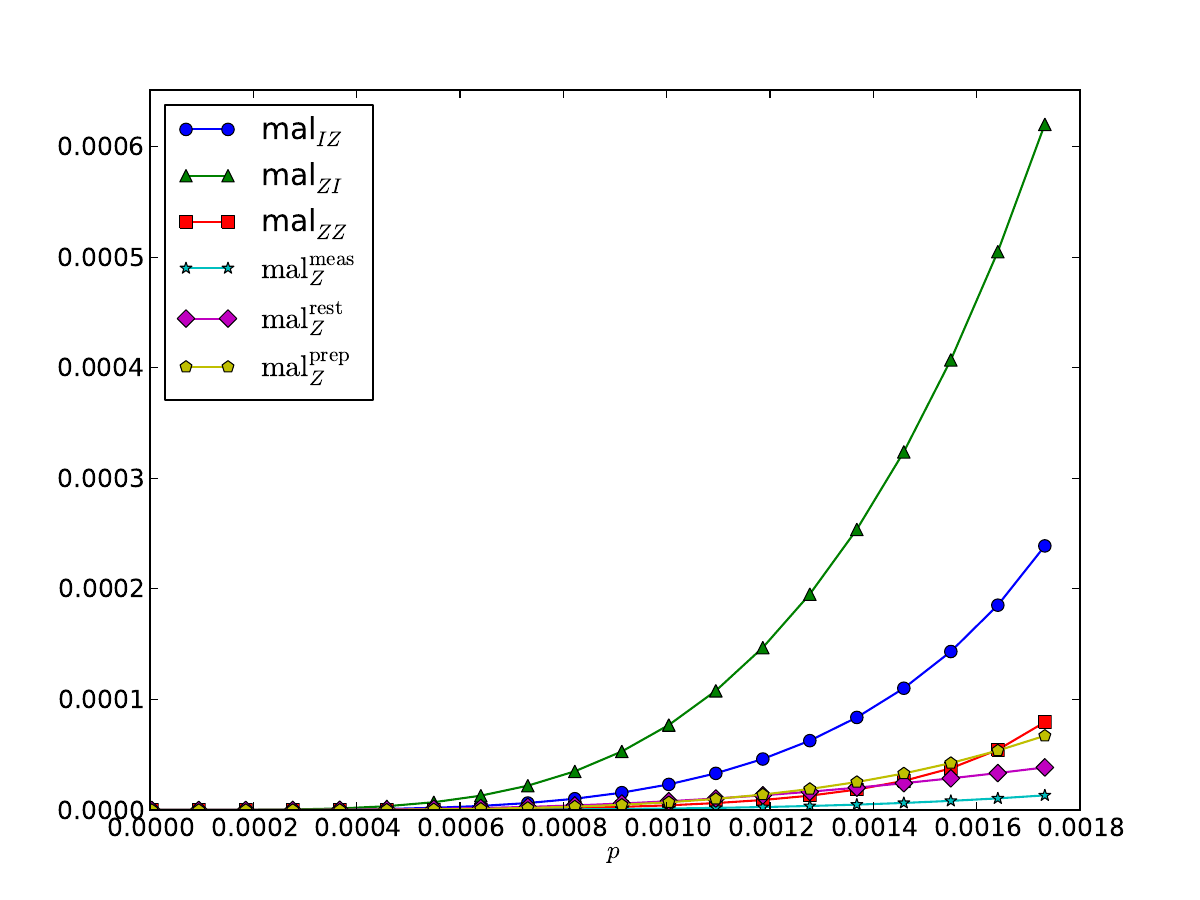}
\label{fig:malig-events-z}
}
\caption{These plots show upper bounds on probability of malignant events for the different level-one exRecs.  The $\malig_{IX}$, $\malig_{XI}$, $\malig_{XX}$, $\malig_{IZ}$, $\malig_{ZI}$ and $\malig_{ZZ}$ events all pertain to the CNOT exRec; the $\malig^{\text{prep}}_X$ and $\malig^{\text{prep}}_Z$ events correspond to the $\ket 0$ and $\ket +$ preparation exRecs, respectively; $\malig^{\text{meas}}_X$ and $\malig^{\text{meas}}_Z$ correspond to $Z$-basis and $X$-basis measurement exRecs; $\malig^{\text{rest}}_X$ and $\malig^{\text{rest}}_Z$ pertain to the rest exRecs.  Note that the upper bound on $\malig_{ZI}$ is significantly higher than that of its dual counterpart $\malig_{IX}$.  This is due largely to the arbitrary choice in error correction to correct $Z$ errors first and $X$ errors second.}
\label{fig:malig-events}
\end{figure}

\begin{table}
\centering
\begin{tabular}{c|c|c}
\hline\hline
Verification schedule & CNOT Pseudothreshold & Threshold \tabularnewline
\hline
Steane-$4$  & $1.72 \times 10^{-3}$ & $\threshSteane$ \tabularnewline
Overlap-$4$ & $1.73 \times 10^{-3}$ & $\threshOverlap$ \\
\hline \hline
\end{tabular}
\caption{Threshold lower bounds for circuits based on our four-ancilla preparation and verification schedules for the Golay code. Thresholds are given with respect to $p$ the probability that a physical CNOT gate fails, according to the depolarizing noise model defined in \secref{sec:NoiseModel}}.
\label{tab:threshold-results}
\end{table}

Our thresholds compare favorably to threshold results for similar circuits.  For a six-ancilla preparation and verification circuit, Aliferis and Cross~\cite{AliferisCross07subsystem} give a threshold estimate based on malignant set sampling of $p \approx 1 \times 10^{-4}$ for adversarial noise.  Our results beat this by an order of magnitude and provide strong evidence that our counting technique is an improvement over malignant set sampling and malignant set counting for the case of depolarizing noise.  
Our results also essentially close the gap with other analytical and Monte Carlo threshold estimates for depolarizing noise.  Using a closed form analysis, Steane~\cite{Steane03} estimated a threshold on the order of $10^{-3}$ for the Golay code with similar noise parameters.  Cross et al.~\cite{CrossDiVincenzoTerhal} estimated a pseudo-threshold of $2.25 \times 10^{-3}$ based on Monte Carlo simulations of a twelve-ancilla preparation and verification circuit.

Beyond circuits based on the Golay code, our results may be the highest rigorous threshold lower bounds known.  Aliferis and Preskill~\cite{Aliferis2009} prove a lower bound of $p \geq 1.25 \times 10^{-3}$. Their analysis applies to teleportation-based gates due to Knill~\cite{Knill04NoisyDevices} in which Bell pairs encoded into an error correcting code $C_2$ are prepared by first encoding each qubit of the $C_2$ block into an error \emph{detecting} code $C_1$ and performing error detection and postselection after each step of the $C_2$ encoding. Our best threshold is only about $5$ percent better, but apply to circuits that usually require far less overhead (see~\secref{sec:overhead-upper-bounds}).  This implies only that in the depolarizing noise model our analysis is more accurate, and not that our schemes tolerate more noise.

The limiting factor on the threshold value is the event $\malig_{ZI}$.  That is, $\malig_{ZI}$ is the event $E$ for which $\Pr[\malig_E^{(2)}] = \Pr[\malig_E^{(1)}]$ takes the smallest value of $p$.  In fact, the corresponding threshold values for nearly all $Z$-error malignant events are lower than threshold values for \emph{any} of the $X$-error events.  This asymmetry is due to the arbitrary order with which we perform error correction---$Z$ first, then $X$.  Some $X$ errors resulting from the leading $Z$-error correction will be corrected by the $X$-error correction that follows.  However, $Z$ errors resulting from the $X$-error correction may propagate through the encoded operation before arriving at the $Z$-error correction on the trailing end.  As a result, it is more likely for $Z$ errors on individual blocks to be combined by the CNOT gate and create an uncorrectable error.  Evidence of this effect can be seen in the level-one malignant event probabilities shown in \figref{fig:malig-events}.  

It should be possible to reduce such lopsided event probabilities by customizing the error correction order for each EC based on the specifics of the ancilla preparation circuits.  However, analyzing such a scheme would require consideration of up to $36$ different full or partial CNOT exRecs (two choices for each EC) instead of four and is likely to yield only a small improvement in the threshold.  Note that other small improvements could be made by, for example, eliminating measurement or rest exRecs at level-two.  For simplicity, these optimizations were not considered.

\section{Resource overhead}\label{sec:Overhead}
To evaluate the practical importance of our optimizations, we now analyze the resource requirements of our fault tolerance scheme.  In~\secref{sec:overhead-ancilla-prep} we use Monte Carlo simulation to compare overhead of our ancilla preparation and verification circuits to that of standard circuits.  In~\secref{sec:overhead-upper-bounds} we use threshold results from~\secref{sec:threshold-analysis} to compute upper bounds on the resource overhead of our scheme as a whole.

\subsection{Simulation of the ancilla preparation overhead} \label{sec:overhead-ancilla-prep}
One natural measure for the overhead is the number of CNOT gates used to ready an ancilla.  Another overhead measure, important given the difficulty of scaling quantum computers, is the space complexity, i.e., the number of qubits that must be dedicated to ancilla preparation in a pipeline so that an ancilla is always ready in time for error correction.  We consider both measures.  

\begin{figure}
\centering
\subfigure[\label{fig:sim-cnot-overhead}]{
\includegraphics[width=6.79cm]{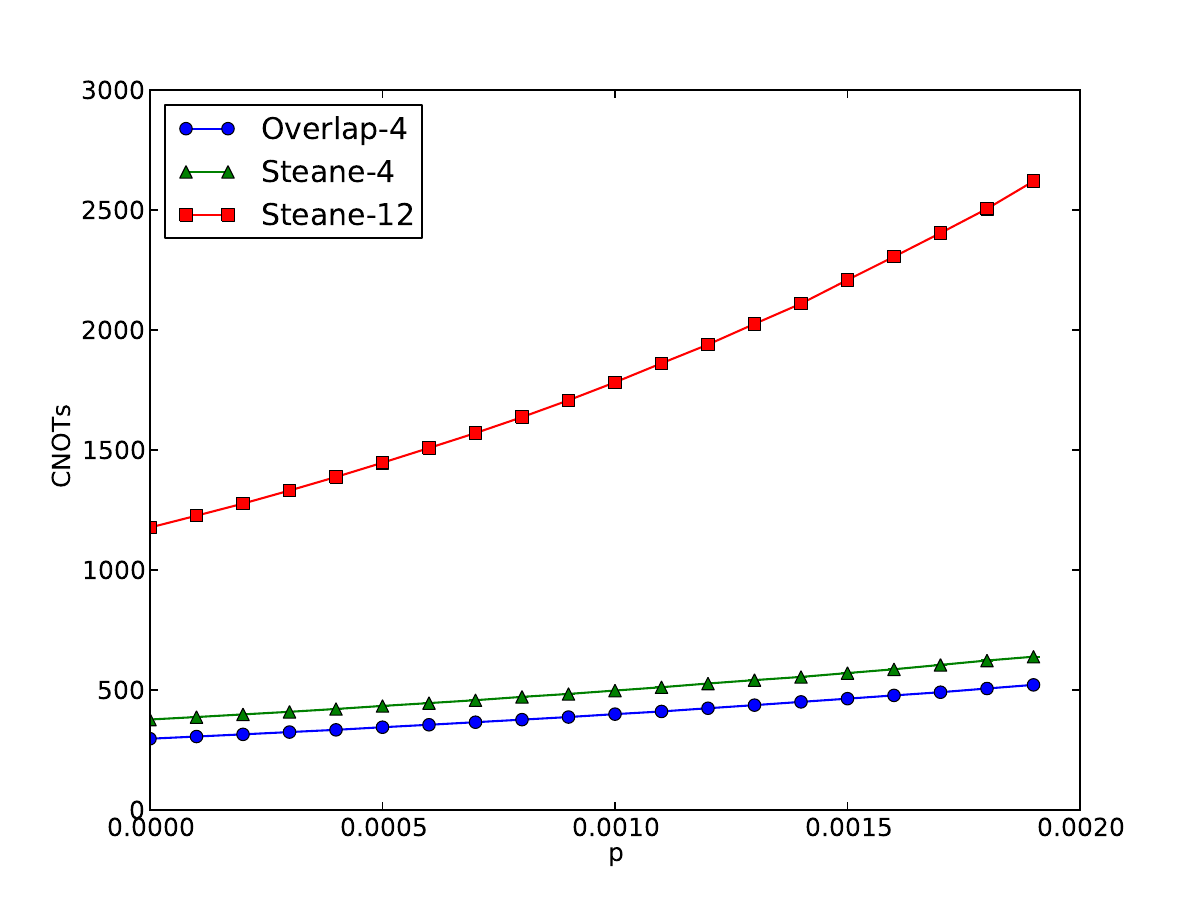}
}
\subfigure[\label{fig:sim-qubit-overhead}]{
\includegraphics[width=6.79cm]{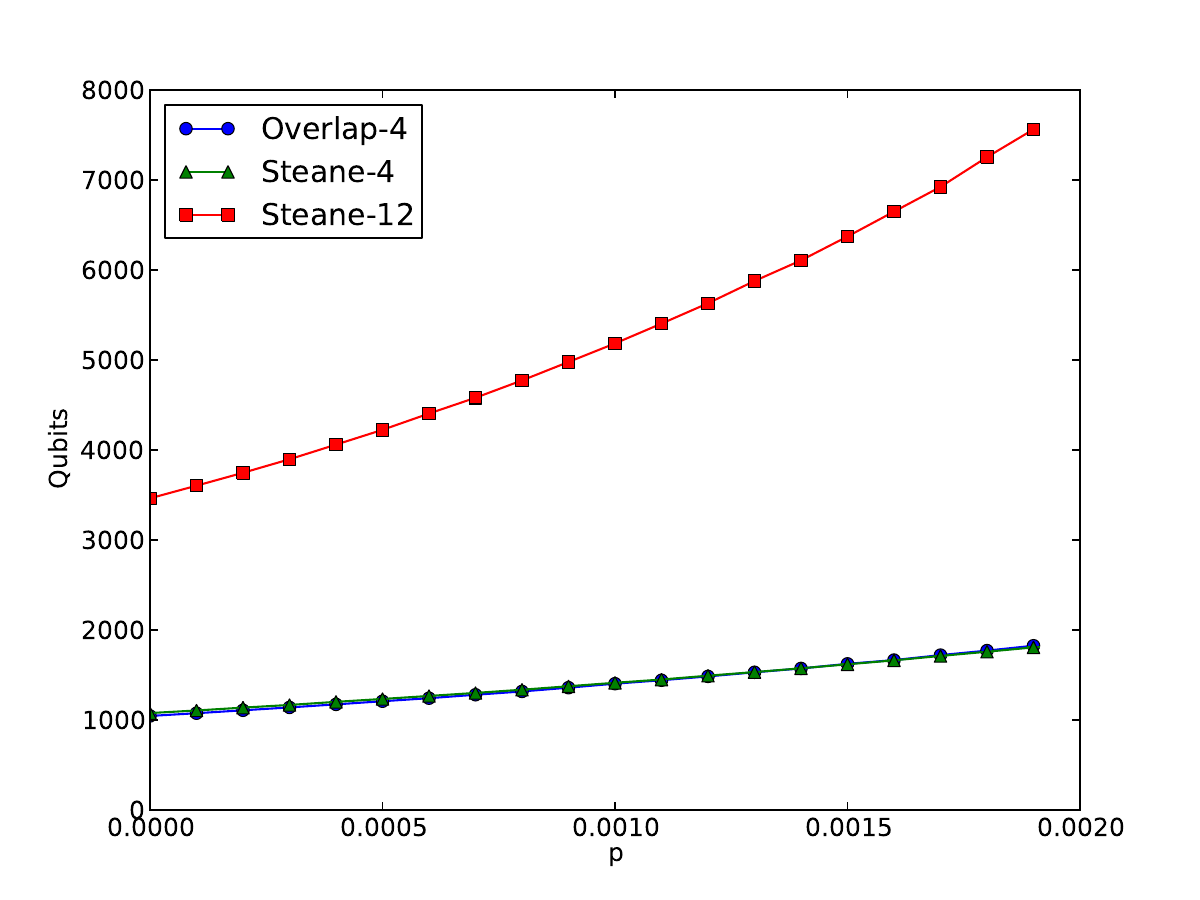}
}
\caption{Overhead estimates for the twelve-ancilla ancilla preparation and verification circuit and for each of our optimized circuits.  (a) Expected number of CNOT gates required to produce a verified encoded $\ket 0$.  (b) Number of qubits required to produce one verified encoded $\ket{0}$, in expectation, at every time step.  
Standard error intervals are too small to be seen here.}
\end{figure}

As listed in the third column of \tabref{tbl:verify-compare}, the overlap method-based four-ancilla preparation and verification circuit involves roughly a factor of four fewer CNOT gates than the standard twelve-ancilla circuit.  In fact, this understates the improvement.  The overhead also depends on the acceptance rates of each verification test.  For an ancilla to leave the twelve-ancilla circuit, it must pass eleven tests, compared to only three tests for the four-ancilla circuit.  The probability of passing all tests should be significantly higher for the optimized circuit, and so one expects the ratio between the \emph{expected} numbers of CNOT gates used by the two circuits to be greater than four.  

To estimate the expected overhead, each circuit was modeled and subjected to depolarizing noise in a Monte Carlo computer simulation.  We assumed that test results are available soon enough that a failed verification circuit can be immediately aborted; later test failures are therefore the most costly.  This assumption impacts the twelve-ancilla circuits the most, since there are many ways to construct the hierarchy of verifications.  The circuit shown in~\figref{fig:TwelveAncillaVerifyCkt} is a reasonable choice here because only six of the verification tests depend on results of previous tests.  Other circuits---see, e.g., \cite[Sec.~2.3.2]{ReichardtThesis06}---may contain as many as nine dependent tests.

Estimates of the expected number of CNOT gates required for each circuit are given in the last column of \tabref{tbl:verify-compare} for the CNOT depolarization rate $p = 10^{-3}$, and are plotted versus~$p$ in \figref{fig:sim-cnot-overhead}.  At $p = 10^{-3}$, the overlap method reduces the expected number of CNOT gates by roughly a factor of~$4.5$, compared to the twelve-ancilla circuit, and the improvement for our optimized Latin rectangle scheme is a factor of~$3.6$.  At lower error rates, the improvement is less.  To investigate the effects of different error parameters, we also considered setting the rest error rate to zero; in this case, the expected number of CNOT gates used in the overlap circuit further decreases by about~$11$ percent, compared to less than four percent for our other four-ancilla circuit and less than two percent for the twelve-ancilla circuit.  The larger improvement for the overlap circuit is due primarily to the fact that the overlap preparation method replaces many CNOT gates with rest locations. 

To evaluate the space overhead, we plot in \figref{fig:sim-qubit-overhead} the number of qubits required to produce a single verified encoded $\ket 0$, in expectation, per time step, for each of the preparation and verification circuits. Thus, for example, the space overhead for a pipeline to produce a single \emph{unverified} ancilla state is $8 \cdot 23 = 184$ qubits; at any given time step, one $23$-qubit block is initialized, and CNOT gates are applied to seven other blocks---one per round in, e.g., \figref{fig:overlap-prep-ckt}---so that one ancilla is prepared. (In fact, the overhead is slightly less than this since some of the qubits in the block can be prepared during rounds one and two.)  Estimates are calculated recursively by computing E[qubits] = (E[qubits]$_1$ + E[qubits]$_2$)/Pr[accept] for each verification step where the numerator is the expected number of qubits required required to prepare the two states used in that verification step and Pr[accept] is the probability that the verification measurement detects no errors.  The results at $p = 10^{-3}$ are given in the second column of \tabref{tbl:verify-compare}.  Both of our optimized schemes reduce the required space by a factor of~$3.6$ at~$p = 10^{-3}$.  

To judge the significance of these results, recall that the ancilla production pipeline can consume the majority of resources in a fault-tolerant quantum computer.  In the case of~\cite{IsalovicNemanjaPatelKubiatowicz08}, physical ancilla production space is proportional to the number of CNOT gates in the pipeline.  A factor of~$4.5$ reduction in the CNOT overhead for ancilla preparation should give, very roughly, about a $50$~percent improvement in the total footprint of the quantum computer.

\subsection{Resource overhead upper bounds} \label{sec:overhead-upper-bounds}

The calculations above provide an estimate of the improvement of our ancilla preparation schemes compared to the standard method.  We would also like to understand the overhead for our scheme as a whole.  In this section we calculate upper bounds on the number of physical gates and the number of physical qubits required to implement a single logical gate with a given effective error rate.

Our threshold analysis assumes that an infinite supply of ancilla qubits is available for use in error correction.  In order to bound the resource overhead we instead assume that some finite number of ancillas are available to each $k$-EC. Error correction proceeds normally unless all ancilla verifications fail. If the number of available ancillas is high enough, then the probability that all verifications fail will be small and the impact on the logical errors will be similarly small.

More precisely, our approach is as follows.  The ancilla verification circuit (\figref{fig:ancilla-verification}) is considered as a single unit.  Each level-$k$ $Z$-error correction consists of $m_k$ $\ket 0$ verifications performed in parallel plus a transversal rest, CNOT and $X$-basis measurement.  If all of the $m_k$ verifications fail, then $Z$-error correction is aborted and the data is left idle.  Level-$k$ $X$-error correction is similar.  For simplicity, if any of the error corrections are aborted, then we consider the entire top-level logical gate to have failed. 


We need to bound the number of ancilla verifications for each level of concatenation.  First, bound $\Pr[\reject]^{(k)}$ the probability of rejecting at level $k$.
\begin{equation}
\begin{split}
  \Pr[\reject](\Pevent_1^{(k)}, \Pevent_2^{(k)}, \ldots) 
  &\leq \Pr[\reject](\epsilon^{4(k-2)-3} \alpha_1 \Gamma, \ldots) \\
  &\leq \epsilon^{4(k-2)-3} \Pr[\reject]^{(2)}
\end{split}
\end{equation}
The first inequality follows from~\eqnref{eq:PrE-recurrence}. The second follows from monotonicity and the form of $\Pr[\accept]$ (\appref{sec:monotonicity}).

Let $p_\text{target}$ be overall target error rate per logical gate, $\Pevent^{(k)} := \max_i \Pevent_i^{(k)}$, and let $K$ be the minimum level of concatenation that acheives $\Pevent^{(k)} < p_\text{target}$ assuming an unbounded number of ancilla.  If $\Pr[\fail^{(k)}]$ is the probability that a $Z$-error correction fails then we require that
\begin{equation}
 4 \Pr[\fail^{(k)}] \leq 4 \Pr[\reject]^{m_k} + 16 m_k A_{\text{ECZ}} \Pr[\fail^{(k-1)}] \leq \delta^{(k)}
 \enspace ,
\end{equation}
where $\delta^{(k)} = p_\text{target} - \Pevent^{(k)}$ and $A_{\text{ECZ}}$ is the number of locations in a $Z$-error correction circuit.
The first term bounds the probability that all $\ket{0}$ verifications fail.  The second term bounds the probability that a $(k-1)$-EC fails in any of the $m_k$ verifications.  Multiplication by four on the left-hand side accounts for the four $X$- and $Z$-error corrections in a CNOT rectangle.  We must divide $\delta^{(k)}$ between the two terms.  We could try to optimize the division, but the overhead is dominated by the number of concatenation levels $K$ and so the division is not terribly important. Instead, we simply choose $\Pr[\reject^{(k)}]^{m_k} \leq \delta^{(k)} / 4k$.
Solving for $m_k$ we obtain,
\begin{equation}
 m_k \leq \frac{\log{(\delta^{(k)} / 4k)}}{\log{\Pr[\reject^{(k)}]}}
 \enspace .
\end{equation}
We may then compute $m_{k-1}$ using
\begin{equation}
 \delta^{(k-1)} = \frac{\delta^{(k)} - \delta^{(k)} / 4k}{16 m_k A_{\text{ECZ}}}
 \enspace .
\end{equation}
Once we have a bound on $m_k$ for every $k \leq K$, the total gate overhead $g(k)$ for a $k$-Rec can be computed recursively by $g(k) \leq (2 m_k A_\text{EC} + 23) \cdot g(k-1)$.

%
%

\begin{figure}
\label{fig:gate-overhead}
\centering
\subfigure[Golay scheme with overlap ancilla preparation]{
\label{fig:gate-overhead-overlap}
\includegraphics[width=6.79cm]{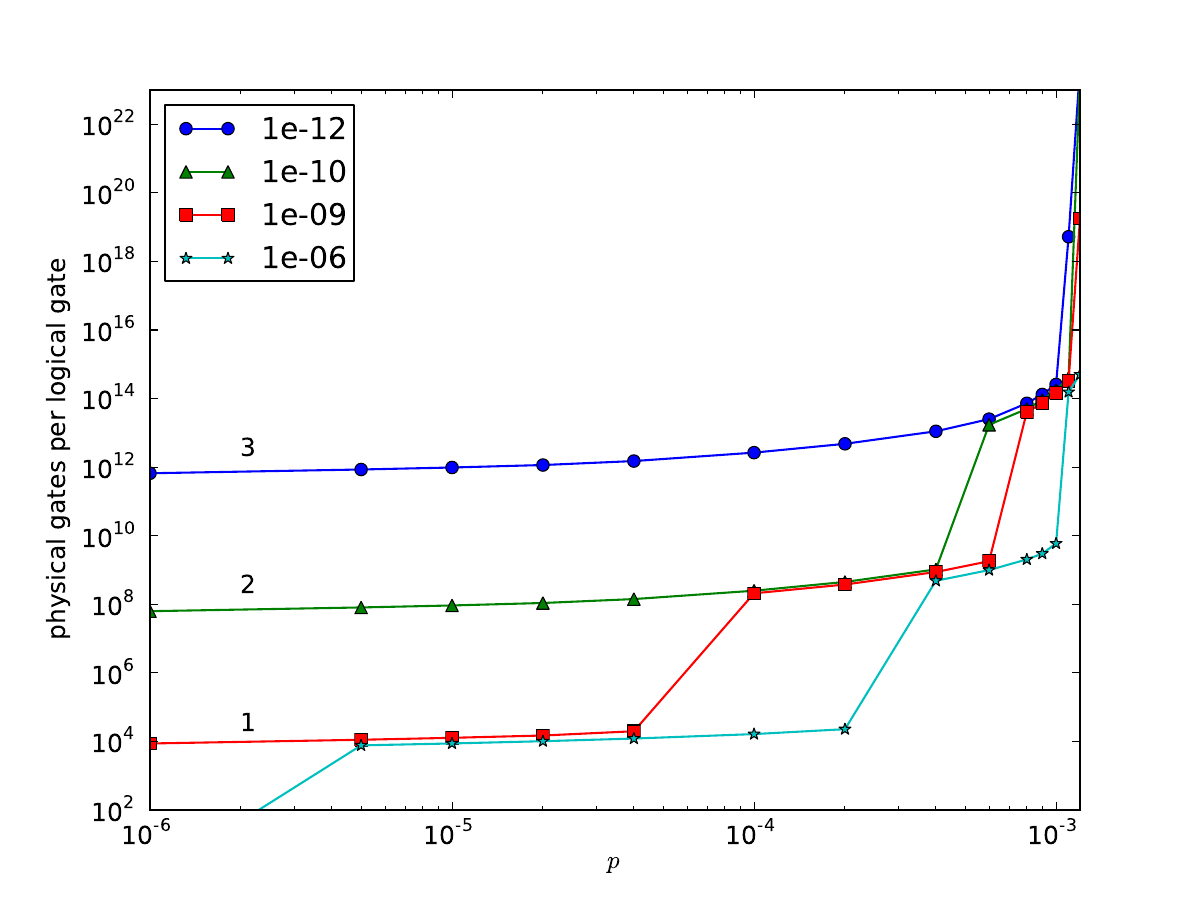}
}
\subfigure[{$[[4,2,2]]$} Fibonacci scheme]{
\label{fig:gate-overhead-fibonacci}
\includegraphics[width=6.79cm]{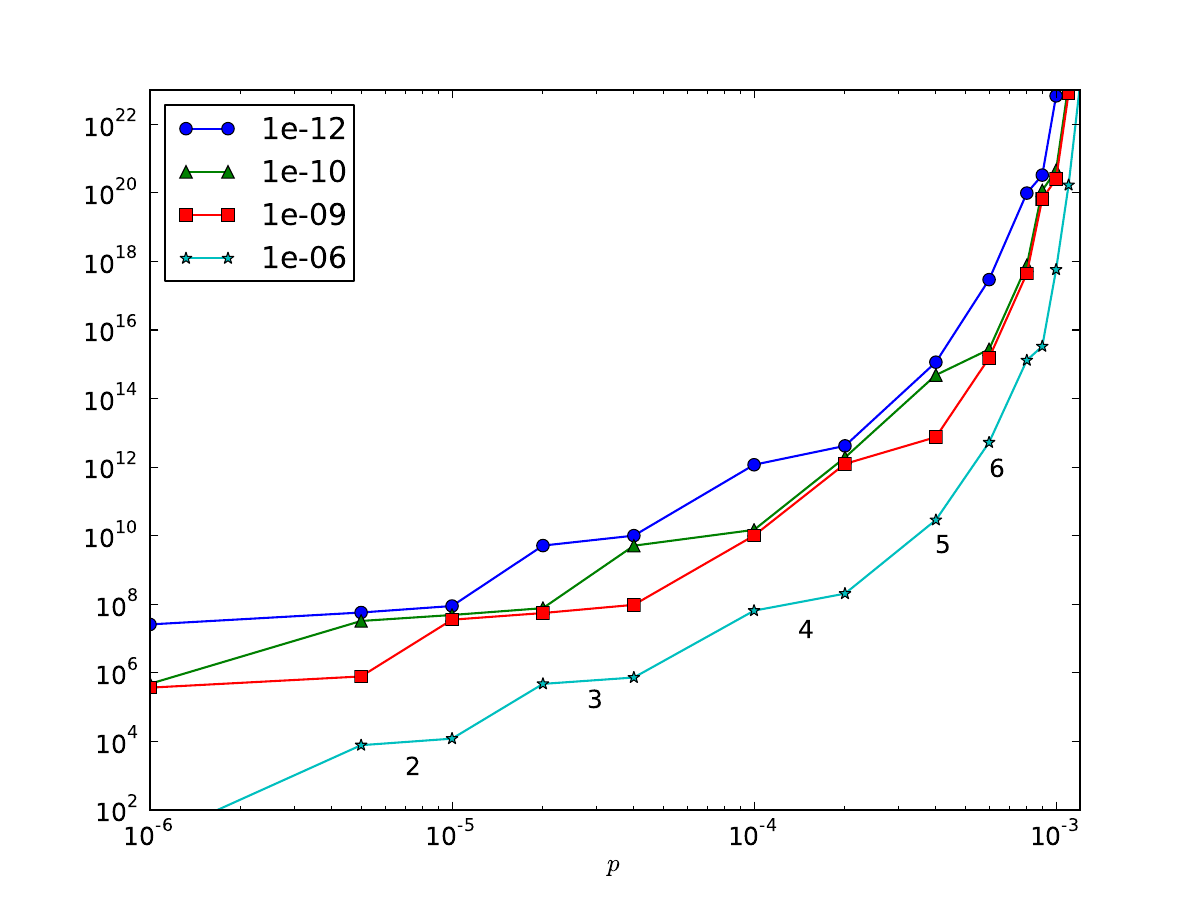}
}
\caption{\label{fig:gate-overhead-plot}
Gate overhead upper bounds for (a) our Golay scheme with overlap ancilla preparation and (b) the Fibonacci scheme presented in~\cite{Aliferis2009}. Each plot shows the number of physical gates required to implement a logical gate with target error rates $p_\text{target} \in \{ 10^{-12}, 10^{-10}, 10^{-9}, 10^{-6} \}$. Black text labels indicate the required level of concatenation and colored lines are a guide for the eye.}
\end{figure}

\begin{figure}
\centering
\subfigure[Golay scheme with overlap ancilla preparation] {
\includegraphics[width=6.79cm]{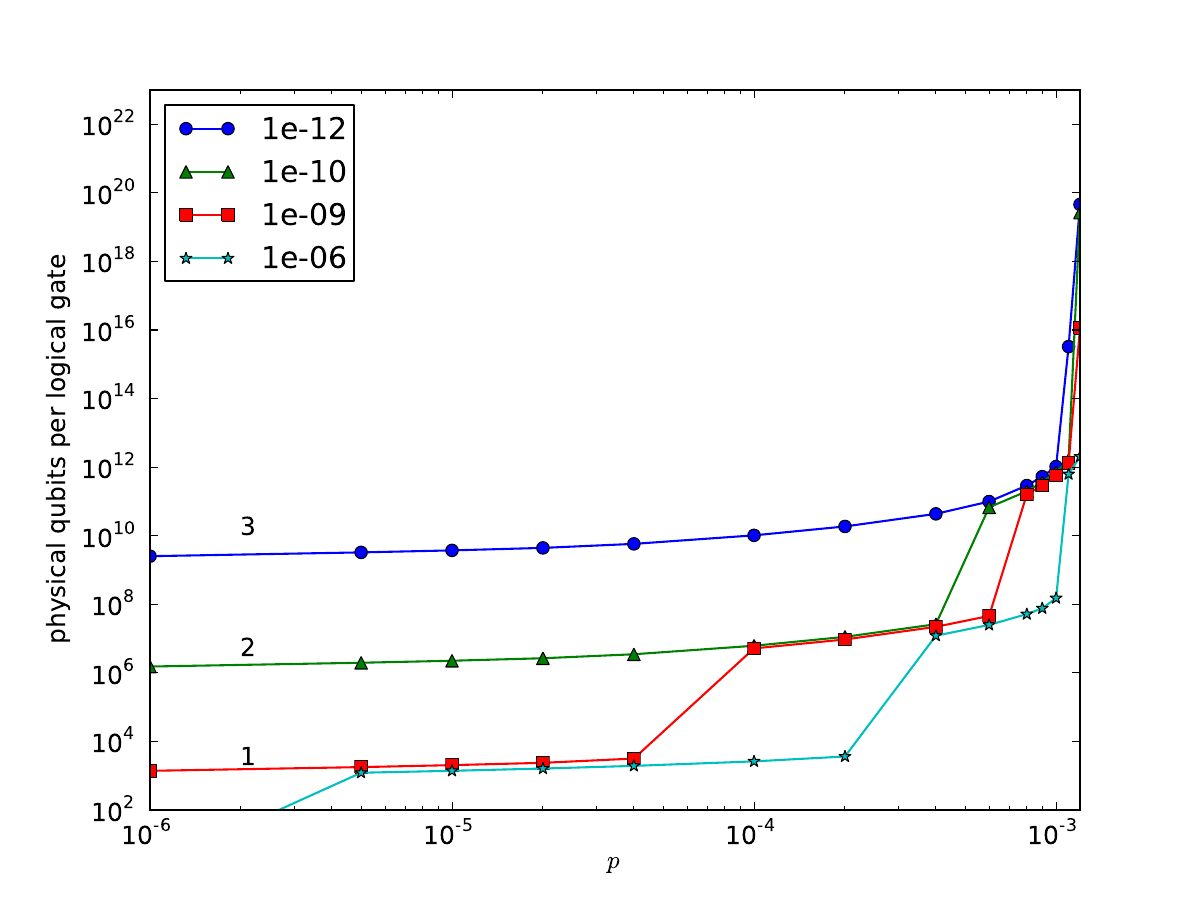}
}
\subfigure[{$[[4,2,2]]$} Fibonacci scheme] {
\includegraphics[width=6.79cm]{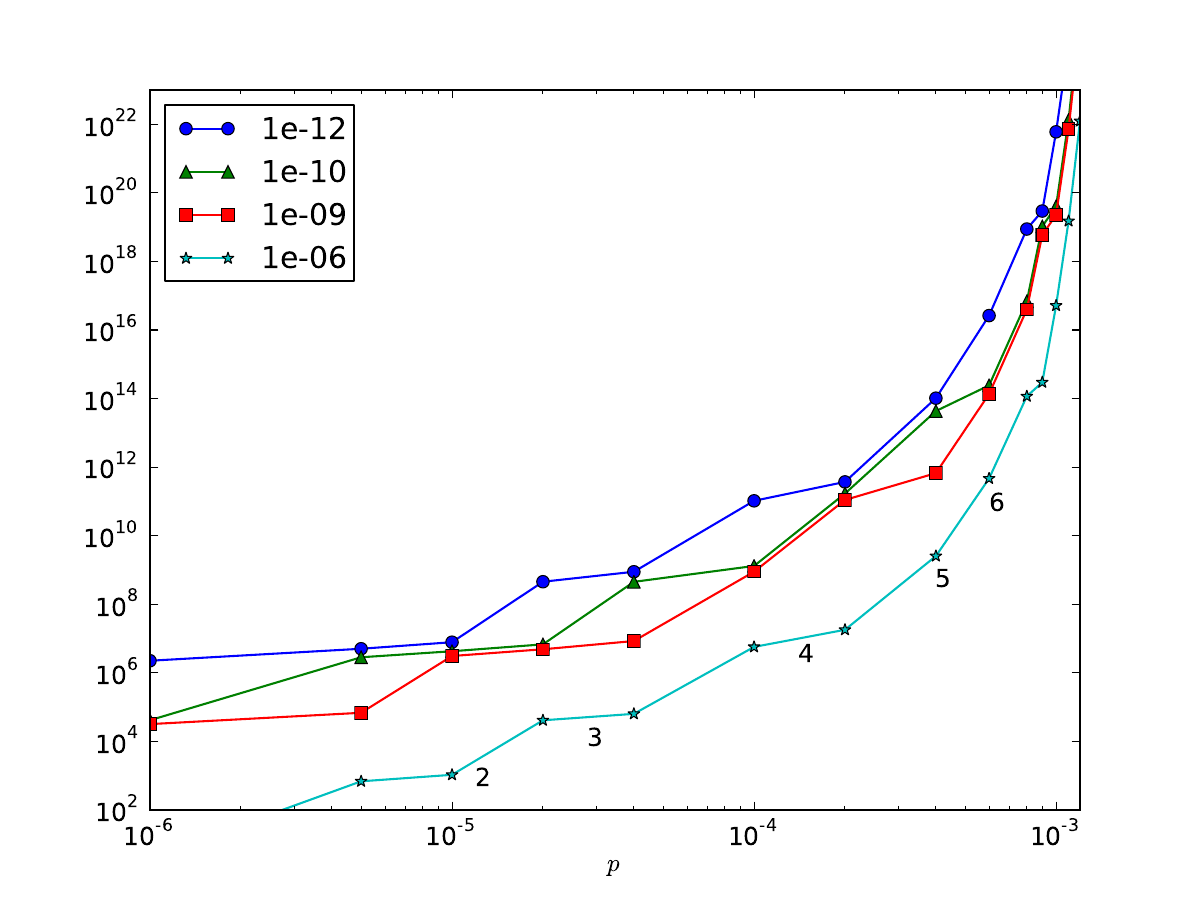}
}
\caption{\label{fig:qubit-overhead-plot}
Qubit overhead upper bounds.  Plots are formatted identically to~\figref{fig:gate-overhead-plot}.
}
\end{figure}

Gate overhead upper bounds for the overlap-based scheme are shown in~\figref{fig:gate-overhead-overlap}.  The overhead increases dramatically as the target logical error rate decreases.  However, compared to similar upper bounds for the Fibonacci scheme---which has a similar threshold lower bound~\cite{Aliferis2009}---our scheme is better for a wide range of error rates often by several orders of magnitude.  One reason for the improved overhead is that our scheme is based on a code with higher distance than the Fibonacci scheme which uses the $[[4,2,2]]$ error detecting code. The logical error rate for our Golay scheme falls faster and thus requires fewer levels of concatentation.

Bounds on qubit overhead may be obtained from the gate overhead.  Our threshold analysis requires that all ancillas  be ready on-demand without delay---i.e., each $k$-Rec has depth three, independent of $k$.  We, therefore, pessimistically assume that once a qubit is measured it cannot be re-used within the same rectangle.  The qubit overhead then depends only on the gate overhead and the qubit-gate ratio for $\ket 0$ verification.  Using a ratio of $8 \cdot 23 / (A_\text{EC} - 46)$ we obtain $q(k) \leq 23^k + 0.15^k g(k)$.  Therefore, the level-$k$ qubit overhead is roughly $k$ orders of magnitude lower than the level-$k$ gate overhead.

The qubit-gate ratio for Bell-state preparation in the Fibonacci scheme is relatively large ($\approx 0.6$ for levels three and above).  Therefore, similar to gate overhead, qubit overhead for the Golay scheme compares favorably to the Fibonacci scheme for a wide range of noise parameters. See~\figref{fig:qubit-overhead-plot}.

The drawback of using a larger code is that the increase in overhead from one level of concatenation to the next is much larger.  This makes it harder to ``tune'' the overhead parameters to some specific error rates.  For example, for $p_\text{target}=10^{-10}$ and $p = 10^{-6}$ our scheme requires two levels of concatenation and about $10^8$ physical gates per logical gate.  For the same error rates, the Fibonacci scheme requires three levels of concatenation, but fewer than $10^6$ gates.


Finally, note that bounds for our scheme when $p_\text{target}=10^{-12}$ are a bit loose due to a constant offset that is added during the transformed noise model construction (\appref{sec:bounding-polys}).  In our computer analysis, these offsets were on the order of $\epsilon \approx 10^{-13}$.  In principle, this offset does not affect the actual error rates; rather it is an artifact of our construction.

\section{Future work}

We have shown two alternative circuits for the fault-tolerant preparation of Golay encoded ancillas.  Our circuits require a total of only four encoded ancillas, and thus out outperform the previous best known circuits in terms of overhead. We have also demonstrated a new malignant set counting technique and threshold analysis tailored specifically for depolarizing noise.  With this technique, we have proved a tolerable noise threshold of $\threshOverlap$, which may be the highest rigorous threshold known.  

There are a number areas for future work.
First, using the overlap encoding method given in \secref{sec:Overlap}, we were able to reduce the number of CNOT gates when compared to the standard encoding procedure.  However, the circuit given in \figref{fig:overlap-prep-ckt} was constructed by hand and is not necessarily optimal.  It would be ideal have an algorithm for finding a preparation circuit with the fewest number of CNOT gates, while maintaining circuit depth bounds.  Second, our techniques for optimizing verification circuit overhead could be applied to other large codes.  For example, the self-dual BCH $[[47, 1, 11]]$ and concatenated Steane $[[49, 1, 9]]$ codes (see, e.g., \cite{GrasslBeth99}) are similar in nature to the Golay code and may yield similar results.  Variations of the verification procedure could also be analyzed.  Overhead might be reduced by, for example, attempting to correct some errors detected during verification rather than always scrapping the entire preparation.

Another possible avenue of interest is to apply our malignant set counting technique to other types of fault-tolerant error correction methods.  In particular, the teleportation-based schemes due to Knill are enticing candidates.  Simulations indicate that these schemes can tolerate a depolarizing noise rate as high as $p = 1\%$.
By dividing the exRec into small components, we may be able to count larger sets of faulty locations and obtain a tighter bound. However, teleportation-based schemes introduce new analytical challenges.  In particular, unlike error-correction gadgets, error-detection gadgets are non-deterministic. Our analysis assumed a trivial input syndrome to the exRec.  For non-deterministic exRecs, all possible input syndromes must be considered.  Additionally, error events must be conditioned on acceptance of \emph{all} error-detection gadgets, not just those inside of the exRec.

\subsection*{Acknowledgements}

We would like to thank Ashley Stephens for helpful feedback. Research conducted at the Institute for Quantum Computing, University of Waterloo, supported by Mitacs and NSERC.

\bibliographystyle{alpha-eprint}
\bibliography{bibliography}

\appendix

\section{Component counting}
\label{sec:counting-detail}

\subsection{Bounding bad events}
\label{sec:bounding-bad}

Bad fault events are defined, in part, by establishing some limit $k_\good$ on the number of failures~$K$ within the component.  In the $X$-error case, $\ket{+}$ preparation and $X$-basis measurement locations are ignored (they cannot produce $X$ errors).  For a component containing $n_c$ CNOT gates, $n_r$ rests, and $n_p + n_m = n_{pm}$ $\ket{0}$ preparations and $Z$-basis measurements, let $A_{\vec n}$ and $\beta_{\vec{n}}(\vec{k})$ and  be defined as 
\begin{equation}\begin{split}
  A_{\vec{n}} &= (1-12\gamma)^{n_c}(1-8\gamma)^{n_r}(1-4\gamma)^{n_{pm}} \\
  \beta_{\vec{n}}(\vec{k}) &= \binom{n_c}{k_c}\binom{n_r}{k_r}\binom{n_{pm}}{k_{pm}}  \enspace .
\end{split}
\label{eq:component-AB}
\end{equation}
Then the probability of more than $k_\good$ $X$ failures is
\begin{equation}\begin{split}
  \Pr[k_\good < K_X < k_\text{max}] 
  & = A_{\vec{n}} \!\!\!\!\! \sum_{k_\good < \abs{\vec{k}} < k_\text{max}}
    \!\!\!\!\! \beta_{\vec{n}}(\vec{k})
    \left( \frac{12\gamma}{1-12\gamma} \right)^{k_c} 
    \left( \frac{8\gamma}{1-8\gamma} \right)^{k_r}
    \left( \frac{4\gamma}{1-4\gamma} \right)^{k_{pm}} \\
  & \leq A_{\vec{n}}\sum_{k_\good < \abs{\vec{k}} < k_\text{max}} 
    \beta_{\vec{n}}(\vec{k})
    \left( \frac{\gamma}{1-12\gamma} \right)^{\abs{\vec k}}
    12^{k_c} 8^{k_r} 4^{k_{pm}}
  \enspace . 
\end{split}
\label{eq:PrBad}
\end{equation}
The sums are over all possible failure partitions $\vec{k}=(k_c,k_r,k_{pm})$ for which $\abs{\vec{k}}=k_c + k_r + k_{pm}$ is in the correct range.  The failure partition represents the number of CNOT failures  $k_c$, rest failures~$k_r$, and preparation and measurement failures $k_{pm} = k_p + k_m$. The upper bound $k_\text{max}$ is used to avoid double counting between components and sub-components, and may be omitted in which case take $k_\text{max} = n_c + n_r + n_{pm} + 1$.  For components with a large number of locations, it is convenient to approximate \eqnref{eq:PrBad} by evaluating the sum up to only a fixed number of failures $k' = k_\good + \text{const.}$ and then upper bounding the rest of the sum by
\begin{equation}
  \Pr[k' < K_X] 
  \le \sum_{\vec{k}: \, \abs{\vec{k}} = k' + 1}
    \beta_{\vec{n}}(\vec{k})
    (12\gamma)^{k_c}(8\gamma)^{k_r}(4\gamma)^{k_{pm}}
  \enspace .
\label{eq:prBad-bound-without-kmax}
\end{equation}  
This bound is easy to compute, so long as the number of sub-components is reasonable, and is better than the simpler bound $\binomial{n}{k'+1}(12\gamma)^{k'+1}$.  

The $Z$-error case is completely analogous except that now $\ket +$ and $X$-basis measurement locations are counted and $\ket 0$ preparation and $Z$-basis measurement locations are ignored.  All of the equations remain the same.

\subsection{\texorpdfstring{$X$}{X}-error verification} \label{sec:X-verification-details}

\begin{figure}
\centering
\leavevmode
\includegraphics[scale=1]{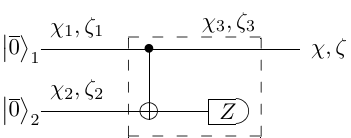}
\caption{$X$-error verification. The $X$-error verification component consists of two Golay encoded $\lket{0}$ preparations, a transversal CNOT and a transversal $Z$-basis measurement. $X$ ($\chi_i$) and $Z$ ($\zeta_i$) errors for each part are labeled.  The output is an $X$ error $\chi$ and $Z$ error $\zeta$.}
\label{fig:X-error-verification}
\end{figure}

The $X$-error verification component is illustrated in \figref{fig:X-error-verification}.  It includes three sub-components that fail independently.  Let $(\chi_1, \zeta_1)$ be the $X$ and $Z$ errors resulting from the first $\lket 0$ preparation and $(\chi_2, \zeta_2)$ be the errors from the second $\lket 0$ preparation.  Let $(\chi_3, \zeta_3)$ be the errors on the remaining transversal CNOT and $Z$-basis measurement locations; we will denote by $(\chi_3', \zeta_3')$ the portion of these errors on the control (upper) code block, and $(\chi_3'', \zeta_3'')$ the errors on the target (lower) code block.  Denote the final output errors by $(\chi, \zeta)$.

Define sub-component $j$ to be $\badZ^{(j)}$ if it contains five or more $Z$ failures.  If the sub-component is not $\badZ^{(j)}$, it is $\goodZ^{(j)}$.  Similarly define $\badX^{(j)}$ and $\goodX^{(j)}$ for $X$ failures.

Define the $X$-error verification component to be ``$\badZ$'' if any of the sub-components are $\badZ^{(j)}$, or there are more than six $Z$ failures.  If the component is not $\badZ$, it is $\goodZ$.  Similarly define $\badX$ and $\goodX$ for $X$ failures.
 Define the $X$-error verification component to be ``$\best$'' if there are fewer than four failures of any kind.

The quantities that we will compute for an $X$-error verification component are:
\begin{gather*}
 \Pr[(\chi, \zeta)=(x,z), K = k, \best \vert \acceptX],\; \Pr[\chi=x, K_X=k, \goodX \vert \acceptX],\\
 \Pr[\zeta=z, K_Z \leq k, \goodZ \vert \acceptX],\; \Pr[\badX \vert \acceptX] \text{ and } \Pr[\badZ \vert \acceptX]
 \enspace .
\end{gather*}
Here, for example, the first quantity is the probability of $X$-error $x$ and $Z$-error $z$ occurring with exactly~$k$ failures and the $\best$ event, conditioned on $X$-error verification accepting (the event $\acceptX$).  The second and third quantities are similar, except tracking only $X$ or $Z$ errors, respectively.  

Begin by placing a lower bound on the probability of the event $\acceptX$ that no $X$ errors are detected.
Define $\out(x) := \{ \vec{x} : x_1 x_3' \equiv x\}$, $\out(z) := \{ \vec{z} : z_1 z_2 z_3' \equiv z\}$ and $\accept_X := \{ \vec{x} : x_1 x_2 x_3'' \equiv 0 \}$.  
Use 
\begin{equation}
  \label{eq:prAcceptX}
  \Pr[\acceptX] \ge 
  \Pr[\acceptX, \goodX] =
  1 - \Pr[\rejectX, \goodX] - \Pr[\badX] \\
\end{equation}
and
\begin{equation}
  \Pr[\rejectX, \goodX] \leq 
  \sum_{\substack{
        k \leq 6, \abs{\vec k}=k \\
        \vec{x} \not\in \accept_X
  }}
  \prod_{j=1}^3 \Pr[\chi_j=x_j, K_{X,j}=k_j, \goodX^{(j)}]
 \enspace .
 \label{eq:verifyX-pr-reject}
\end{equation}
Here, the sum is over all possible divisions $\vec k = (k_1, k_2, k_3)$ of the number of failures among the three sub-components, and of $X$ errors that lead to a nontrivial syndrome measurement; it is a discrete convolution of the error probabilities. \figref{fig:pr-accept-lower-bounds} shows computed lower bounds on $\Pr[\accept]$ for one particular ancilla preparation and verification circuit.  

The calculations we relate here and in the sequel are generally dictated by constraints of combinatorial complexity.  In Eq.~\eqnref{eq:verifyX-pr-reject}, for instance, the number of terms in the sum is, na{\"i}vely, on the order of $(2^{12})^4 \cdot 3^6 \approx 2 \times 10^{17}$, since there are $2^{12}$ inequivalent $X$ errors on a single code block.  Summing so many terms would be infeasible.  In fact, though, the number of inequivalent $X$ errors produced by an ancilla preparation circuit with $k \leq 2$ faults is much less than $2^{12}$.  For ancillas prepared using~\figref{fig:overlap-prep-ckt}, there are $58$ inequivalent $X$ errors created with $k = 1$, and $1225$ created for $k = 2$.  The number of inequivalent $X$ errors for the transversal CNOT scales as $\binomial{23}{k} 3^k$.  Since the number of possible partitions of $k$ faults into $m$ components is $O(k^m)$, the worst case partition with $\abs{\vec k} = 6$, $\vec{k} = (1,2,3)$, involves only about $3 \times 10^9$ error combinations.  The bound in~\eqnref{eq:verifyX-pr-reject} can therefore be computed with relative ease.  This combinatorial analysis is very similar for the other equations below.  

We similarly compute for $k \in \{0, 1, 2, 3\}$, 
\begin{align}
  \label{eq:Pr-bestX-acceptX}
  \Pr[(\chi, \zeta)=(x,z), K=k, \best, \acceptX]
  &=\!\!\!\!\!\!\!\!\!\!\! \sum_{\substack{
        \abs{\vec k}=k \\
        \vec{x} \in \out(x) \cap \accept_X \\
        \vec{z} \in \out(z)         
     }}
     \!\!\!\! \prod_{j=1}^3 \Pr[(\chi_j, \zeta_j)=(x_j,z_j), K_j = k_j]
      \enspace ,
\end{align}
and for $k \in \{0, \ldots, 6\}$, 
\begin{equation} \label{eq:Pr-goodX-acceptX}
  \Pr[\chi=x, K_X=k, \goodX, \acceptX] = \!\!\!\!\!\!\!
   \sum_{\substack{
        \abs{\vec{k}}=k \\
        \vec{x} \in \out(x) \cap \accept_X
   }}
   \prod_{j=1}^3 \Pr[\chi_j=x_j, K_{X,j}=k_j, \goodX^{(j)}]
\enspace .
\end{equation}

It is more difficult to compute $\Pr[\zeta=z, K_Z=k, \goodZ, \acceptX]$ accurately.  The na{\"i}ve bound $\Pr[\zeta=z, K_Z=k, \goodZ, \acceptX] \leq \Pr[\zeta=z, K_Z=k, \goodZ]$ is quite poor, since $X$-error verification catches many $Z$ faults that occur with $X$ faults, i.e., as a $Y$ fault.  The problem, though, is that we lack $X$-error information with which to determine whether verification is successful or not.  Therefore we use instead the bound 
\begin{equation}
\label{eq:pr-z-good-acceptX-corrected}
  \Pr[\zeta=z, K_Z \leq k, \goodZ \vert \acceptX]
  \leq \frac{\left(
    \begin{aligned}
      &\Pr[\zeta=z, K_Z \leq k, \goodZ]\\
      &\quad -\Pr[\zeta=z, K \leq k, \best, \rejectX]
    \end{aligned}
    \right)}
    {\Pr[\acceptX]}
  \enspace ,
\end{equation}
which holds because $\best$ is a subset of $\goodZ$.  The numerator can be decomposed as
\begin{align} 
  \Pr[\zeta=z, K_Z \leq k, \goodZ] - \Pr[\zeta=z, K \leq k, \best, \rejectX]
   &= \sum_{k'=0}^k \PrC(z,k')
\end{align}  
where
\begin{align} 
  \PrC(z,k)
  &:= \Pr[\zeta=z, K_Z=k, \goodZ] 
   - \sum_x \Pr[(\chi, \zeta)=(x,z), K=k, \best, \rejectX]
   \enspace .
\label{eq:correction-k-term}
\end{align}  

The first term of~\eqnref{eq:correction-k-term} represents the pessimistic assumption that all $Z$ errors pass verification under the $Z$-only noise model. It does not require any $X$-error information and may be computed in the same way as Eqs.~\eqnref{eq:Pr-bestX-acceptX} and~\eqnref{eq:Pr-goodX-acceptX}: 
\begin{align}
  \Pr[\zeta=z, K_Z=k, \goodZ]
  &= \sum_{\substack{
        \abs{\vec{k}}=k \\  
        \vec{z} \in \out(z)
     }}
     \prod_{j=1}^3 \Pr[\zeta_j=z_j, K_Z=k_j, \goodZ^{(j)}]
 \enspace .
\label{eq:pr-goodZ-uncorrected}
\end{align}
The second term uses the full $XZ$ noise model and corrects the over-counting of the first term by subtracting off most of correlated $Z$ errors that are rejected. It is nearly identical to \eqnref{eq:Pr-bestX-acceptX} except that the \emph{rejected} errors are counted instead of the accepted errors.  It is computed as
\begin{align}
  \Pr[(\chi, \zeta)=(x,z), K=k, \best, \rejectX]
  &= \!\!\!\!\!\!\!\!\! \sum_{\substack{
        \abs{\vec{k}}=k \\  
        \vec{x} \in \out(x) \setminus \accept_X \\
        \vec{z} \in \out(z)
     }}
     \!\!\!\!\!\! \prod_{j=1}^3 \Pr[(\chi_j, \zeta_j)=(x_j,z_j), K_j=k_j]
 .
 \label{eq:pr-best-correction}
\end{align}

To get a quantitative estimate of the significance of this correction, we show in \tabref{tbl:verifyX-corrections} the sum of~\eqnref{eq:pr-goodZ-uncorrected} and~\eqnref{eq:pr-best-correction} over all nontrivial $Z$ errors for $p=10^{-3}$.  From this table, we compute a ratio $\Pr[\zeta \neq 0, \best, \rejectX] / \Pr[\zeta \neq 0, \goodZ]$ of about $0.57$, indicating that, as expected, the correction cuts the probability of a $Z$ error roughly in half.  First-order quantities account for most of the correction. Third-order quantites are negligible, providing further justification for our choice of $\kBest=3$.

We see from~\figref{fig:pr-accept-lower-bounds} that the lower bound on $Z$-error verification acceptance at $p=10^{-3}$ is about $0.84$.  The correction eliminates better than half of the $Z$ errors going into $Z$-error verification, so we crudely estimate a lower bound \emph{without} the correction of about $0.63$, a decrease by a factor of~$1.3$.  There are four $Z$-error verifications of encoded $\ket 0$ in the (full) exRec and four similar $X$-error verifications of encoded $\ket +$.  Thus, in the normalization factor alone, the correction reduces upper bounds on the malignant event probabilities (see~\eqnref{eq:PrMalig}) by roughly a factor of $1.3^8 \approx 8$.  The savings is less, of course, as $p$ decreases.

Finally, bound the probability of the $\badX$ event by $\Pr[\badX] \leq \Pr[K_X > 6] + \sum_j \Pr[\bad_X^{(j)}]$, and use $\Pr[\badX \vert \acceptX] \leq \Pr[\badX] / \Pr[\acceptX]$.  The probability of the $\badZ$ event is similarly bounded.  

{
\setlength{\tabcolsep}{5pt}
\begin{table}
\centering
\small
\begin{tabular}{c|cccccc}
\hline\hline
& $k=1$ & $k=2$ & $k=3$ & $k=4$ & $k=5$ & $k=6$ \tabularnewline
\hline
$\Pr[\zeta \neq 0, K_Z=k, \goodZ]$        & $0.1228$ & $0.0101$ & $0.0005$ & $2\times 10^{-5}$ & $6 \times 10^{-7}$ & $1 \times 10^{-8}$ \tabularnewline
$\Pr[\zeta \neq 0, K=k, \best, \rejectX]$ &  $0.0614$ & $0.0140$ & $0.0012$ & - & - & - \tabularnewline
\hline \hline
\end{tabular}
\caption{This table shows the (un-normalized) probability that a non-trivial $Z$ error occurs during $X$-error verification of $\lket{0}_1$ for the Overlap-$4$ verification schedule, evaluated at $p=10^{-3}$.  The first row gives upper bounds on the probability of a nontrivial $Z$-error under the $Z$-only noise model, assuming that all $Z$ errors pass verification.  The second row gives lower bounds on the correction applied in \eqnref{eq:pr-z-good-acceptX-corrected} based on the full depolarizing noise model.} 
\label{tbl:verifyX-corrections}
\end{table}
}

\subsection{\texorpdfstring{$Z$}{Z}-error verification} \label{sec:Z-verification-details}

\begin{figure}
\centering
\includegraphics[scale=1]{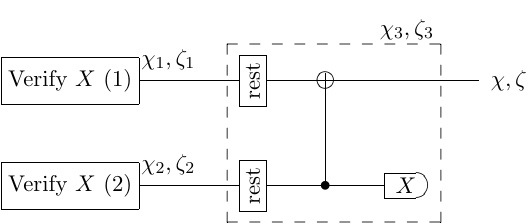}
\caption{$Z$-error verification.  The $Z$-error verification component consists of two $X$-error-verified ancillas, two transversal rest operations, a transversal CNOT and a transversal $X$-basis measurement.} 
\label{fig:Z-error-verification}
\end{figure}

The $Z$-error verification component is illustrated in \figref{fig:Z-error-verification} with its three labeled independent sub-components.  Define the component to be $\badX$ if any of the sub-components are $\badX$ or there are more than seven $X$ failures.  The $\badX$ events for sub-components one and two are defined in \appref{sec:X-verification-details} and $\badX^{(3)}$ occurs when component three contains more than four $X$ failures.  Thus $\Pr[\badX] \leq \Pr[K > 7] + \sum_{j=1}^3 \Pr[\badX^{(j)}]$.  Similarly define $\badZ$ for $Z$ failures.  Also define $\out(x) := \{ \vec{x} : x_1 x_2 x_3' \equiv x \}$, $\out(z) := \{ \vec{z} : z_1 z_3' \equiv z \}$ and $\accept_Z := \{ \vec{z} : z_1 z_2 z_3'' \equiv 0 \}$.

Begin by placing a lower bound on the acceptance probability $\Pr[\acceptZ]$ conditioned on acceptance of both $X$-error verifications.  As in the $X$-error verification component, we use an estimate based on the good events only: 
\begin{equation}\begin{split} \label{eq:prAcceptZ}
  \Pr[\acceptZ \vert \accept^{(1,2)}] &\ge \Pr[\acceptZ, \goodZ \vert \accept^{(1,2)}] \\ &= 1 - \Pr[\rejectZ, \goodZ \vert \accept^{(1,2)}] - \Pr[\badZ \vert \accept^{(1,2)}]
\end{split}
\end{equation}
\begin{equation}
\begin{split}
  \Pr[\rejectZ, \goodZ \vert \accept^{(1,2)}]
  &= \sum_{\substack{
              \vec{z} \not\in \accept_Z \\
              \abs{\vec{k}} \leq 7
        }}
     \left(\begin{split}
     \Pr[\zeta_3&=z_3, K_{Z,3}=k_3, \goodZ^{(3)}] \\
     \cdot \prod_{j=1}^2 &
        \begin{split}
          &\Pr[\zeta_j=z_j, K_{Z,j}=k_j, \goodZ^{(j)}\\ 
          &\qquad\qquad \vert \accept^{(j)}]
        \end{split}
     \end{split}
\right) \\
  &\leq \sum_{\substack{
              \vec{z} \not\in \accept_Z \\
              \abs{\vec{k}} \leq 7                        
        }}
        \left(\begin{split}
        &\Pr[\zeta_3=z_3, K_{Z,3}=k_3,\goodZ^{(3)}]\\
        &\cdot \prod_{j=1}^2 \frac{\PrC_j(z_j,k_j)}{\Pr[\accept^{(j)}]}
        \end{split} \right)                 
 \enspace ,
\end{split}\end{equation}
where $\PrC_j$ is defined according to~\eqnref{eq:correction-k-term}.

Now consider $Z$ errors, the simpler case.  We have
\begin{equation}
\begin{split}
&\Pr[\zeta=z, K_Z \leq k, \goodZ, \accept \vert \accept^{(1,2)}] 
 \leq \\
 &\qquad\qquad\qquad
 \sum_{\substack{
        \vec{z} \in \out(z) \cap \text{\accept}_Z \\
        \abs{\vec{k}} \leq \min\{k,7\}
     }}
     \left(\begin{split}
     \Pr[\zeta_3&=z_3, K_{Z,3}=k_3, \goodZ^{(3)}] \\
     \cdot \prod_{j=1}^2 &\frac{\PrC_j(z_j,k_j)}{\Pr[\accept^{(j)}]} 
     \end{split}\right)
 \enspace .
\end{split}
\end{equation}
We upper-bound $\Pr[\zeta=z, K_Z \leq k, \goodZ \vert \acceptZ, \accept^{(1,2)}]$ by the minimum of one and the ratio of $\Pr[\zeta=z, K_Z \leq k, \goodZ, \acceptZ \vert \accept^{(1,2)}]$ divided by the previously computed lower bound on $\Pr[\acceptZ \vert \accept^{(1,2)}]$.  

Next consider $X$ errors.  Under the $X$-only noise model, we have no information about $Z$ errors and we must pessimistically assume that all $X$ errors pass verification (i.e., $\Pr[\chi,\acceptZ]=\Pr[\chi]$).  In reality, some of the $X$ errors will occur with $Z$ errors and will be rejected.  In the same way that corrections were applied for $Z$-error counts during $X$-error verification, we apply low-order corrections to the $X$-error counts by considering $X$ and $Z$ errors together.  In a similar manner to Eq.~\eqnref{eq:pr-z-good-acceptX-corrected} we have
\begin{equation}\begin{split}
\label{eq:pr-x-good-acceptZ-corrected}
  &\Pr[\chi=x, \goodX, \acceptZ \vert \acceptX^{(1,2)}]
   \\
   &\qquad\qquad\qquad \leq \frac {\Pr[\chi=x, \goodX, \acceptX^{(1,2)}] - 
               \Pr[\chi=x, \best, \rejectZ, \acceptX^{(1,2)}]}
              {\Pr[\acceptX^{(1)}]\Pr[\acceptX^{(2)}]}
  \enspace .
\end{split}\end{equation}
The numerator terms on the right-hand side can be computed as 
\begin{align}
\label{eq:verifyZ-correction-term1}
&\Pr[\chi=x, \goodX, \acceptX^{(1,2)}] 
 = \sum_{\substack{
        \vec{x} \in \out(x) \\
        \abs{\vec{k}} \leq 7
     }}
     \prod_{j=1}^3 \Pr[\chi_j=x_j, K_{X,j}=k_j, \goodX^{(j)}, \acceptX^{(j)}] \\
\label{eq:verifyZ-correction-term2}
&\Pr[(\chi, \zeta)=(x,z), \best, \rejectZ, \acceptX^{(1,2)}] = \\
&\qquad\qquad\qquad\qquad\qquad 
    \sum_{\substack{
        \vec{z} \in \out(z) \setminus \accept_Z \\
        \vec{x} \in \out(x),
        \abs{\vec k} \leq 1
     }}
     \prod_{j=1}^3 \Pr[(\chi_j, \zeta_j)=(x_j, z_j), K_j=k_j, \acceptX^{(j)}] \nonumber 
\enspace .
\end{align}
Equation \eqnref{eq:pr-x-good-acceptZ-corrected} is then upper bounded by upper bounding \eqnref{eq:verifyZ-correction-term1}, lower bounding $\Pr[\acceptX^{(1)}]$ and $\Pr[\acceptX^{(2)}]$, and computing \eqnref{eq:verifyZ-correction-term2} with equality, all of which can be accomplished with quantities from \appref{sec:X-verification-details}. 

This correction is less significant than the similar correction applied in $X$-error verification.  By the time $Z$-error verification occurs, many of the $X$-errors have already been eliminated by $X$-error verification.  The $X$ errors that do pass verification are less correlated with $Z$ errors; the correction eliminates only about $39$ percent of the nontrivial $X$ errors compared to about $57$ percent for the analogous $X$-error verification correction.  Furthermore, this correction has no effect on normalization because acceptance at this stage depends only on $Z$ errors, and there are no further postselection steps in the exRec.

\subsection{Error correction}
\label{sec:EC-detail}

\def\malig{\text{mal}}
\def\benign{\text{ben}}
\def\chiIn{\chi_{\text{in}}}
\def\chiOut{\chi_{\text{out}}}
\def\zetaIn{\zeta_{\text{in}}}
\def\zetaOut{\zeta_{\text{out}}}

\begin{figure}
\centering
\includegraphics[width=14.2cm]{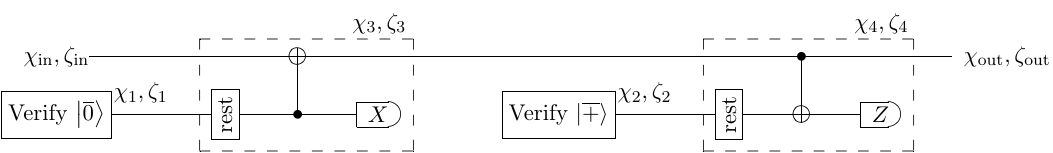}
\caption{The error correction component.  Error correction consists of a $Z$-error correction followed by an $X$-error correction.  The $Z$-error correction requires a verified encoded $\ket 0$ ancilla and the $X$-error correction requires a verified encoded $\ket +$ ancilla.} \label{fig:EC}
\end{figure}

The error correction component (\figref{fig:EC}) contains four independent sub-components. 
The $\badX$ events for sub-components one and two are defined in \appref{sec:Z-verification-details}.  The $\badX$ event for sub-components three and four occur when there are more than four failures inside of the sub-component.  Define the error correction component to be $\badX$ if any of the sub-components are $\badX$ or there are more than eleven $X$ failures total. Similarly define $\badZ$ for $Z$ failures.  

All events are conditioned on the successful verification of both the $\lket{0}$ and $\lket{+}$ ancillas. We have
\begin{equation}\begin{split}
  \Pr[\badX \vert \accept^{(1,2)}]
  &\leq \sum_{j=1}^2 \Pr[\badX^{(j)} \vert \accept^{(j)}] + \sum_{j=3}^4 \Pr[\badX^{(j)}] + \Pr[K_X > 11] \enspace .
\end{split}\end{equation}
Here, $\accept^{(j)}$ means that \emph{all} verification tests, $X$ and $Z$, within that subcomponent have passed.  

Consider first $X$ errors.  For the two leading error corrections, we are concerned only with the case in which the incoming error syndrome is zero.  Nonzero syndromes on the input may result in a (undetectable) logical error at the output.  However, as noted in \secref{sec:counting-EC}, this has no impact on either the output syndrome or the correctness of the $1$-Rec.  The probability of an $X$ error $x$ at the output is expressed as $\Pr[\chiOut=x, \good \vert \chiIn \equiv 0]$.  

For the two trailing error corrections, we must consider all possible inequivalent errors on the input.  However, we do not need to compute the probability of each individual error at the output.  Rather, we care only about the the probability of an uncorrectable error.  Let $\mathcal{E}$ be the set of correctable errors on a single block and $\bar{\mathcal{E}}$ be the set of uncorrectable errors, and for an error~$e$, let 
\begin{equation}
\label{eq:correctable-error-identifier}
 D(e) = 
 \begin{cases}
    1 & \text{if } e \in \bar{\mathcal{E}} \\
    0 & \text{if } e \in \mathcal{E}
 \end{cases}
 \enspace .
\end{equation} 
We use, for $d \in \{0,1\}$, 
\begin{align}
  \Pr[D(\chiOut) = d, \goodX \vert \chiIn=\xIn] 
  &\leq \sum_{x : \, D(x)=d} \Pr[\chiOut=x, \goodX \vert \chiIn=\xIn] \enspace .
\end{align}

The terms of the sum may be expanded as usual by partitioning $k$ failures among the sub-components of the EC.  Define $\out(x) := \{ \vec{x} : \xIn x_1 x_3' x_4' \text{Corr}(\xIn x_1 x_2 x_3' x_4'') \equiv x \}$, where $\text{Corr}(\cdot)$ gives the classically computed correction for the syndrome of its argument.  Then 
\begin{equation}
  \label{eq:prEC-xIn-given-xOut}
  \Pr[\chiOut=x, K_X \leq k, \goodX \vert \chiIn=\xIn]
  =\!\!\!\! \sum_{ \substack{
            \vec{x} \in \out(x) \\
            \abs{\vec{k}} \leq 11
        }}
        \prod_{j=1}^4 \Pr[\chi_j=x_j, K_{X,j}=k_j, \goodX^{(j)}]
  \enspace ,
\end{equation}
which can be upper bounded using quantities from~\appref{sec:Z-verification-details}.  Calculations for $Z$ errors are analogous except using $\out(z) := \{ \vec{z} : \zIn z_2 z_3' z_4' \text{Corr}(\zIn z_1 z_3'') \equiv z \}$.

\subsection{exRec}
\label{sec:exRec-detail}

\begin{figure}
\centering
\includegraphics[scale=1]{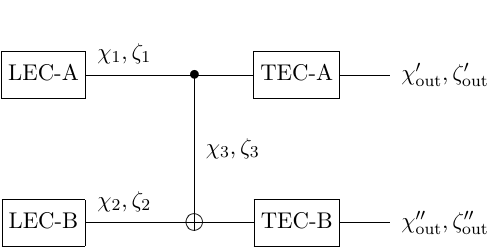}
\caption{The CNOT exRec consists of two leading error corrections, a transversal CNOT with controls on block A and targets on block B, and two trailing error corrections. Incremental errors for the LECs and the CNOT are labeled, as are the total errors at the output of each block.}
\label{fig:exRec-detail}
\end{figure}

The exRec is divided into five sub-components: the leading error correction on block A (LEC-A), the leading error correction on block B (LEC-B), the transversal CNOT from block A to block~B, the trailing error correction on block A (TEC-A) and the trailing error correction on block B (TEC-B).  See \figref{fig:exRec-detail}.  

The $\badX$ event for the error correction sub-components is defined in \appref{sec:EC-detail}.  The $\badX$ event for the transversal CNOT occurs when it contains more than two $X$ failures.
The $\badX$ event for the exRec (and analogously the $\badZ$ event) occurs when any of the following are true: 
\begin{itemize}
  \item Any of the sub-components are $\badX$.  
  \item There are more than $25$ $X$ failures in the exRec.  
  \item There is more than one $X$ failure in the transversal CNOT \emph{and} there are more than than three $X$ failures in each of the two leading ECs.  
\end{itemize}

The last condition eliminates faults that are particularly difficult to count. The time required to count an exRec fault is proportional to the product of the number of unique syndromes that can result at the output of the two leading ECs and the transversal CNOT.  The number of unique syndromes that can result from the transversal CNOT with two $X$ failures is $\binomial{23}{2} 3^2 = 2277$, while the number of unique syndromes with one $X$ failure is $23 \cdot 3 = 69$.  The numbers of unique syndromes at the output of the leading ECs are $24$, $277$ and $2048$ for one, two, and three $X$ failures respectively.  So, for example, the event $K_{X,1}=2, K_{X,2}=3, K_{X,3}=1$ ($277 \cdot 2048 \cdot 69 \approx 4 \cdot 10^7$) requires far less time than the event $K_{X,1}=2, K_{X,2}=3, K_{X,3}=2$ ($277 \cdot 2048 \cdot 2277 \approx 1 \cdot 10^{9}$).  In particular, we would like to avoid counting faults for which $K_{X,3} = 2$.  

The probability of the $\badX$ event, conditioned on acceptance of all $X$-error and $Z$-error verifications is 
\begin{equation}\begin{split}
  \Pr[\badX \vert \accept^{(1,2,4,5)}] 
  &\leq \!\!\!\!  \sum_{j\in \{1,2,4,5\}} \Pr[\badX^{(j)} \vert \accept^{(j)}] + \Pr[\badX^{(3)}] + \Pr[K_X > 25] \\
  &\qquad + \Pr[K_{X,3}>1] \prod_{j=1}^2 \Pr[K_{X,j} > 3 \vert \accept^{(j)}] \enspace .
\end{split}
\label{eq:prBadExRec}
\end{equation}
Computed upper bounds for this quantity are plotted in \figref{fig:prBad-cnot-upper-bounds}.  Na{\"i}vely one might expect bounds for partial exRecs---those for which one or more TECs have been removed---to be lower than bound for the full exRec by as much as a factor of two.
However,~\eqnref{eq:prBadExRec} is dominated by either the transversal CNOT ($\Pr[\bad_X^{(3)}]$) or the condition involving the transversal CNOT and the two LECs ($\Pr[K_{X,3}>1] \prod_{j=1}^2 \Pr[K_{X,j} > 3 \vert \accept^{(j)}]$) over most of the domain of $p$.  Thus removing the TECs has little impact on the probability that the exRec is bad.  

Recall from \eqnref{eq:correctable-error-identifier} the definition of $D:(\mathcal{E} \cup \bar{\mathcal{E}}) \rightarrow \{0,1\}$.  For $\vec x = (x_1, x_2, x_3, x_\out', x_\out'')$ and $\vec z = z_1, z_2, z_3, z_\out', z_\out'')$, we can then define the malignant $X$- and $Z$-error events for the CNOT $1$-Rec as 
\begin{equation}\begin{split} \label{eq:PrMalig}
  \malig_{IX}(\vec x) &:= \left[D(x_1) = D(x_\out') \right] \wedge 
                  \left[D(x_1) \oplus D(x_2) \neq D(x_\out'') \right] \\
  \malig_{XI}(\vec x) &:= \left[D(x_1) \neq D(x_\out') \right] \wedge 
                  \left[D(x_1) \oplus D(x_2) = D(x_\out'') \right] \\
  \malig_{XX}(\vec x) &:= \left[D(x_1) \neq D(x_\out') \right] \wedge 
                  \left[D(x_1) \oplus D(x_2) \neq D(x_\out'') \right] \\
  \malig_{IZ}(\vec z) &:= \left[ D(z_1) \oplus D(z_2) = D(z_\out') \right] \wedge 
                  \left[ D(z_2) \neq D(z_\out'') \right] \\
  \malig_{ZI}(\vec z) &:= \left[ D(z_1) \oplus D(z_2) \neq D(z_\out') \right] \wedge
                  \left[ D(z_2) = D(z_\out'') \right] \\
  \malig_{ZZ}(\vec z) &:= \left[ D(z_1) \oplus D(z_2) \neq D(z_\out') \right] \wedge
                  \left[ D(z_2) \neq D(z_\out'') \right]  \enspace .
\end{split}\end{equation}
These are the events for which the behavior of the $1$-Rec differs from the behavior of an ideal decoder followed by an ideal (level-zero) CNOT gate, i.e., the $1$-Rec is incorrect.  The subscripts denote the logical error introduced by the exRec.  For example, $\malig_{IX}$ is the event in which the action of the exRec followed by an ideal decoder is the same as that of an ideal decoder followed by an ideal CNOT gate plus the two-qubit error error $I\otimes X$.  

For each error event~$E \in \{ IX, XI, XX \}$, we are interested in the probability that $\malig_E(\vec \chi)$ occurs along with the $\goodX$ event. Define $G := \{ \vec{k} : \abs{\vec{k}} \leq 25, k_3 \leq 1 \text{ if } k_1 \geq 4 \text{ and } k_2 \geq 4 \}$.  Then, letting $\chiIn'$ and $\chiIn''$ be the errors input to the two TECs, 
\begin{equation}\begin{split} 
  \label{eq:PrMaligExRec}
  \Pr[\malig_E(\vec \chi), \goodX] &= \!\!\!\! 
  \sum_{\substack{
        \vec x : \, \malig_E(\vec x) \\
        \vec{k} \in G
     }}
  \left[
  \begin{split}
       \prod_{j=1}^3 \Pr[\chi_j = x_j, K_{X,j}=k_j, \goodX^{(j)}]  \\
  \cdot \Pr[\chiOut' = \xOut', K_{X,4}=k_4, \goodX^{(4)} \vert \chiIn' \equiv x_1 x_3'] \\
  \cdot \Pr[\chiOut'' = \xOut'',K_{X,5}=k_5, \goodX^{(5)} \vert \chiIn'' \equiv x_1 x_2 x_3'']
  \end{split}
  \right]
  \enspace ,
\end{split}\end{equation}
which may be upper bounded using quantities from~\appref{sec:EC-detail}.  Calculation of $\Pr[\malig_E(\vec \zeta), \goodZ]$ for $E \in \{ IZ, ZI, ZZ \}$ is analogous.

\begin{figure}
\centering
\subfigure[]{
\includegraphics[width=6.79cm]{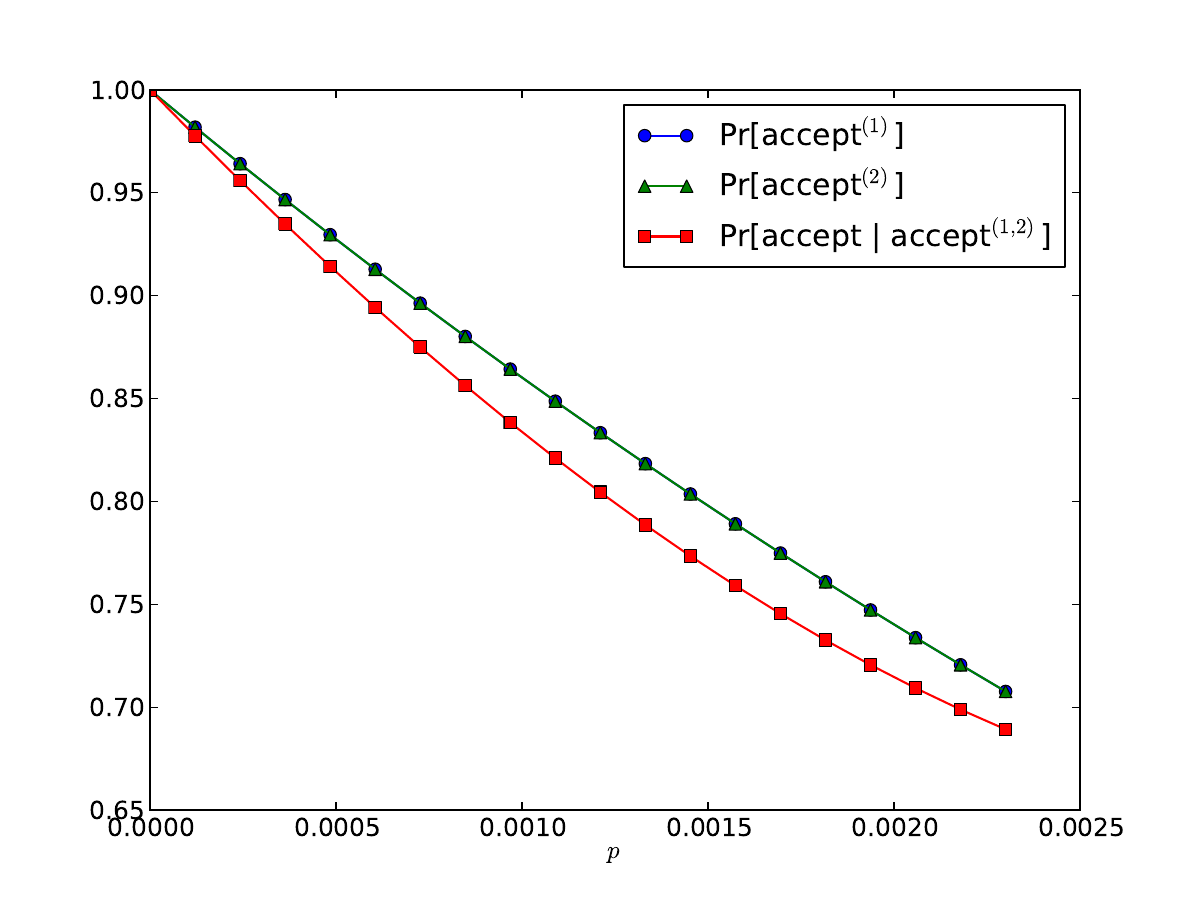}
\label{fig:pr-accept-lower-bounds}
}
\subfigure[]{
\includegraphics[width=6.79cm]{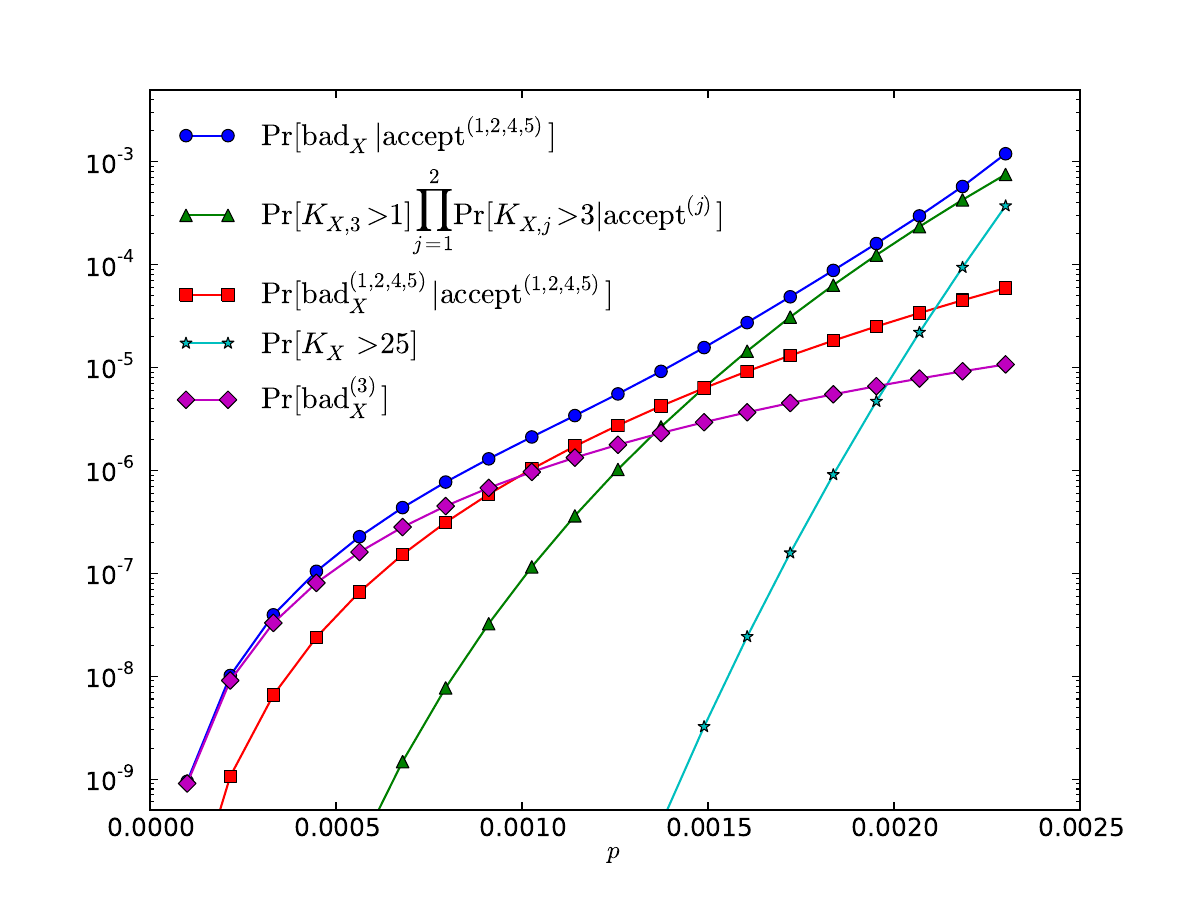}
\label{fig:prBad-cnot-upper-bounds}
}
\caption{Plotted in (a) are lower bounds on the Overlap-$4$ acceptance probabilities for the two $X$-error verifications ($\accept^{(1)}$ and $\accept^{(2)}$) and for the $Z$-error verification ($\accept$) conditioned on success of the $X$-error verifications.  The plot in~(b) shows upper bounds on~\eqnref{eq:prBadExRec}, the probability that the CNOT exRec is bad.
}
\end{figure}

\section{Implementation of component counting}  
\label{sec:computer-analysis-details}

Equations~\eqnref{eq:prBadExRec} and~\eqnref{eq:PrMaligExRec} in \appref{sec:counting-detail} are conceptually straightforward and easy to compute numerically for a fixed $\gamma$.  However, we would like to compute {exact} bounds that hold for a range of~$\gamma$.  In this appendix, we will specify a few of the implementation details that allow for maintaining the bounds as polynomials with integer coefficients.  

The ultimate goal is to compute upper bounds on the probabilities of malignant events at the outermost layer of the exRec.  That is, we want to compute \eqnref{eq:prBadExRec}, \eqnref{eq:PrMaligExRec} and combine them to~get 
\begin{equation}
  \Pr[\malig_E(\vec \chi) \vert \accept] \leq \Pr[\malig_E(\vec{\chi}), \goodX \vert \accept] + \Pr[\badX \vert \accept] \enspace .
\end{equation}
The right-hand side of this inequality decomposes into sums of individual component quantities of the form
\begin{equation}\begin{split} 
  \Pr[\chi=x,K_X=k]
  & = \sum_{\abs{\vec{k}}=k}
      \Pr[\chi=x, \vec{K}_X=\vec{k}] \enspace ,
\end{split}
\label{eq:ComponentQuantity}
\end{equation}
where $\vec{k} = (k_c, k_r, k_p, k_m)$ expresses the number of failing CNOT, rest, $\ket 0$ preparation and $Z$-basis measurements, respectively.  

For each term in the sum, the number of failures for each type of location is fixed, but the particular locations on which those failures occur are not fixed, nor are the errors that occur at those locations.  Let $L(\vec k) := \{ (\vec{l}_c, \vec{l}_r, \vec{l}_p, \vec{l}_m) : (\abs{\vec{l}_c}, \abs{\vec{l}_r}, \abs{\vec{l}_p}, \abs{\vec{l}_m}) = \vec{k} \}$ be the set of all possible tuples of failing locations consistent with $\vec k$. Also, let $E(\vec{l})$ be the set of all possible tuples of $X$ errors consistent with failures at locations $\vec l$.  To fix the locations and the errors, use
\begin{equation}\begin{split}
  \Pr[\chi=x, \vec{K}_X=\vec{k}]
  & = \sum_{\substack{ \vec{l} \in L(\vec{k}),
                       \vec{e} \in E(\vec{l})
                    }}
      \Pr[\chi=x, \vec{E}=\vec{e}] \\
  & = \sum_{\substack{ \vec{l} \in L(\vec{k}),
                       \vec{e} \in E(\vec{l})
                    }}
      \mathcal{F}(x, \vec{e}) \Pr[\vec{E}=\vec{e}]
\end{split}
\label{eq:LocationSum}
\end{equation}
where in the second line we have made the substitution $\mathcal{F}(x, \vec{e}) = \Pr[\chi=x \vert \vec{E}=\vec{e}]$.

The boolean function $\mathcal{F}(x, \vec{e})$ takes value one if the component produces the error $x$ for a given ``configuration'' of errors $\vec{e}$ and value zero otherwise.  The error configuration $\vec{e}$ fully specifies the the locations that have failed and the error at each failing location. Let $n_c, n_r, n_p, n_m$ be the total number of CNOT, rest, $\ket{0}$ preparations and $Z$-basis measurements in the component, respectively.  Then from the marginal noise model discussed in \secref{sec:NoiseModel} we have
\begin{equation}\begin{split}
  \Pr[\vec{E}=\vec{e}]
  &= (1-4\gamma)^{n_p + n_m} (1-8\gamma)^{n_r} (1-12\gamma)^{n_c} \\
  &\qquad \cdot
     \Big( \frac{4\gamma}{1-4\gamma} \Big)^{k_p + k_m}
     \Big( \frac{8\gamma}{1-8\gamma} \Big)^{k_r}
     \Big( \frac{4\gamma}{1-12\gamma} \Big)^{k_c} \\
  &\leq A_{\vec{n}}
     4^k 2^{k_r}
     \Big( \frac{\gamma}{1-12\gamma} \Big)^{k}
      \enspace ,
\end{split}
\label{eq:PrConfig}
\end{equation}
where $A_{\vec{n}}$ is defined as in \eqnref{eq:component-AB}.  This inequality is a reasonable approximation for the range $\gamma \leq \frac{2 \times 10^{-3}}{15}$ with which we are concerned.  It allows us to move $\gamma$ into a prefactor in front of the sum of \eqnref{eq:ComponentQuantity}, and permits an integer representation in the computer analysis.  Indeed, substituting back into equation \eqnref{eq:ComponentQuantity} gives 
\begin{equation}
  \label{eq:Component-expression-X}
  \Pr[\chi=x,K_X=k]
  \leq A_{\vec{n}}
    \Big( \frac{\gamma}{1-12\gamma} \Big)^{k}
    \sum_{\substack{ \abs{\vec{k}} = k \\
                     \vec{l} \in L(\vec{k}),
                     \vec{e} \in E(\vec{l})
         }}
    4^k 2^{k_r}
    \mathcal{F}(x, \vec{e})
  \enspace .
\end{equation}

$X$- and $Z$-error verification require corrections that involve counting $X$ and $Z$ errors together (\appref{sec:X-verification-details}). In that case, the probability of an error configuration is computed according to \defref{def:depolarizing-noise} and we require a lower bound.  We have 
\begin{equation}\begin{split}
  \Pr[\vec{E}=\vec{e}]
  &= (1-4\gamma)^{n_p' + n_m'} (1-12\gamma)^{n_r} (1-15\gamma)^{n_c} \\
  &\qquad \cdot
     \Big( \frac{4\gamma}{1-4\gamma} \Big)^{k_p + k_m}
     \Big( \frac{4\gamma}{1-12\gamma} \Big)^{k_r}
     \Big( \frac{\gamma}{1-15\gamma} \Big)^{k_c} \\
  &\geq A'(\vec{n}')
     4^{k_p + k_m + k_r}
     \Big( \frac{\gamma}{1-4\gamma} \Big)^{k} \\
  &= A_{\vec{n}}
\Big( \frac{A'(\vec{n'})}{A_{\vec{n}}} \Big)
\Big( \frac{\gamma}{1-12\gamma} \Big)^{k}
\Big( \frac{1-12\gamma}{1-4\gamma} \Big)^k
     4^{k_p + k_m + k_r} \\
  &\geq A_{\vec{n}}
     \Big( \frac{\gamma}{1-12\gamma} \Big)^{k}
     \mathcal{S}
     4^{k_p + k_m + k_r}
  \enspace .
  \label{eq:XZ-conversion}
\end{split}
\end{equation}
Here, $n_p'$ and $n_m'$ are the total numbers of preparation and measurement locations (including now $\ket +$ preparations and $X$-basis measurements), and $A'(\vec n') = (1-4\gamma)^{n_p' + n_m'} (1-12\gamma)^{n_r} (1-15\gamma)^{n_c}$.  The scaling factor $\mathcal{S} = \big\lfloor \frac{A'(\vec{n}')}{A_{\vec{n}}} \big( \frac{1-12\gamma_{\text{max}}}{1-4\gamma_{\text{max}}} \big)^k \big\rfloor$ converts the $XZ$ probability into a form that is compatible with $X$-only and $Z$-only probabilities (Eq.~\eqnref{eq:PrConfig}) while maintaining the lower bound and integer representation. The constant $\gamma_\text{max}$ is chosen so that it is higher than the expected threshold value.  Now, \eqnref{eq:pr-z-good-acceptX-corrected} and \eqnref{eq:pr-x-good-acceptZ-corrected} can be rewritten so that the sums do not depend on $\gamma$, and terms with corrections such as~\eqnref{eq:correction-k-term} can be represented by a single integer.

Another advantage of counting component probabilities as likelihoods, is that the counts compose nicely.  If we apply \eqnref{eq:ComponentQuantity} to itself and combine with \eqnref{eq:Component-expression-X}, we end up with
\begin{equation}\begin{split}
  \label{eq:prX-recursive}
  \Pr[\chi=x,K_X=k]
  &= \sum_{\substack{ \abs{\vec{k}} = k \\
                     \vec{x} \in \out(x)
          }}
     \prod_j \Pr[\chi_j=x_j, K_{X,j}=k_j] \\
  &\leq A_{\vec{n}}
    \Big( \frac{\gamma}{1-12\gamma} \Big)^k
    \Bigg[
    \sum_{\substack{ \abs{\vec{k}} = k \\
                     \vec{x} \in \out(x)
          }}
    \prod_j
    \sum_{\substack{ \abs{\vec{k}_j} = k_j \\
                     \vec{l} \in L(\vec{k}_j), \vec{e} \in E(\vec{l})
         }}
    4^{k_j} 2^{k_{j,r}}
    \mathcal{F}(x_j, \vec{e})
    \Bigg] \enspace .
\end{split}\end{equation}

The substitution made in the first line can be applied successively for each sub-component $j$.  Once the lowest level component is reached, we use \eqnref{eq:Component-expression-X} to push dependence on~$\gamma$ outside of the sum.  The integer value inside of the brackets is the discrete convolution of weighted counts from the sub-components summed over all possible failure partitions $\vec k$ of size $k$.  It is a weighted count of all possible ways to produce error $x$ with an order $k$ fault.  

The primary task of the Python program is to compute $\mathcal{F}$ for each (good) error configuration, starting with the lowest level component, and to store the resulting weighted sums 
\begin{equation}
  \sum_{\substack{ \abs{\vec{k}} = k \\
                     \vec{l} \in L(\vec{k}), \vec{e} \in E(\vec{l})
         }}
    4^{k} 2^{k_{r}}
    \mathcal{F}(x, \vec{e})
\end{equation}
(or equivalent) for use in the counting for larger components.  At each level, counts for the sub-components are convolved to generate new counts.  The prefactor $A_{\vec{n}} \big( \frac{\gamma}{1-12\gamma} \big)^{k}$ need only be computed at the end, when calculating the threshold.

\section{Monotonicity of malignant event upper bounds} \label{sec:monotonicity}

We now show how to prove that the level-one malignant event polynomials constructed by our counting method are monotone non-decreasing over the interval $\gamma \in [0,1.8 \times 10^{-3}]$, which encompasses our threshold values.  Monotonicity of level-one upper bounds is not strictly required for the proof of \thmref{thm:threshold}.  However, it is useful constructing the transformed noise model (see \appref{sec:bounding-polys}) and in finding the maximum $\gamma$ that satisfies $\mathcal{P}^{(2)}_E(\Gamma^{(1)}) \leq \alpha_E \Gamma^{(1)}$.  Monotonicity of upper bounds for level-two and above follow from \claimref{clm:P2-multiplicative} which depends only on the construction defined by our counting method and not on the actual counting results---see \appref{sec:P2-multiplicative-proof}.  Level-one polynomials, however, can include terms that decrease with $\gamma$. Monotonicity statements for level-one bounds, therefore, depend on coefficients---i.e., weighted counts---computed by our Python implementation.  

Recall that the upper bound $\mathcal{P}_E$ for malignant event $\malig_E$ as defined by \appref{sec:counting-detail} is of the form
\begin{equation}
  \mathcal{P}_E \geq \frac{\Pr[\malig_E,\good]}{\Pr[\accept]} + \Pr[\bad \vert \accept]
  \label{eq:malig-event-poly}
  \enspace .
\end{equation}
Consider first the $\Pr[\bad \vert \accept]$ term.  This term is expressed as sums and products of $\Pr[\bad]$ terms, some of which contain $\Pr[\accept]$ terms in the denominator.  The $\Pr[\bad]$ and $\Pr[\accept]$ terms are, in turn, expressed as sums and products of polynomials $\poly$ of the form
\begin{equation}
   \poly(\gamma)
   = A_{\vec{n}}(\gamma)
     \sum_{k=\kMin}^{\kMax} c(k) \Big( \frac{\gamma}{1-12\gamma} \Big)^k
   \enspace ,
\label{eq:pr-general-form}
\end{equation}
where $\poly(\gamma) \geq 0$ for all $\gamma \geq 0$, and integer coefficients $c(k)$ do not depend on $\gamma$.
The quantity $\Pr[\malig_E,\good]$ is also expressed in this way.  Our goal then is to use \eqnref{eq:pr-general-form} to show that $\Pr[\accept]$ is monotone non-increasing, and that $\Pr[\bad]$ and $\Pr[\malig_E,\good]/\Pr[\accept]$ are monotone non-decreasing over the desired range.  Note that $\Pr[\malig_E,\good]$ is, in fact, not monotone over our chosen range. 

The derivative of $\poly$ is a sum of two terms.  The first term can be lower-bounded by using
\begin{equation}\begin{split}
  \frac{d A_{\vec{n}}}{d\gamma} \sum_{k=\kMin}^{\kMax} c(k) \Big( \frac{\gamma}{1-12\gamma} \Big)^k
  &= \Big( \frac{-12 n_c}{1-12\gamma} - \frac{8 n_r}{1-8\gamma} - \frac{4 n_{pm}}{1-4\gamma} \Big)
     A_{\vec{n}} \sum_{k} c(k) \Big( \frac{\gamma}{1-12\gamma} \Big)^k \\
  &\geq -(12 n_c + 8 n_r + 4 n_{pm}) \frac{\poly(\gamma)}{1-12\gamma} 
  \enspace .
\end{split}
\end{equation}
If all coefficients $c(k)$ are non-negative, then the second term may be lower-bounded as
\begin{align}
  A_{\vec{n}} \sum_{k} c(k) \frac{d}{d\gamma} \Big( \frac{\gamma}{1-12\gamma} \Big)^k
  &\geq \frac{\kMin}{\gammaMax} \frac{\poly(\gamma)}{1-12\gamma}
  \enspace .
\end{align}

In the case of $\Pr[\bad]$, all of the coefficients $c(k)$ are, indeed, non-negative.  Using $\gammaMin=0$, and $\kMin = k_\good + 1$ we obtain
\begin{equation}\begin{split}
  \frac{d\Pr[\bad]}{d\gamma}
  &\geq \Big( \frac{k_\good+1}{\gammaMax} - 12 n_c - 8 n_r - 4 n_{pm} \Big)
        \frac{\Pr[\bad]}{1-12\gamma}
  \enspace ,
  \label{eq:dprBad-lowerBound-monotone}
\end{split}\end{equation}
which is non-negative over $[0,1.8\times 10^{-3}]$ for all of our components.  

We would like to upper bound the denominator quantities $\Pr[\accept]$ using a condition analogous to \eqnref{eq:dprBad-lowerBound-monotone}.  Such a condition is insufficient, however, because for $X$-error verification $\kMin=0$, and for $Z$-error verification some coefficients may be negative due to low-order $XZ$ corrections.
Instead, we bound the second derivative using the following inequality due to Markov.
\begin{lemma}[see, e.g., \cite{Hazewinkel1990} pp.~100]
 Let $\poly$ be a univariate polynomial of degree at most $n$.  Then the $m$-th order derivative $\poly^{(m)}$ is bounded by
 \begin{equation}
   \max_{\gamma \in [\gammaMin,\gammaMax]} \abs{\poly^{(m)}(\gamma)} 
   \leq \frac{2^m \prod_{k=0}^{m-1} (n^2 - k^2)}{(\gammaMax-\gammaMin)^m (2m-1)!!} 
        \max_{\gamma \in [\gammaMin,\gammaMax]} \abs{\poly(\gamma)}
   \enspace .
 \end{equation}
  \label{lem:markov-ineq}
\end{lemma}
In our case, $n$ is bounded by the number of locations in the corresponding component. The maximum of $\poly$ is obtained by separating the positive and negative coefficients and upper bounding~by 
\begin{equation}\begin{split} 
  \max_{\gamma \in [\gammaMin,\gammaMax]} \poly(\gamma)
  &\leq  A_{\vec{n}}(\gamma) \sum_{k=\kMin}^{\kMax} \frac{\max \{c(k),0\} \gammaMax^{k}}{(1-12\gammaMax)^{k}} + 
                                                      \frac{\min \{c(k),0\} \gammaMin^{k}}{(1-12\gammaMin)^{k}}
  \enspace .
  \label{eq:pr-upperBound-monotone}
\end{split} \end{equation}
If $\Delta$ is the bound obtained from \lemref{lem:markov-ineq}, then the first derivative can be bounded using
\begin{equation}
  \max_{\gamma \in [\gammaMin,\gammaMax]} \poly^{(1)} \leq \poly^{(1)}(\gammaMin) + \Delta (\gammaMax - \gammaMin)
  \enspace .
  \label{eq:d1-interval-bound}
\end{equation}
Depending on the values of the coefficients $c(k)$, bounding the first derivative below zero may require dividing up the interval into smaller sub-intervals and successively applying \eqnref{eq:d1-interval-bound}.  

Analysis for $\poly = \Pr[\malig_E,\good]/\Pr[\accept]$ is similar.  Monotonicity over the range $[\epsilon,\gammaMax]$ for small constant $\epsilon$ can be shown by using the lower bound equivalent of \eqnref{eq:d1-interval-bound}.  The maximum of $\poly$ is calculated by using \eqnref{eq:pr-upperBound-monotone} on $\Pr[\malig_E,\good]$ and evaluating $\Pr[\accept]$ at $\gammaMax$.  Lower bounding the first derivative in this way is not adequate for $[0,\epsilon]$, however, because the first derivative vanishes at $\gamma=0$. Due to the strict fault-tolerance of our circuits, coefficients $c(k)$ of $\Pr[\malig_E,\good]$ are zero for $0 \leq k \leq 3$ and so the derivatives up to third order are also zero. To show monotonicity over $[0, \epsilon]$ we evaluate the fourth derivative at $\gamma=0$ and then use \lemref{lem:markov-ineq} to bound the fifth derivative.

\section{The transformed noise model} \label{sec:transformed-noise-model}

\subsection{Construction of the model} \label{sec:transformed-noise-model-construction}

The transformed noise model uses upper bounds on the level-one malignant event probabilities to model each $1$-Rec as a single effective ``location'' in the level-two simulation. The construction here considers only $X$-error malignant events. Construction for $Z$-error events is nearly identical, and the upper bounds obtained from level-one counting contain no information about $X$ and $Z$ correlations so $X$ and $Z$ errors are not considered together at level-two.

From level-one counting, we have upper bounds on $\Pr[\malig_{IX}]$, $\Pr[\malig_{XI}]$, and $\Pr[\malig_{XX}]$ of the transversal CNOT, $\Pr[\malig_X^{\text{prep}}]$ of encoded $\ket{0}$ preparation, $\Pr[\malig_X^{\text{meas}}]$ of transversal $Z$-basis measurement and $\Pr[\malig_X^{\text{rest}}]$ of the transversal rest. Each of these bounds is a polynomial in $\gamma$, the probability of a given error on a physical CNOT (see \figref{fig:malig-events-x}).  Denote the upper bound polynomial for each event $\malig_E$ by $\mathcal{P}_E$.  Then let $\Gamma_X(\gamma)$ be a polynomial defined over the interval $0 \leq \gamma \leq \gamma_{\text{max}}$ such that, for all $\malig_E$, 
\begin{equation}
  \mathcal{P}_E(\gamma) \leq \alpha_E \Gamma_X(\gamma)
\label{eq:Gamma-definition}
\end{equation}
where $\alpha_E \geq 1$ is a constant.  The procedure for obtaining such a polynomial $\Gamma_X$ and constants $\alpha_E$ is outlined in the \secref{sec:bounding-polys}.

Now consider the level-two simulation. Level-two rectangles are composed of many level-one rectangles.  Following~\cite{AliferisGottesmanPreskill05} we replace each $1$-Rec with an implementation of the corresponding (level-zero) gate, starting with the right-most $1$-Recs and moving left.  If a $1$-Rec is correct, then it is replaced with an ideal gate.  If a $1$-Rec is incorrect, then the entire exRec is replaced with a faulty version of the gate.  Unlike~\cite{AliferisGottesmanPreskill05}, however, exRecs containing incorrect rectangles are replaced according to the malignant event that occurred.  For example, a CNOT $1$-Rec that is $\malig_{IX}$ is replaced with an ideal CNOT gate followed by the error $I\otimes X$. A $1$-Rec ($A$) that precedes an incorrect $1$-Rec ($B$) is replaced with a faulty gate only if $A$ is still incorrect after the ECs shared with the exRec containing $B$ have been removed.  Such an incorrect $1$-Rec ($A$) is said to be ``independently'' incorrect.

Let $K_1, K_2, K_3, K_4, K_5$ be the number of level-one exRecs that are independently $\malig_{IX}$, $\malig_{XI}$, $\malig_{XX}$, $\malig_X^{\text{prep}}$, $\malig_X^{\text{meas}}$ and $\malig_X^{\text{rest}}$, respectively.  Then the probability $\Pr[\vec{K}=\vec{k}]$ of this event is bounded by
\begin{equation}
    \label{eq:prK-level-2}
    \Pr[\vec{K}=\vec{k}] 
    \leq \Gamma_X(\gamma)^\abs{\vec{k}} \prod_{i=1}^6 \lceil \alpha_i \rceil^{k_i}
\end{equation}
where the ceiling is taken to allow integer representation in the computer analysis. Thus, characterization and analysis of the level-two components is similar to that used for the depolarizing noise model in \secref{sec:counting-detail} except that the weights associated with each type of error are different.  There is also no prefactor $A_{\vec{n}}$ in \eqnref{eq:prK-level-2} as there is in \eqnref{eq:PrConfig}.  This is because the error probabilities are now specified as upper bounds rather than equalities.  We do not have a proper upper bound, other than one, for the probability that a location ($1$-Rec) does \emph{not} fail.

\begin{figure}
\centering
\subfigure{
\includegraphics[width=6.79cm]{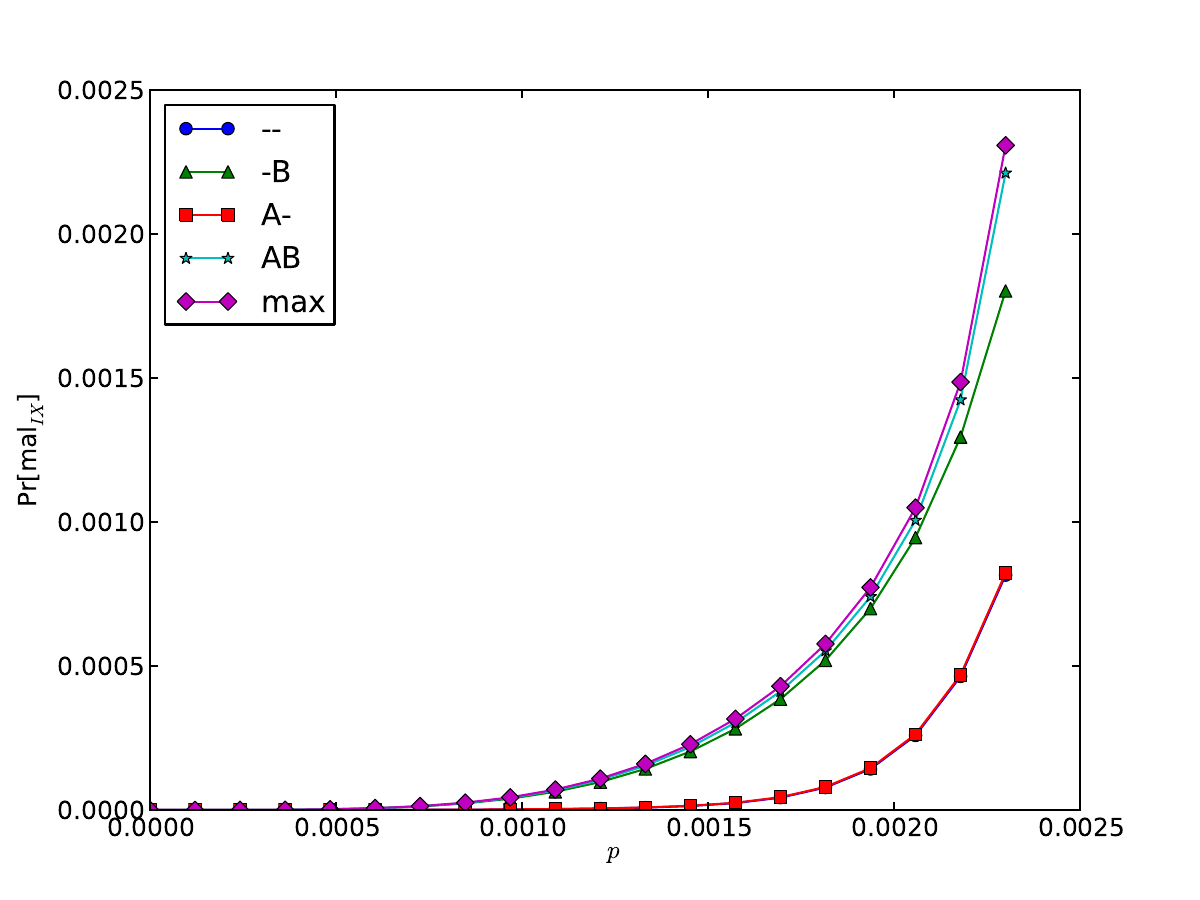}
}
\subfigure{
\includegraphics[width=6.79cm]{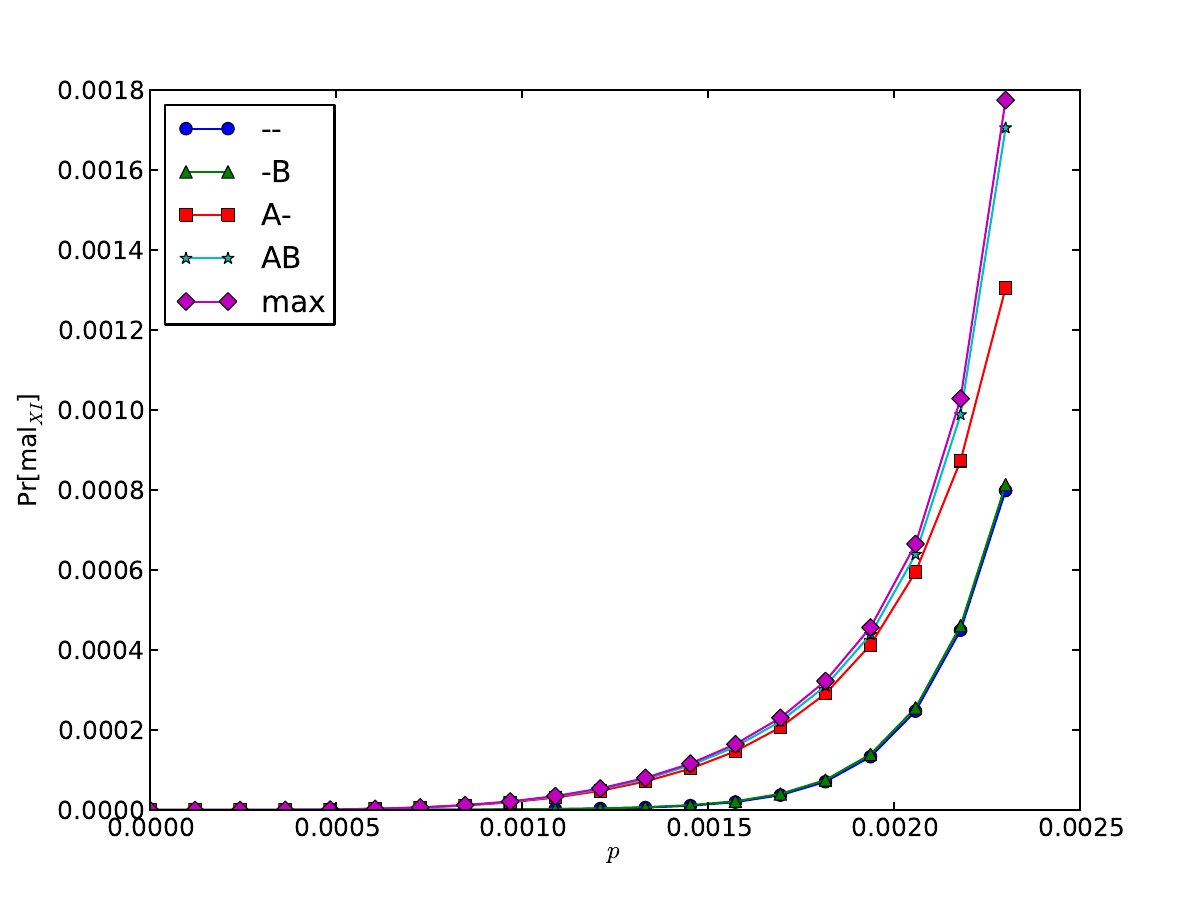}
}
\subfigure{
\includegraphics[width=6.79cm]{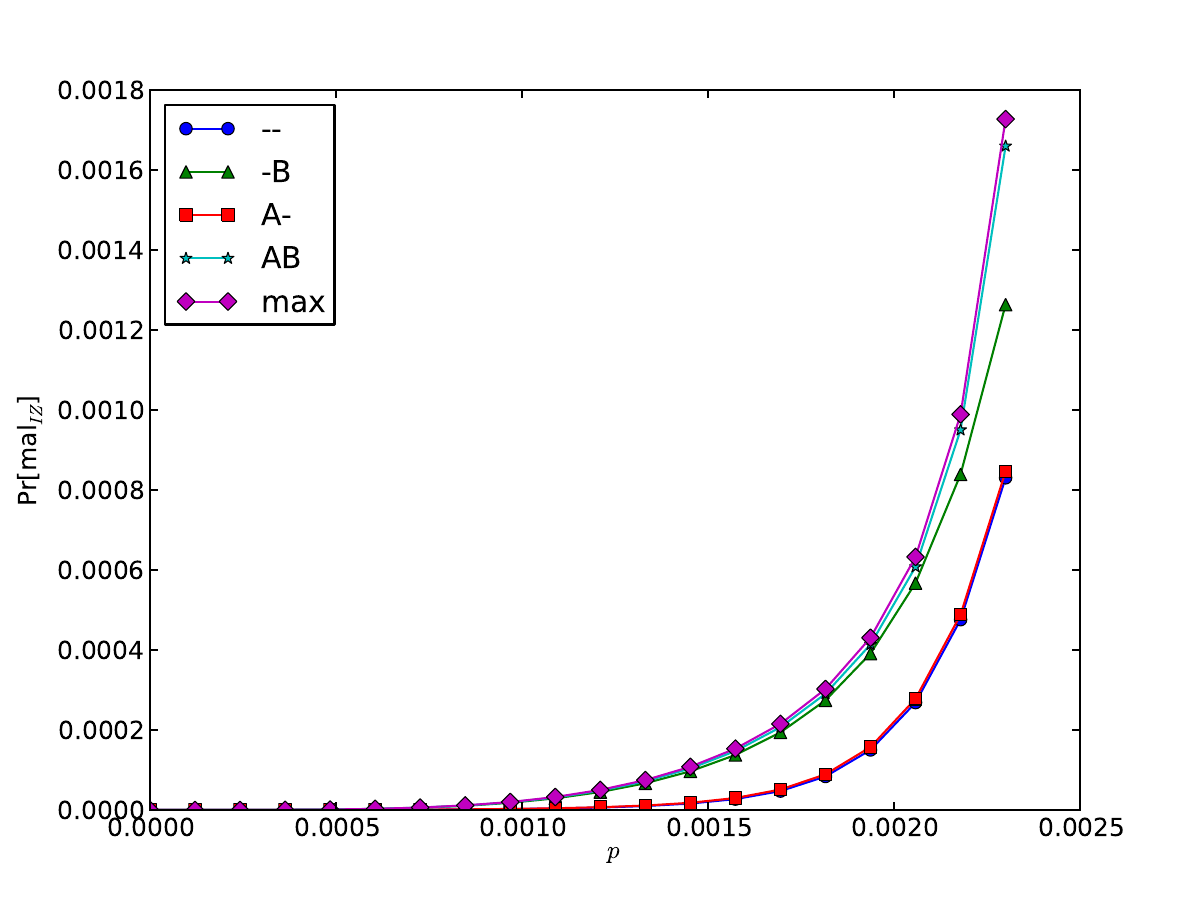}
}
\subfigure{
\includegraphics[width=6.79cm]{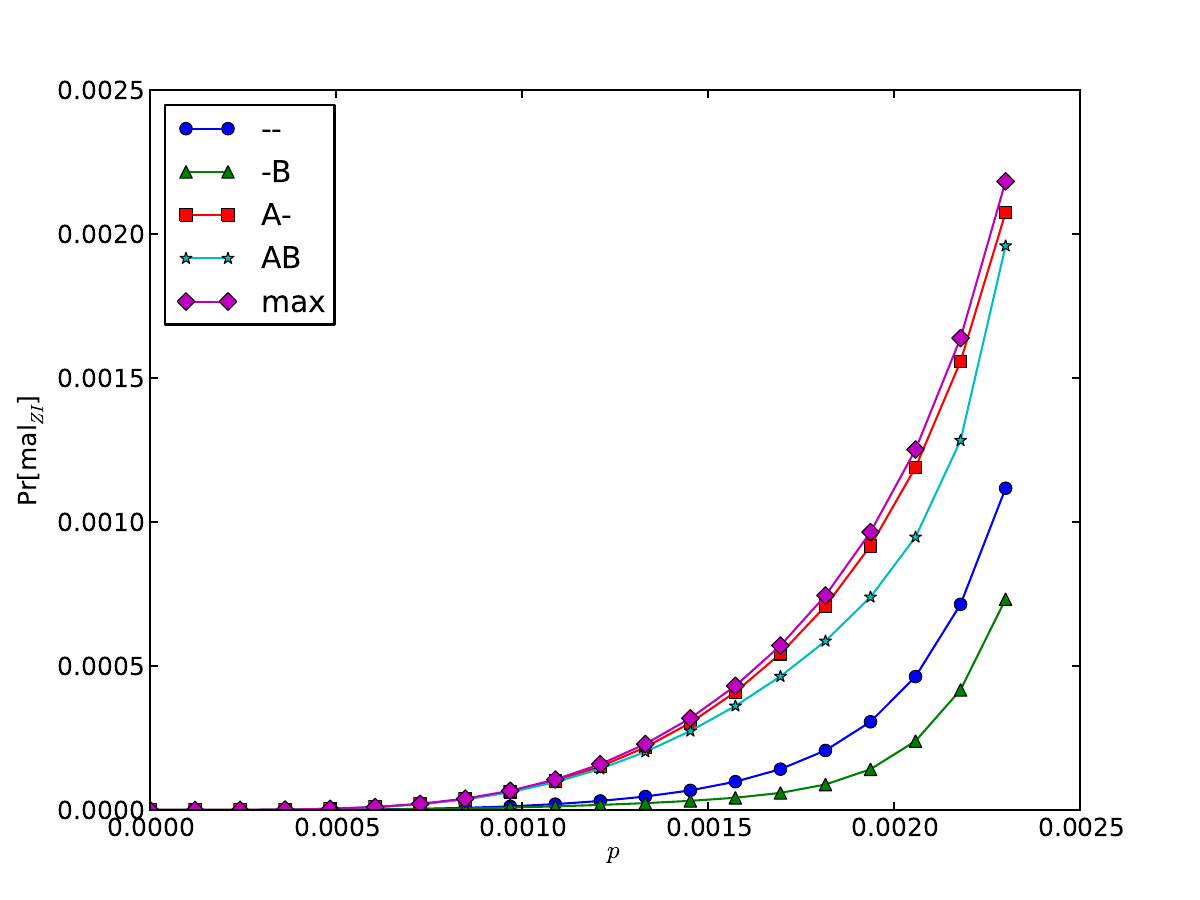}
}
\subfigure{
\includegraphics[width=6.79cm]{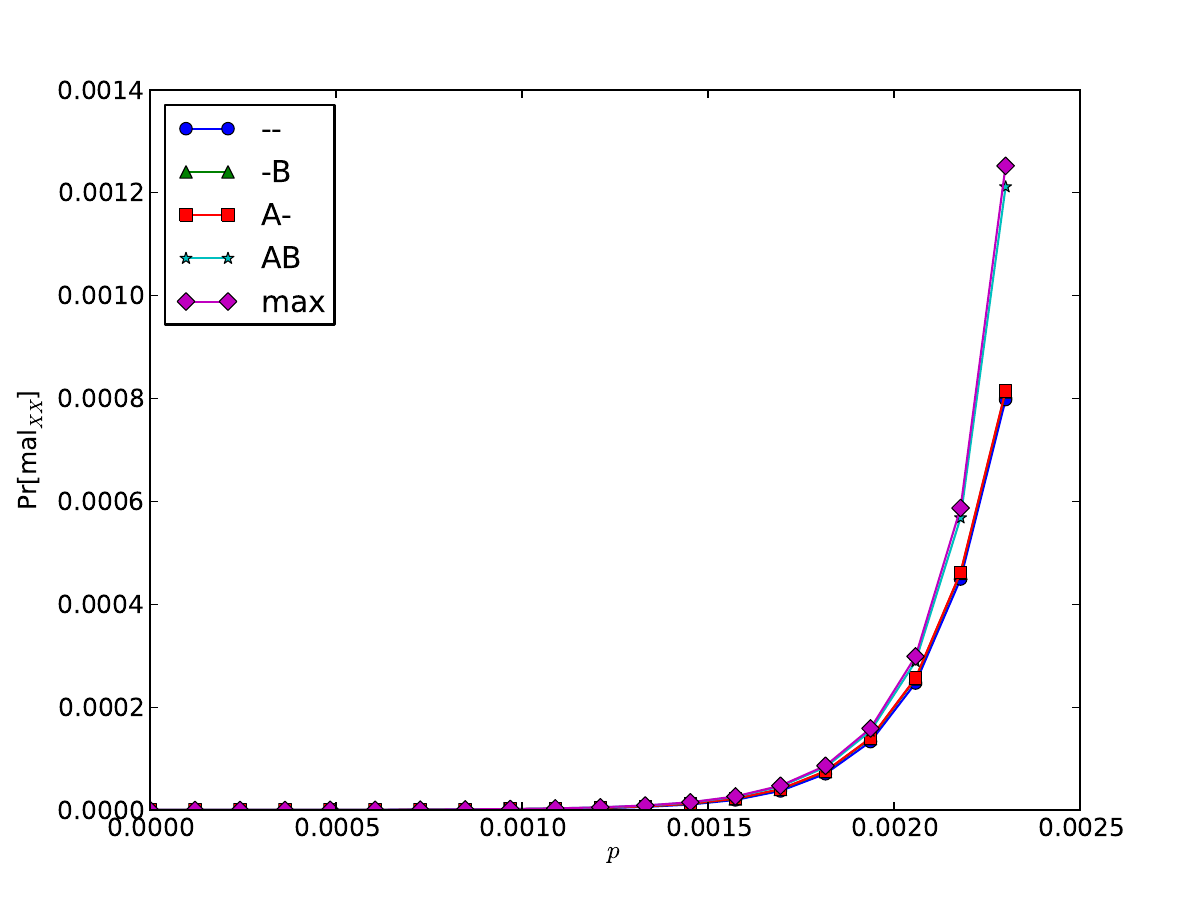}
}
\subfigure{
\includegraphics[width=6.79cm]{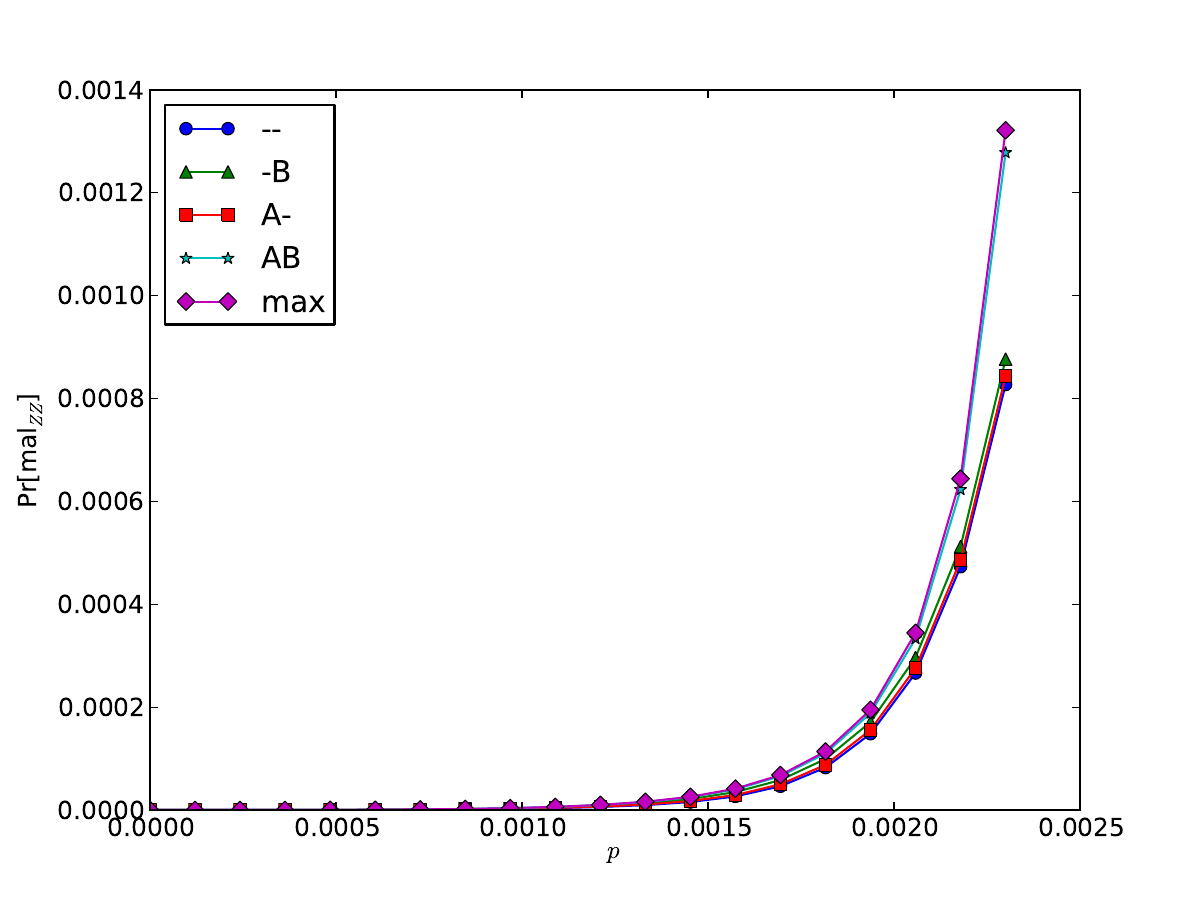}
}
\caption{The above plots show upper bounds on the probability of malignant events for the level-one CNOT exRec where the error corrections are based on Overlap-$4$ verification schedule. Upper bounds are plotted separately for each of four different CNOT exRecs: the full exRec (labeled AB), and the three incomplete exRecs in which the TEC on block~A---the control block---has been removed (-B), the TEC on block B---the target block---has been removed (A-) or the TEC on both blocks have been removed~(\texttt{--}).
Also shown is the polynomial ``max'' used to simultaneously upper bound all four possibilities (see \secref{sec:asymptotic-threshold}).}
\label{fig:cnot-malig-events}
\end{figure}

\subsection{Bounding malignant event polynomials} \label{sec:bounding-polys}

Construction of the transformed noise model requires bounding of several sets of polynomials in two different ways.  The first case compares polynomials of a fixed malignant event for each possible partial CNOT exRec (see \figref{fig:cnot-malig-events}).  In this case, we require only a single polynomial $\mathcal{P}^*$ which is strictly greater than or equal to all other polynomials in the set $P := \{\mathcal{P}_E \}$ over the interval $[0, \gammaMax]$.  The monotonicity of each of the polynomials (see \secref{sec:monotonicity}) over this interval means that constructing a $\mathcal{P}^*$ is relatively simple. First, choose some reasonably small $\gammaMin > 0$ and fix some $\Delta > 0$. Then a sufficient condition for $\mathcal{P}^*$ to be greater than all polynomials in $P$ over the interval $[\gammaMin, \gammaMax]$ is
\begin{equation}
  \mathcal{P}^*(\gammaMin + n \Delta) \geq \mathcal{P}_E(\gammaMin + (n+1) \Delta)
  \label{eq:maximality-condition}
\end{equation}
for all $\mathcal{P}_E \in P$ and all integers $0 \leq n \leq \lceil (\gammaMax-\gammaMin) / \Delta \rceil$.  $\mathcal{P}^*$ can be constructed by taking the $\mathcal{P}_E$ with the largest value at $\gammaMax$ (say), and adding a constant offset of at least $\max_E \mathcal{P}_E(\gammaMin)$ so that \eqnref{eq:maximality-condition} is satisfied.  Maximality over the remaining interval $[0,\gammaMin]$ follows by monotonicity.

In the second case, we compare malignant events from different types of exRecs.   We need to construct $\Gamma$ and determine values $\alpha_E$ for which the upper bound $\mathcal{P}_E \leq \alpha_E \Gamma$ in \eqnref{eq:Gamma-definition} is satisfied.  Construction of $\Gamma$ is similar to that of $\mathcal{P}^*$ from above. Let $\mathcal{P}_j$ be the polynomial with the largest derivative at $\gammaMax$.  Take $\mathcal{P}_j$ and divide by some appropriately large value $r$. Then add a constant offset $\epsilon := \max_{i\neq j} \mathcal{P}_i(\gammaMin)$ so that $\Gamma = \mathcal{P}_j/r + \epsilon$.  Finally, for each $\malig_E$, find the minimum value of $\alpha_E$ such that $\mathcal{P}^* := \alpha_E \Gamma$ satisfies condition \eqnref{eq:maximality-condition}.

In practice, the quality of the resulting bounds depends on the choice of $\gammaMin$ and the number of plotted points $n$.  We find that a value of $\gammaMin = \gammaMax/10$ or $\gammaMin=\gammaMax/100$, and $n=1000$ works well.  More sophisticated methods can also be used.  For example, the value of $\Delta$ could vary over the interval to better capture exponential behavior of the polynomials.

\subsection{Proof of \texorpdfstring{\claimref{clm:P2-multiplicative}}{Claim~\ref{clm:P2-multiplicative}}}\label{sec:P2-multiplicative-proof}

We conclude our analysis of the transformed error model by proving \claimref{clm:P2-multiplicative} that the level-two malignant event upper bounds decrease exponentially with $\gamma$.

\proof{
From \appref{sec:counting-detail} and \appref{sec:transformed-noise-model-construction} we see that $\mathcal{P}^{(2)}_E$ is expressed as
  \begin{equation}
    \frac{\Pr[\malig_E,\good]}{\Pr[\accept]} + \Pr[\bad \vert \accept]
    \enspace .
  \end{equation}
  The $\Pr[\malig_E,\good]$ term is expressed as a sum of the form
  \begin{equation}
    \sum_{k=0}^{\kMax} c(k) \Gamma^k
    \label{eq:pr-malig-form}
  \end{equation}
  where all of the coefficients $c(k)$ are non-negative (because there are no $XZ$ corrections) and it is understood that $\Gamma$ is a function of $\gamma$.  The $\Pr[\accept]$ term in the denominator is a product of terms of the form
  \begin{equation}
    1 - \sum_{k=0}^{\kMax} c(k) \Gamma^k
    \label{eq:pr-accept-form}
  \end{equation}
  where, again, all $c(k)$ are non-negative.  $\Pr[\bad \vert \accept]$ is a sum of terms similar to \eqnref{eq:pr-malig-form}, some of which contain \eqnref{eq:pr-accept-form} terms in the denominator. 
  
  Due to the strict fault-tolerance of our circuits, the coefficients $c(k)$ of \eqnref{eq:pr-malig-form} and the numerator coefficients of $\Pr[\bad \vert \accept]$ are zero for $k \leq 3$. Therefore, for $0\leq \epsilon \leq 1$, $\mathcal{P}^{(2)}_E(\epsilon \Gamma)$ is a sum of non-negative terms of the form
  \begin{align}
    \frac{\sum_{k=0}^{\kMax} c(k) (\epsilon \Gamma)^k}{1 - \sum_{k=0}^{\kMax} c(k) (\epsilon \Gamma)^k}
    \leq \frac{\epsilon^4 \sum_{k=4}^{\kMax} c(k) \Gamma^k}{1 - \sum_{k=0}^{\kMax} c(k) \Gamma^k}
  \end{align}
  which completes the proof.
}

\end{document}

%% file: header.tex
\usepackage{fullpage}
\usepackage{amsfonts, amssymb, amsmath, amsthm}
\usepackage{latexsym}
\usepackage[tracking=smallcaps]{microtype}	
\usepackage{url}
\usepackage{color}
\definecolor{DarkGray}{rgb}{0.1,0.1,0.5}
\usepackage{graphicx}
\usepackage[tight, TABBOTCAP]{subfigure}
\usepackage[small]{caption}	

\usepackage{import}
\usepackage{booktabs}
\usepackage{rotating}	

\usepackage{calc}	
\usepackage[colorlinks=true,breaklinks, linkcolor= DarkGray,citecolor= DarkGray,urlcolor= DarkGray]{hyperref}	

\newcommand{\ket}[1]{{|#1\rangle}}

\newcommand{\lket}[1]{{\ket{\overline{#1}}}}
\newcommand{\abs}[1]{{\lvert #1\rvert}}	

\newcommand{\binomial}[2]{\ensuremath{\left(\begin{smallmatrix}#1 \\ #2 \end{smallmatrix}\right)}}

\def\N {{\bf N}}

\DeclareMathOperator{\poly}{\operatorname{poly}}

\newcounter{sprows}

\newlength{\spheight}
\newlength{\spraise}

\newlength{\commentslength}

\newcommand{\rem}[1]{}

\newtheorem{theorem}{Theorem}[section]
\newtheorem{lemma}[theorem]{Lemma}

\newtheorem{claim}[theorem]{Claim}

\newtheorem{definition}[theorem]{Definition}

\newfont{\subsubsecfnt}{ptmri8t at 10pt}
\renewcommand{\subparagraph}[1]{\smallskip{\subsubsecfnt #1.}}

\numberwithin{equation}{section} 

\newcommand{\eqnref}[1]{\hyperref[#1]{{(\ref*{#1})}}}
\newcommand{\thmref}[1]{\hyperref[#1]{{Theorem~\ref*{#1}}}}
\newcommand{\lemref}[1]{\hyperref[#1]{{Lemma~\ref*{#1}}}}
\newcommand{\corref}[1]{\hyperref[#1]{{Corollary~\ref*{#1}}}}
\newcommand{\defref}[1]{\hyperref[#1]{{Definition~\ref*{#1}}}}
\newcommand{\secref}[1]{\hyperref[#1]{{Section~\ref*{#1}}}}
\newcommand{\figref}[1]{\hyperref[#1]{{Figure~\ref*{#1}}}}
\newcommand{\tabref}[1]{\hyperref[#1]{{Table~\ref*{#1}}}}
\newcommand{\remref}[1]{\hyperref[#1]{{Remark~\ref*{#1}}}}
\newcommand{\appref}[1]{\hyperref[#1]{{Appendix~\ref*{#1}}}}
\newcommand{\claimref}[1]{\hyperref[#1]{{Claim~\ref*{#1}}}}
\newcommand{\propref}[1]{\hyperref[#1]{{Proposition~\ref*{#1}}}}
\newcommand{\exampleref}[1]{\hyperref[#1]{{Example~\ref*{#1}}}}
\newcommand{\conjref}[1]{\hyperref[#1]{{Conjecture~\ref*{#1}}}}

\allowdisplaybreaks[1]